\documentclass[journal,twoside]{IEEEtran}

\usepackage{cite}

\ifCLASSINFOpdf
\else
\fi
\usepackage{amsmath,amssymb,amsfonts,amsthm}
\usepackage{algorithmic}
\usepackage{graphicx}
\usepackage{textcomp}
\usepackage{xcolor}
\usepackage{abbrevs}
\usepackage[normalem]{ulem}
\usepackage{xurl}
\usepackage[caption=false]{subfig}

\allowdisplaybreaks

\newtheorem{corollary}{Corollary}
\newtheorem{theorem}{Theorem}

\newtheorem{problemformulation}{Problem Formulation}

\newcommand{\dif}{\mathrm{d}}

\newcommand{\lambdau}{\lambda^\mathrm{u}}

\newcommand{\lfa}{\lambda^\mathrm{FA}}

\newcommand{\traj}{X} 
\newcommand{\stseq}{x} 
\newcommand{\tb}{\beta} 
\newcommand{\td}{\varepsilon} 
\newcommand{\tlen}{\ell} 
\newcommand{\targetStateSpace}{\mathcal{X}} 

\newcommand{\timeseq}[2]{#1:#2}
\newcommand{\timeset}[2]{\mathbb{N}_{#1}^{#2}}
\newcommand{\trajStateSpace}[1]{\mathcal{T}_{#1}} 
\newcommand{\existencespace}[1]{I_{#1}}

\newcommand{\measStateSpace}{\mathcal{Z}} 
\newcommand{\trackTable}{\mathbb{T}}

\newcommand{\conv}[2]{\left\langle #1 ; #2 \right\rangle}

\newcommand{\infvec}{y}
\newcommand{\infmat}{Y}

\newcommand{\trajdensityparams}{\theta}

\newcommand{\thepapertitle}{Poisson Multi-Bernoulli Mixtures\\ for Sets of Trajectories}

\begin{document}
%
\title{\thepapertitle}
%
%
%

\author{Karl Granstr\"om,~\IEEEmembership{Member,~IEEE,}
        Lennart Svensson,~\IEEEmembership{Senior Member,~IEEE,}
        Yuxuan Xia,\\
        Jason Williams,~\IEEEmembership{Senior Member,~IEEE,}
        and~\'Angel F. Garc\'ia-Fern\'andez
\thanks{K. Granstr\"om is with Zoox, San Francisco, CA, USA. He was previously with the Department of Electrical Engineering, Chalmers University of Technology, Gothenburg, Sweden. E-mail: \texttt{kagranstrom@gmail.com}.} 
\thanks{L. Svensson is with the Department of Electrical Engineering, Chalmers university of Technology, Gothenburg, Sweden. E-mail: \texttt{lennart.svensson@chalmers.se}.}
\thanks{Y. Xia is with the Department of Automation, Shanghai Jiaotong University, Shanghai, China. Email: \texttt{yuxuan.xia@sjtu.edu.cn}. }
\thanks{J. Williams is with Whipbird Signals, Australia. E-mail: \texttt{jlw@ieee.org}.}
\thanks{A. F. Garc\'ia-Fern\'andez is with the Department of Electrical Engineering and Electronics, University of Liverpool, Liverpool, UK.  E-mail: \texttt{angel.f.garcia.fernandez@gmail.com}.}}

\maketitle

\begin{abstract}

The Poisson Multi-Bernoulli Mixture (\pmbm) density is a conjugate multi-target density for the standard point target model with Poisson point process birth. This means that both the filtering and predicted densities for the set of targets are \pmbm. In this paper, we first show that the \pmbm density is also conjugate for sets of trajectories with the standard point target measurement model. Second, based on this theoretical foundation, we develop two trajectory \pmbm filters that provide recursions to calculate the posterior density for the set of all trajectories that have ever been present in the surveillance area, and the posterior density of the set of trajectories present at the current time step in the surveillance area. These two filters therefore provide complete probabilistic information on the considered trajectories enabling optimal trajectory estimation. Third, we establish that the density of the set of trajectories in any time window, given the measurements in a possibly different time window, is also a \pmbm. Finally, the trajectory \pmbm filters are evaluated via simulations, and are shown to yield state-of-the-art performance compared to other multi-target tracking algorithms based on random finite sets and multiple hypothesis tracking.

\end{abstract}


%
\IEEEpeerreviewmaketitle


\section{Introduction}
\label{sec:Introduction}

Multi-target tracking (\mtt) refers to the processing of sensor measurements to estimate the unknown and time-varying number of targets and their trajectories \cite{BarShalomWT:2011}. In this paper, we consider discrete-time multi-target tracking in which measurements at taken at discrete time steps, and each trajectory is characterised by a start time, and a sequence of target states. An alternative definition of a trajectory is to consider it a function of time, which can then be estimated using optimisation techniques \cite{Milan16,Li19d,Zhou21}. This approach has the benefit of estimating continuous-time trajectories, but it does not provide associated uncertainties.

This paper uses the Bayesian framework of random finite sets (\rfs) of trajectories to address the \mtt problem \cite{GarciaFernandezSM:2019}. Within this framework, we are given probabilistic models for multi-target birth, dynamics and death, and also models for the sensor measurements, and the goal is to compute or approximate the posterior density over the set of trajectories, as it is the density that enables us to answer all trajectory-related questions we may have, such as the estimation of the set of trajectories, and their associated uncertainties. 

This paper focuses on the sequential calculation of two important densities: the density for the set of all trajectories that have ever been present in the surveillance area, and the density for the set of trajectories present in the surveillance area at the current time step. We show that, in both cases, for the standard point target measurement model and standard dynamic model with Poisson point process (\ppp) birth, the conjugate density of the set of trajectories is a Poisson Multi-Bernoulli Mixture (\pmbm), which was first introduced for sets of targets in \cite{Wil12}. The resulting algorithms are referred to as trajectory \pmbm (\tpmbm) filters.

Solutions to the \mtt problem presuppose the ability to solve single target Bayesian filtering. Let $x_{k}\in\mathbb{R}^{n_x}$ and  $z_{k}\in\mathbb{R}^{n_z}$ denote the single target state and the measurement state at time step $k$. In single target Bayesian estimation, linear and Gaussian models are important because they yield closed-form solutions to Bayesian filtering, provided by the Kalman filter, and smoothing, provided by the Rauch-Tung-Striebel (\textsc{rts}) smoother or the two-filter smoother \cite{Sarkka:2013}. Specifically, the density of the state $x_{k}$ conditioned on the sequence $z_{\timeseq{1}{k'}}=(z_1,...,z_{k'})$ is Gaussian \cite{Sarkka:2013}
\begin{align}
p(x_k | z_{\timeseq{1}{k'}}) = \Npdfbig{x_k}{m_{k | k'}}{P_{k | k'}}
\label{eq:Gaussian_state}
\end{align}
where $m_{k | k'}$ and $P_{k | k'}$ are the mean and covariance matrix, and $k>k'$ for prediction, $k=k'$ for filtering and $k<k'$ for smoothing. 

Even more generally, if $x_{\timeseq{\alpha}{\gamma}}$ and $z_{\timeseq{\xi}{\chi}}$ denote the concatenation of target state and measurements in any (finite) time intervals $\timeseq{\alpha}{\gamma}$ and $\timeseq{\xi}{\chi}$, it holds that \cite{Sarkka:2013,Koch11}
\begin{align}
p(\stseq_{\timeseq{\alpha}{\gamma}}| z_{\timeseq{\xi}{\chi}}) = \Npdfbig{\stseq_{\timeseq{\alpha}{\gamma}}}{m_{\timeseq{\alpha}{\gamma}|\timeseq{\xi}{\chi}}}{P_{\timeseq{\alpha}{\gamma}|\timeseq{\xi}{\chi}}}
\label{eq:Gaussian_state_sequence}
\end{align}
where ${m_{\timeseq{\alpha}{\gamma}|\timeseq{\xi}{\chi}}}$ and ${P_{\timeseq{\alpha}{\gamma}|\timeseq{\xi}{\chi}}}$ are the corresponding mean and covariance matrix. That is, regardless of the time intervals we consider in the state or in the measurement, the posterior density is Gaussian. 

This paper is an extension of the conference paper \cite{GranstromSXGFW:2018}. In this paper, in addition to deriving the \tpmbm filters, we also show that a result similar to \eqref{eq:Gaussian_state_sequence} holds for \mtt with the standard point target models. Specifically, the density of the set of trajectories in some time interval $\timeseq{\alpha}{\gamma}$ given the measurements in any time interval $\timeseq{\xi}{\chi}$ is \pmbm. Consequently, for prediction, filtering and smoothing, the set of trajectories in a given time interval is \pmbm distributed. A special case of this result is obtained if we consider the density of the current set of targets and the sequence of measurements up to the current time step, which is a \pmbm and can be calculated by the \pmbm filtering recursion \cite{Wil12}. 

To summarise the contributions in this paper, for the standard point target model, we
\begin{enumerate}
	\item extend \pmbm conjugacy to sets of trajectories and derive the \tpmbm filters for the set of all trajectories, and the set of alive trajectories,
	\item present computationally efficient algorithms for computing the \pmbm posterior,
	\item show the time marginalisation theorem for \pmbm densities. This theorem states that the posterior density for the set of trajectories in any time window whose trajectories are alive at some point in a time sub-interval is a \pmbm. In addition, the associated \pmbm parameters can be computed from the posterior \pmbm on the set of all trajectories.
	\item show that the posterior density for the set of trajectories in any time window, given the measurements in a possibly different time window, is a \pmbm.
	\item derive relevant properties of the probability mass functions (pmfs) of the start time step and of the end time step for each Bernoulli trajectory.
\end{enumerate}

The rest of the paper is organised as follows. Related work is discussed in the next section. In Section~\ref{sec:problem_formulation}, we present the two considered problem formulations and the background theory on sets of trajectories. The prediction and update steps of the \tpmbm filter are presented in Section~\ref{sec:pmbm_trackers}. In Section~\ref{sec:pmbm_relations}, we show the time marginalisation theorem and that the distribution of the set of trajectories over any time window is \pmbm. Linear and Gaussian implementations of the \tpmbm filters are presented in Section~\ref{sec:linear_gaussian_implementation}, and Section~\ref{sec:simulation_study} analyses the performance of \tpmbm filters via simulations. The paper is concluded in Section~\ref{sec:conclusion}.

\section{Related work}

Approaches to the \mtt problem include the joint probabilistic data association filter (\jpda) \cite{BarShalomWT:2011}, multiple hypothesis tracking (\mht) \cite{Rei79}, and filters based on \rfs \cite{Mahler:2014}. Links between \rfs, \jpda, and \mht have been established in \cite{Wil12,Brekke18, Brekke21}. 

\mht refers to a collection of algorithms that address the data association problem in \mtt by forming hypotheses across multiple scans \cite{Chong19}. In the original \mht \cite{Rei79}, each measurement could potentially be the first detection of a new target, and the number of newly detected targets was given a Poisson distribution. \mht can be implemented by forming tracks for each target and assigning them a score, which is then used for track confirmation and deletion, and to determine the score of a data association hypothesis \cite[Chap. 16]{BlackmanP:1999} \cite{werthmann1992step}. In \cite{MorCho86}, the \mht model was made rigorous using a multi-target state-space which allows for an unknown number of targets, under the assumption that the number of targets is unknown but constant, with a Poisson prior. Target state sequences were formed under each global hypothesis, and the Poisson distribution of targets remaining to be detected provided a Bayes prior for events involving newly detected targets. Hypotheses were constructed as being data-to-data, since no a priori data was assumed on target identity. This formulation was revisited in \cite{MorCho16}, and compared to random finite sets, and finite point processes. \mht was extended in \cite{CorCar14}, to include a time-varying number of targets by using a target birth distribution in the measurement space and assuming stationary object dynamics. 

The Bayesian multi-target filter, using a \ppp birth model and the standard point target models, was derived in \cite{Wil12} resulting in a \pmbm posterior. This \pmbm provides all information on the current set of targets, but not their trajectories. The \pmbm posterior represents the union of two independent \rfs processes: a \ppp that has information on targets that have not been detected yet, and a multi-Bernoulli mixture (\textsc{mbm}) that has information on the potential targets that have been detected at some time up to now. Information on potentially undetected targets is important, for example, in search-and-track sensor management applications so that sensor resources can be directed to the regions of the surveillance area that may contain undetected targets  \cite{Bostrom-Rost21,Bostrom-Rost22}. The \mbm contains information on all data-to-data global hypotheses, their weights, and a multi-Bernoulli (\mb) density for each global hypothesis with information on the detected targets that may be present. Therefore, the data-to-data global hypotheses structure of \pmbm resembles the \mht approach. Nevertheless, to our knowledge, no \mht algorithm considers a \mb density for each global hypothesis, which results in computational benefits. Therefore, the \pmbm filter \cite{Wil12,GarciaFernandezWGS:2018} can be regarded as a state-of-the-art Bayesian \mht-type multi-target filter. Explicit links between \pmbm and \mht are provided in \cite{Wil12,Brekke18}. 

Three important advantages of the \pmbm filter are its efficient representation of the global hypotheses \cite{GarciaFernandezWGS:2018}, that it maintains information on undetected targets, and that it is the result of the Bayesian filtering recursion. The \pmbm filter has been applied to solve \mtt problems with different types of data such as radar \cite{CamentAC:2018}, lidar \cite{CamentAC:2018,GranstromSRXF:2018,pang2021multi}, audio \cite{zhao2022audio} and images \cite{ScheideggerG:2018,CamentAC:2018}. The \pmbm filter has also been successfully applied to mapping of stationary objects \cite{FatemiGSRH:2016_PMBradarmapping}, simultaneous localisation and mapping \cite{Ge22_prov,Kim2022PMBM}, and joint tracking and sensor localisation \cite{FrohleLGW:multisensorPMB}.

If the birth model is an \mb rather than a \ppp, the posterior of the current set of targets becomes an \mbm with the same recursion as the \pmbm filter with the difference that the birth Bernoulli components are added to the \mbm in the prediction step \cite{GarciaFernandezWGS:2018}. That is, tracks in the \mbm filter are initiated using a priori birth information, while in the \pmbm filter, track initiation is measurement-driven, which improves performance \cite{GarciaFernandezXGSW:2019}. After each \mbm filter prediction step, we can also represent the \mbm with global hypotheses that consider targets that either exist or not (deterministic target existence). This approach has the drawback of requiring an exponential increase in the number of global hypotheses, giving rise to the \mbmo filter \cite{GarciaFernandezWGS:2018}. Deterministic target existence for each global hypothesis is also the approach in standard \mht algorithms \cite{BlackmanP:1999}.

In \mtt, apart from having information on the set of targets, it is important to represent information on the target trajectories \cite{Koch11}. One approach to estimate target trajectories sequentially is to augment the target states with labels, which are unique for each potential target and not allowed to change in time. Labels have been used in non-Bayesian \mtt algorithms in computer vision to design an optimisation problem to estimate trajectories \cite{Milan16}. In a Bayesian setting with labelled targets, we can calculate the corresponding posterior, estimate the labelled targets at each time step, and estimate trajectories by connecting target state estimates with the same label \cite{VoV:2013, VoVH:2017, GarciaFernandezGM:2013,AokiMSBB:2016}. This approach gives rise to \mtt algorithms based on labelled \rfs if the multi-target state is a set \cite{VoV:2013,VoVH:2017}. There are also related approaches for sequential track formation such as the use of \pmbm filter metadata (auxiliary variables) \cite{Angel20_e} or the methods in  \cite{PantaCV2009,HoussineauC:2018}. Labelled versions of the \mbm and \mbmo filters can be obtained by using an \mb birth model where each Bernoulli of the birth process gives rise to a unique and fixed label \cite{VoVH:2017}. The posterior can then be calculated using the equations of the \mbm or the \mbmo filters \cite{GarciaFernandezXGSW:2019}. The $\delta$-generalised labelled \mb (\dglmb) recursion \cite{VoVH:2017} is equivalent to the \mbmo recursion, which means that the \dglmb filter makes use of an less efficient representation of the global hypotheses than \mbm and \pmbm filters \cite[Sec. IV]{GarciaFernandezWGS:2018}. 

The \dglmb filter often estimates trajectories well. However, \dglmb filter trajectory estimation has some weaknesses. For example, the \dglmb filter sequential trajectory estimator can produce trajectory estimates with holes even though the existence of holes in the trajectories is not a possible event according to the standard dynamic model, in which targets may die but are not allowed to resurrect. In addition, the \dglmb filter has problems estimating trajectories of multiple targets born at the same time if the birth process is an independent and identically distributed cluster process covering a large area \cite[Example 2]{GarciaFernandezSM:2019}. A practical solution to solve this problem that has been used in large-scale problems \cite{Beard20} is to use an adaptive birth model \cite{ReuterVVD:2014}, in which each measurement gives rise to a birth Bernoulli, resembling how the \pmbm filter initiates Bernoulli components. However, using adaptive birth implies that the \dglmb filter is no longer a fully Bayesian filter. For example, target birth is no longer independent of past measurements, as required in \dglmb modelling assumptions, and an extra user-defined parameter is needed \cite{Angel22_b}.

One way to avoid the above issues with trajectory estimation is to estimate trajectories directly from the posterior on the sequence of sets of labelled targets \cite{Vu14} or the posterior on the set of trajectories \cite{GarciaFernandezSM:2019}. The benefits of using sets of trajectories instead of sequences of sets of labelled targets are thoroughly discussed in \cite{GarciaFernandezSM:2019}. The paper \cite{GarciaFernandezSM:2019} introduced the \mbmo filter for the set of all trajectories: the trajectory \mbmo (\tmbmo) filter. In the \tmbmo filter, the set of trajectories can also be labelled, without adding relevant information or changes to the \tmbmo recursion \cite[Sec. IV.A]{GarciaFernandezSM:2019}, and the marginal distribution on the current set of targets corresponds to the posterior \dglmb density \cite[Sec. IV.C]{GarciaFernandezSM:2019}. In addition, the \tmbmo filter has also been derived using the sequence of labelled sets of targets and \dglmb notation in \cite{VoV:2019}, and was referred to as multi-scan \dglmb. However, the \tmbmo (and the multi-scan \dglmb) filter still suffers from an inefficient global hypothesis representation, since hypotheses have deterministic target existence. The trajectory \mbm (\tmbm) filter \cite{XiaGSGFW:jaifMultiScanPMBMtrackers} improves the global hypothesis representation of the \tmbmo filter by enabling probabilistic target existence. 

In this paper, we show that the benefits of the global hypotheses representation of the \pmbm filter can be combined with sets of trajectories to derive trajectory \pmbm (\tpmbm) filters. These filters provide full information on the set of trajectories, including undetected ones, and make use of a more efficient representation of the global hypotheses than \tmbmo and \tmbm filters. The proposed \tpmbm filters are fully Bayesian methods that can be considered state-of-the-art \mht-type algorithms that provide full trajectory information on all targets.

Several subsequent papers\footnote{\matlab code of these algorithms is provided at \url{https://github.com/yuhsuansia} and \url{https://github.com/Agarciafernandez}.} have been proposed based on the derivations presented in this paper, through the preprint \cite{Granstrom19_prov2} or the preliminary conference version  \cite{GranstromSXGFW:2018}. For instance, the trajectory Poisson \mb (\tpmb) filter in \cite{Angel20_e} approximates the \tpmbm posterior at each update step as a \tpmb, which only has one mixture component, via a Kullback-Leibler divergence minimisation, see \cite[Figure 1]{Angel20_e}. The extended target \tpmb filter in \cite{xia2023trajectory} obtains the \tpmb approximation at each update step using set-type belief propagation. \tpmb filters are computationally faster than \tpmbm filters but they not provide a closed-form solution to the posterior. Reference \cite{XiaGSGFW:jaifMultiScanPMBMtrackers} proposed a multi-scan implementation of the  \tpmbm filter, as well as the \tmbm filtering recursion, which results from an \mb birth model instead of a \ppp birth model. When target spawning is considered, rather than sets of trajectories, it is more appropriate to consider sets of tree trajectories to keep full genealogy information. The resulting \tpmbm filter for sets of tree trajectories is proposed in \cite{Angel22}. A continuous-discrete \tpmbm filter that optimally processes out-of-sequence measurements is proposed in \cite{Angel21_d}.


Differently from the above-mentioned works, this paper provides the theoretical foundation for the \tpmbm filters, presents three possible single-trajectory filtering implementations, includes a thorough comparison with other state-of-the-art \rfs filters and also with an \mht implementation, derives pertinent properties of the pmfs of the start time step and of the end time step for each Bernoulli trajectory, derives the time marginalisation theorem for \pmbm densities on sets of trajectories, and proves the theorem on \tpmbm densities for arbitrary time intervals, which shows that the posterior of the set of trajectories on any time window given measurements in a possibly different time window is a \tpmbm. Paper \cite{Granstrom20b} is based this theorem, and shows that if we add spatio-temporal constraints to a \tpmbm density, the density remains \tpmbm.


\section{Problem formulations and background}
\label{sec:problem_formulation}
In this section, we first present the problem formulations. Then, we review the trajectory state representation, and present the densities for sets of trajectories. All details on how to define densities and integration for sets of trajectories are given in \cite{GarciaFernandezSM:2019} using finite set statistics \cite{Mahler:2014}, and in \cite{XiaGSGFW:jaifMultiScanPMBMtrackers} using measure theory.

\subsection{Problem formulations}
We use notation $x_{k}\in\targetStateSpace$ to denote a single target state at time step $k$, where $\targetStateSpace$ is the single target state space. Similarly, we use $z_{k}\in\measStateSpace$ to denote a single measurement at time step  $k$, where $\measStateSpace$ is the single measurement state space. 

\begin{assumption}
	The multi-target dynamic model is \cite{Mahler:2014}:
	\begin{itemize}
		\item Targets are born at each time step following an independent \ppp with intensity $\lambda^\mathrm{b}(x_k)$.
		\item A target with state $x_k$ survives to the next time step with probability of survival $P^{\rm S}(x_k)$ and single-target transition density $\pi^{x}(x_k|x_{k-1})$, independently of the rest of the targets.
	\end{itemize}
	\label{ass:Dynamics}%
\end{assumption}%
\begin{assumption}%
	The measurement model is \cite{Mahler:2014}:
	\begin{itemize}
		\item A target $x_k$ is detected with a probability  $P^{\rm D}(x_k)$, independently of the rest of the targets, and its associated single measurement has a density $\varphi^{z}(z_k|x_k)$.
		\item Clutter measurements are modelled by an independent \ppp with intensity $\lambda^\mathrm{FA}(z_k)$.
	\end{itemize}%
	\label{ass:Measurement}%
\end{assumption}%
Under these assumptions, we consider the following two problem formulations (\probform) of \mtt:
\begin{problemformulation}\label{pf:all_trajectories} 
 \emph{The set of all trajectories:} the aim is to compute the density of the set of trajectories of all targets that have been present in the surveillance area at some point up to the current time step. 
\end{problemformulation}%
\begin{problemformulation}%
\label{pf:current_trajectories} 
\emph{The set of current trajectories:} the aim is to compute the density of the set of trajectories of the targets that are currently present in the surveillance area.
\end{problemformulation}%

\subsection{Single trajectories}
\label{sec:SingleTrajectoryStateRepresentation}

We first present the notation for time intervals. The (ordered) sequence of consecutive time steps between time steps  $\alpha$ and $\gamma$ is $\timeseq{\alpha}{\gamma} = \left( \alpha,\alpha+1,\ldots,\gamma-1,\gamma \right)$. The corresponding (unordered) set of  time steps is $\timeset{\alpha}{\gamma} = \left\{\alpha,\alpha+1,\ldots,\gamma-1,\gamma\right\}.$

The single trajectory state is $\traj = \left(\tb,\td,\stseq_{\timeseq{\tb}{\td}}\right)$ where $\tb$ is the time step when the trajectory starts, $\td$ is the time step when the trajectory ends, $\stseq_{\timeseq{\tb}{\td}}=\left(x_{\tb},x_{\tb+1},\ldots,x_{\td-1},x_{\td}\right)$, where $ x_{k}\in\targetStateSpace,\ \forall k\in\timeset{\tb}{\td}$. The length in time steps of a trajectory $\traj$ is $\tlen = \td-\tb+1$. Note that it is also possible to define the trajectory state $\traj$ without including $\td$, as $\td$ is implicitly given in the length of $\stseq_{\timeseq{\tb}{\td}}$ \cite{GarciaFernandezSM:2019}.

The single trajectory space in the time interval $\timeseq{\alpha}{\gamma}$ is \cite{GarciaFernandezSM:2019}
\begin{align}
	& \trajStateSpace{\timeseq{\alpha}{\gamma}} = \uplus_{(\tb,\td)\in \existencespace{\timeseq{\alpha}{\gamma}}} \{\tb\}\times\{\td\}\times\targetStateSpace^{\td-\tb+1},
\end{align}
where $\uplus$ is the union of disjoint sets, $\existencespace{\timeseq{\alpha}{\gamma}} = \{ (\tb,\td) : \alpha\leq \tb \leq \td \leq \gamma \}$ and $\targetStateSpace^{\tlen}$ represents $\tlen$ Cartesian products of $\targetStateSpace$. This is a slight generalisation of the definition in \cite{GarciaFernandezSM:2019}, where $\trajStateSpace{k}$ is used to denote the trajectory state space for trajectories in $\timeseq{0}{k}$. The possible lengths of trajectories in $\trajStateSpace{\timeseq{\alpha}{\gamma}}$ are constrained to $1\leq\tlen\leq\gamma-\alpha+1$. It should be noted that state spaces that are the union of different, bounded size state spaces, as the trajectory space, have been used in interacting multiple models \cite{Mazor98} and unknown clutter models \cite{Mahler11}.

A density on the single trajectory space $\trajStateSpace{\timeseq{\alpha}{\gamma}}$ factorises as
\begin{align} 
	p(\traj) = p(\stseq_{\timeseq{\tb}{\td}} | \tb,\td) P(\tb,\td), \label{eq:trajectory_state_density}
\end{align}
where the domain of $P(\tb,\td)$ is $\existencespace{\timeseq{\alpha}{\gamma}}$. Integration is defined as \cite{GarciaFernandezSM:2019}

\begin{equation}\label{eq:trajectoryDensityIntegration}
	\int_{\trajStateSpace{\timeseq{\alpha}{\gamma}}} p(\traj) \diff \traj = \sum_{(\tb,\td)\in\existencespace{\timeseq{\alpha}{\gamma}}} \left[ \int_{\targetStateSpace^{\tlen}} p(\stseq_{\timeseq{\tb}{\td}} | \tb,\td) \diff \stseq_{\timeseq{\tb}{\td}} \right] P(\tb,\td).
\end{equation}

\subsection{Sets of trajectories}
The set of trajectories in the time interval ${\timeseq{\alpha}{\gamma}}$ is $\settraj_{\timeseq{\alpha}{\gamma}}$, such that  $\settraj_{\timeseq{\alpha}{\gamma}} \in \mathcal{F}(\trajStateSpace{\timeseq{\alpha}{\gamma}})$, the set of all finite subsets of $\trajStateSpace{\timeseq{\alpha}{\gamma}}$. The subset of the trajectories in $\settraj_{\timeseq{\alpha}{\gamma}}$ that were alive at some time in the time interval $\timeseq{\eta}{\zeta}$, where $\alpha\leq\eta\leq\zeta\leq\gamma$, is
\begin{align}\label{eq:subset_trajectories}
	\settraj_{\timeseq{\alpha}{\gamma}}^{\timeseq{\eta}{\zeta}} = \left\{ \left(\tb,\td,\stseq_{\timeseq{\tb}{\td}}\right) \in \settraj_{\timeseq{\alpha}{\gamma}} \ : \ \timeset{\tb}{\td}\cap \timeset{\eta}{\zeta} \neq\emptyset \right\}.
\end{align} 
For example, in \probform~\ref{pf:current_trajectories}, we consider \eqref{eq:subset_trajectories} with $\eta=\zeta=\gamma$. As in RFS theory for sets of targets \cite{Mahler:2014}, we assume that two different trajectories can never have the same state at the same time step \cite[Sec. II.A]{GarciaFernandezSM:2019}. This is  actually an event of probability zero in most settings of interest, that is, when at least one variable of the single-target space has a continuous distribution, e,g., one variable takes values in $\mathbb{R}$.

Let $g\left(\settraj_{\timeseq{\alpha}{\gamma}}\right)$ be a real-valued function on a set of trajectories $\settraj_{\timeseq{\alpha}{\gamma}}$. Its corresponding set integral is \cite{Mahler:2014, GarciaFernandezSM:2019}
\begin{multline}
\int{g\left(\settraj_{\timeseq{\alpha}{\gamma}}\right) \delta \settraj_{\timeseq{\alpha}{\gamma}}}
\triangleq \sum_{n=0}^{\infty}\frac{1}{n!}\int{g(\{\traj^1,\dots,\traj^n\})\dif \traj^{1:n}} \label{eq:SetTrajIntegral}
\end{multline}
where $\traj^{1:n}=(\traj^1,\dots,\traj^n)$.

A (multi-trajectory) density $f\left(\settraj_{\timeseq{\alpha}{\gamma}}\right)$, which represents the information on an \rfs of trajectories, is non-negative and integrates to one using \eqref{eq:SetTrajIntegral}. For instance, a trajectory \ppp has density
	\begin{align}
		f\left(\settraj_{\timeseq{\alpha}{\gamma}}\right) & = e^{-\int \lambda(\traj)\diff\traj} \lambda^{\settraj_{\timeseq{\alpha}{\gamma}}}, 
		\label{eq:PPPdensityDefinition}
	\end{align}
with intensity $\lambda(\traj)$, $\lambda^{\settraj_{\timeseq{\alpha}{\gamma}}}=\prod_{\traj\in\settraj_{\timeseq{\alpha}{\gamma}}} \lambda(\traj)$ if $\settraj_{\timeseq{\alpha}{\gamma}}\neq\emptyset$, and $\lambda^{\settraj_{\timeseq{\alpha}{\gamma}}}=1$ if $\settraj_{\timeseq{\alpha}{\gamma}} = \emptyset$. 

A trajectory Bernoulli process has a density
\begin{align}
f\left(\settraj_{\timeseq{\alpha}{\gamma}}\right) &= \begin{cases}
1-r, & \settraj_{\timeseq{\alpha}{\gamma}} = \emptyset , \\
r f(\traj), & \settraj_{\timeseq{\alpha}{\gamma}} = \{\traj\} , \\
0, & \mbox{otherwise} ,
\end{cases}
\label{eq:BernoulliDensityDefinition}%
\end{align}%
where $f(\traj)$ is a single trajectory density \eqref{eq:trajectory_state_density}, and $r$ is the probability of existence. In the same way, the densities of trajectory \mbm and \pmbm \rfs are natural extensions of their definitions for \rfs of targets.


\section{TPMBM filters}
\label{sec:pmbm_trackers}
In this section we present the prediction and update steps of the \tpmbm filters for the two problem formulations. For the set of all trajectories, we seek the density for the \rfs $\settraj_{\timeseq{0}{k}}$, in other words, all trajectories existing at some point in time between the initial time step $0$ to the current time step $k$. For the set of current trajectories, we seek the density for the \rfs	$\settraj_{\timeseq{0}{k}}^{k} \triangleq \settraj_{\timeseq{0}{k}}^{\timeseq{k}{k}}$, see  \eqref{eq:subset_trajectories}.

\subsection{\pmbm density}

As in \cite{Wil12}, we hypothesise that the set of trajectories density is a multi-target conjugate prior of the \pmbm form, and we show that the \pmbm form is preserved through prediction and update. The \ppp represents trajectories that are hypothesised to exist at some time in the considered time interval, but have not been detected. The \mbm contains information on the trajectories that have been detected at some point up to the current time step. Each \mb in the \mbm represents the information on the detected trajectories given a global data association hypothesis. 

Using the shorthand notation $\settraj_k \triangleq \settraj_{\timeseq{0}{k}}$, the \pmbm density is\footnote{For the sake for brevity, we use sub-scripts $_{k}$ and $_{k|k'}$ instead of sub-scripts $_{\timeseq{0}{k}}$ $_{\timeseq{0}{k}|\timeseq{0}{k'}}$ to denote a density for trajectories in $\trajStateSpace{\timeseq{0}{k}}$, conditioned on measurements in the interval $\timeseq{0}{k'}$.}
\begin{subequations}
\begin{align}
	& f_{k|k'}(\settraj_{k}) = \sum_{\settraj_{k}^{\rm u}\uplus\settraj_{k}^{\rm d}=\settraj_{k}} f_{k|k'}^{\rm u}(\settraj_{k}^{\rm u}) f_{k|k'}^{\rm d}(\settraj_{k}^{\rm d}) \label{eq:PMBMdensityDefinition1} \\
	& f_{k|k'}^{\rm d}(\settraj_{k}^{\rm d}) = \sum_{\assoc_{k|k'}\in\assocspace_{k|k'}} w_{k|k'}^{\assoc} \sum_{\uplus_{i\in\trackTable_{k|k'}} \settraj_{k}^{i} = \settraj_{k}^{\rm d}} \prod_{i\in\trackTable_{k|k'}} f_{k|k'}^{i,\assoc^i}(\settraj_{k}^i)
	\label{eq:MBMdensityDefinition}
\end{align}%
\label{eq:PMBMdensityDefinition}%
\end{subequations}%
where $\uplus$ is the union of disjoint sets, $f_{k|k'}^{\rm u}(\settraj_k^{\rm u})$ is a \ppp density \eqref{eq:PPPdensityDefinition} with intensity $\lambdau_{k|k'}(\traj)$, $f_{k|k'}^{i,\assoc^i}(\settraj_k^i)$ are Bernoulli densities \eqref{eq:BernoulliDensityDefinition} with probabilities of existence $r_{k|k'}^{i,\assoc^i}$ and trajectory densities $f_{k|k'}^{i,\assoc^i}(\traj)$. In the \mbm density \eqref{eq:MBMdensityDefinition}, ${\trackTable}_{k|k'}$ is a track table with $n_{k|k'}$ tracks indexed $1$ to $n_{k|k'}$, and $\assoc\in\assocspace_{k|k'}$ is a global data association hypothesis. For each global hypothesis $\assoc$, there is a track local hypothesis $\assoc^i$ for each track $i\in{\trackTable}_{k|k'}$. The number of local hypotheses for track $i$ is denoted as $h_{k|k'}^{i}$. The weight of the global hypothesis $\assoc$ is proportional to the product of the local hypothesis weights $w_{k|k'}^{i,\assoc^i}$  for each track such that $w_{k|k'}^{\assoc} \propto \prod_{i\in{\trackTable}_{k|k'}} w_{k|k'}^{i,\assoc^i}$.

Following \cite{Wil12}, the set of all measurement indices up to time $k$ is $\assocset^k$, and $\assocset^{k}(i,\assoc^i)$ is the subset of indices that are associated to track $i$ with local hypothesis $\assoc^i$. The set of global data association hypotheses $\assocspace_{k|k'}$ can be written in terms of $\assocset^k$, and $\assocset^{k}(i,\assoc^i)$, see \cite[Eq. 35]{Wil12}.

A \pmbm density \eqref{eq:PMBMdensityDefinition} is then characterised by
\begin{align}
	\lambdau_{k|k'}(\cdot), \left\{ \left( w_{k|k'}^{i,\assoc^i} , r_{k|k'}^{i,\assoc^i} , f_{k|k'}^{i,\assoc^i}(\cdot) \right)  \right\}_{\assoc \in\assocspace_{k|k'}, \ i\in\trackTable_{k|k'}} , \label{eq:PMBMparameters}
\end{align}
The following subsections show the prediction and update steps of the trajectory \pmbm density for the two considered problem formulations.

\subsection{Prediction step}
\label{sec:PMBMtrackerPrediction}
In this section we describe the time evolution of the set of trajectories under Assumption \ref{ass:Dynamics}. The birth \ppp in terms of trajectories is
\begin{align}
	\lambda_{k}^{{\rm B}}(\traj) &= \begin{cases} \lambda_{k}^{\mathrm{b}}(x_k) & (\tb,\td) = (k,k), \\ 0 & \text{otherwise.} \end{cases}
	\label{eq:PoissonBirthIntensity}%
\end{align}%

The trajectory state dependent probability of survival at time step $k$ is $P_{k}^{\rm S}(\traj) = P^{\rm S}(x_{\td})\Delta_{k}(\td)$, where $\Delta_{k}(\cdot)$ denotes Kronecker's delta function located at $k$. 

\subsubsection{Dynamic model for the set of all trajectories}
Given a set of trajectories $\settraj_{\timeseq{0}{k-1}}$ with at maximum one trajectory, the Bernoulli transition density without birth is
\begin{subequations}
\begin{align}
	& f_{k|k-1}^{\rm a}(\settraj_{\timeseq{0}{k}} | \settraj_{\timeseq{0}{k-1}}) = \\
	& \qquad \left\{ \begin{array}{lcl}
			1 & & \settraj_{\timeseq{0}{k-1}}=\emptyset, \settraj_{\timeseq{0}{k}}=\emptyset , \\
			\pi^{\rm a}(\traj|\traj') & \text{if} & \settraj_{\timeseq{0}{k-1}}=\{\traj'\},\settraj_{\timeseq{0}{k}}=\{\traj\} , \\
			0 & & \text{otherwise} ,
		\end{array} \right. \nonumber \\
	& \pi^{\rm a}(\traj | \traj') = \pi_{}^{{\rm a},\stseq}(\stseq_{\timeseq{\tb}{\td}} | \tb,\td,\traj') \pi^{\td}_{}(\td | \traj') \Delta_{\tb'}(\tb), \\
	& \pi^{\td}_{}(\td | \traj') = 
		\left\{\begin{array}{ll} 
			1, & \td = {\td}' < k-1 , \\
			1-P_{k-1}^{\rm S}(\traj'), & \td = {\td}' = k-1 ,\\ 
			P_{k-1}^{\rm S}(\traj'), & \td = {\td}'+1 = k , \\
			0, & \text{otherwise} ,
		\end{array} \right. \\
	&\pi^{{\rm a},\stseq}_{}(\stseq_{\timeseq{\tb}{\td}} | \tb,\td,\traj') = \\
	& \qquad 
		\left\{\begin{array}{ll} 
			\delta_{\stseq_{\tb':\td'}'}(\stseq_{\timeseq{\tb}{\td}}), & \td = \td' ,\\
			\pi^{x}_{}(x_{\td} | x'_{\td'}) \delta_{\stseq'_{\tb':\td'}}(\stseq_{\timeseq{\tb}{\td-1}}), & \td = \td'+1 . 
		\end{array}\right. \nonumber
\end{align}%
\label{eq:set_of_all_trajectories_state_transition_density}%
\end{subequations}%
where $\delta(\cdot)$ denotes the Dirac delta function. 
In this problem formulation, every trajectory that has ever existed remains in the set of all trajectories. The probability of survival $P^{\rm S}_{k-1}(\cdot)$ is used to determine whether the trajectory ends, or if it is extended by one more time step with a new state.

The prediction step is presented in the theorem below, where $\conv{f}{g} = \int f(\traj) g(\traj) \diff \traj$.

\begin{theorem}[\tpmbm prediction for all trajectories]\label{th:PredictionAllTrajectories}
	If the posterior density of the set of all trajectories at the previous time step is a \pmbm of the form \eqref{eq:PMBMdensityDefinition}, the predicted density of the set of all trajectories is a \pmbm of the form \eqref{eq:PMBMdensityDefinition} with:
	\begin{subequations}
	\begin{align}
		\lambdau_{k|k-1}(\traj) &= \lambda_{k}^{\rm B}(\traj) + \conv{\lambdau_{k-1|k-1}}{\pi^{\rm a}} , \label{eq:allUndetectedPropagation} \\
		n_{k|k-1} &= n_{k-1|k-1} , \\
		h^i_{k|k-1} & = h^i_{k-1|k-1} \; \forall \; i , \\
		w_{k|k-1}^{i,\assoc^i} &= w_{k-1|k-1}^{i,\assoc^i} \; \forall \; i, \assoc^i , \\
		r_{k|k-1}^{i,\assoc^i} & = r_{k-1|k-1}^{i,\assoc^i}, \; \forall \; i, \assoc^i , \label{eq:allExistProbPropagation} \\
		f_{k|k-1}^{i,\assoc^i} & = \conv{f_{k-1|k-1}^{i,\assoc^i}}{\pi^{\rm a}} , \; \forall \; i, \assoc^i .\label{eq:allKinematicDistPropagation}
	\end{align}%
	\label{eq:PredictionAllTrajectories}%
	\end{subequations}%
\end{theorem}%
\begin{proof}
	 Analogous to the \pmbm prediction proof \cite[Thm. 1]{Wil12} \cite{GarciaFernandezWGS:2018}, using the transition density \eqref{eq:set_of_all_trajectories_state_transition_density} and the \ppp birth intensity \eqref{eq:PoissonBirthIntensity}.
\end{proof}

\subsubsection{Dynamic model for the set of current trajectories}
Given a set of trajectories $\settraj_{\timeseq{0}{k-1}}^{k-1}$ with at maximum one trajectory that is present at time step $k-1$, the Bernoulli transition density without birth is
{
\begin{subequations}
\begin{align}
	& f_{k|k-1}^{\rm c}(\settraj_{\timeseq{0}{k}}^{k} | \settraj_{\timeseq{0}{k-1}}^{k-1}) = \label{eq:set_of_current_trajectories_transition_density}\\
	& \quad \left\{ \begin{array}{ll}
			1, & \settraj_{\timeseq{0}{k-1}}^{k-1}=\emptyset, \settraj_{\timeseq{0}{k}}^{k}=\emptyset , \\
			1-P_{k-1}^{\rm S}(\traj'), & \settraj_{\timeseq{0}{k-1}}^{k-1}=\{\traj'\},\settraj_{\timeseq{0}{k}}^{k}=\emptyset , \\
			P_{k-1}^{\rm S}(\traj')\pi^{\rm c}(\traj|\traj'), & \settraj_{\timeseq{0}{k-1}}^{k-1}=\{\traj'\},\settraj_{\timeseq{0}{k}}^{k}=\{\traj\} , \\
			0, & \text{otherwise} ,
		\end{array} \right. \nonumber \\
	& \pi^{\rm c}(\traj | \traj') =  \pi_{}^{{\rm c},\stseq}(\stseq_{\timeseq{\tb}{\td}} | \tb,\td,\traj')   \Delta_{\td'+1}(\td)  \Delta_{\tb'}(\tb), \\
	& \pi^{{\rm c},\stseq}_{}(\stseq_{\timeseq{\tb}{\td}} | \tb,\td,\traj') = \pi^{x}_{}(x_{\td} | x'_{\td'}) \delta_{\stseq'_{\tb':\td'}}(\stseq_{\timeseq{\tb}{\td-1}}) .
\end{align}%
\label{eq:set_of_current_trajectories_state_transition_density}%
\end{subequations}%
}%
In this problem formulation, if a trajectory dies, it will no longer be a element of the set of current trajectories. If the trajectory survives, then it is kept in the set of current trajectories and is extended by one time step. For the set of current trajectories, the probability of survival $P^{\rm S}_{k-1}(\cdot)$ is therefore used as in tracking a set of targets \cite{Mahler:2014}. 

The prediction step is given in the following theorem.
\begin{theorem}[\tpmbm prediction for current trajectories]\label{th:PredictionCurrentTrajectories}
If the posterior density of the set of current trajectories at the previous time step is a \pmbm of the form \eqref{eq:PMBMdensityDefinition}, the predicted density of the set of current trajectories is a \pmbm of the form \eqref{eq:PMBMdensityDefinition} with:
\begin{subequations}
	\begin{align}
		\lambdau_{k|k-1}(\traj) &= \lambda_{k}^{\rm B}(\traj) + \conv{\lambdau_{k-1|k-1}}{\pi^{\rm c}P_{k-1}^{\rm S}} ,\label{eq:UndetectedPropagation} \\
		n_{k|k-1} &= n_{k-1|k-1} , \label{eq:PredictedNumberOfTracks}\\
		h^i_{k|k-1} & = h^i_{k-1|k-1} \; \forall \; i  , \label{eq:PredictedNumberOfHypotheses} \\
		w_{k|k-1}^{i,\assoc^i} &= w_{k-1|k-1}^{i,\assoc^i} \; \forall \; i, \assoc^i , \label{eq:PredictedAssociationWeight}\\
		r_{k|k-1}^{i,\assoc^i} & = r_{k-1|k-1}^{i,\assoc^i}\conv{f_{k-1|k-1}^{i,\assoc^i}}{P_{k-1}^{\rm S}}, \; \forall \; i, \assoc^i , \label{eq:ExistProbPropagation} \\
		f_{k|k-1}^{i,\assoc^i} & = \frac{\conv{f_{k-1|k-1}^{i,\assoc^i}}{\pi^{\rm c}P_{k-1}^{\rm S}}}{\conv{f_{k-1|k-1}^{i,\assoc^i}}{P_{k-1}^{\rm S}}}, \; \forall \; i, \assoc^i . \label{eq:KinematicDistPropagation}
	\end{align}%
\label{eq:PredictionCurrentTrajectories}%
\end{subequations}%
\end{theorem}%
\begin{proof}
	Analogous to the \pmbm prediction proof \cite[Thm. 1]{Wil12} \cite{GarciaFernandezWGS:2018}, using the transition density \eqref{eq:set_of_current_trajectories_state_transition_density} and the \ppp birth intensity \eqref{eq:PoissonBirthIntensity}.
\end{proof}

\subsection{Update step}
\label{sec:PMBMtrackerUpdate}
The measurement model is the same regardless of the problem formulation. Given a set of trajectories $\settraj$ with at most one trajectory, the Bernoulli measurement density for detections of Assumption \ref{ass:Measurement} can be written in terms of $\settraj$ as:
\begin{subequations}
\begin{align}
	& \varphi_{k}(\setw_{k} | \settraj) = \\
	& \qquad 
		\left\{ \begin{array}{ll}
			1, & \settraj=\emptyset, \setw_k=\emptyset , \\
			1-P_{k}^{\rm D}(\traj), & \settraj=\{\traj\},\setw_k=\emptyset , \\
			P_{k}^{\rm D}(\traj)\varphi(z_k | \traj), & \settraj=\{\traj\},\setw_k=\{z_k\} , \\
			0, & \text{otherwise} ,
		\end{array} \right. \nonumber \\
	 & P_{k}^{\rm D}(\traj) = P^{\rm D}(x_{\td})\Delta_{k}(\td), \label{eq:probability_of_detection}\\ 
	& \varphi(z | \traj) = \varphi^{z}(z | x_{\td}).
\end{align}
\label{eq:single_object_measurement_model}%
\end{subequations}

The measurement model for sets of trajectories is of the same form as for sets of targets, so the update is analogous to the target measurement update in \cite{Wil12, GarciaFernandezWGS:2018}. %
\begin{theorem}[\tpmbm update]\label{th:Update}
Assume that the predicted density of the set of trajectories is a \pmbm of the form \eqref{eq:PMBMdensityDefinition} and the set of measurements is $\{z_k^1,\dots,z_k^{m_k}\}$. Then, the posterior density of the set of trajectories is a \pmbm of the form \eqref{eq:PMBMdensityDefinition} with $n_{k|k} = n_{k|k-1} + m_k$, and	
	{
	\begin{align}
		\lambdau_{k|k}(\traj) &= \left( 1-P_{k}^{\rm D}(\traj)\right)\lambdau_{k|k-1}(\traj),\label{eq:UndetectedUpdate} \\
		\assocset^k &= \assocset^{k-1} \cup \big\{(k,j)|j\in\{1,\dots,m_k\}\big\} .
	\end{align}
	}

	For tracks continuing from previous time steps ($i\in\{1,\dots,n_{k|k-1}\}$), a local hypothesis is created for each combination of a previous local hypothesis and either a missed detection or an update with one of the $m_k$ measurements. This implies that the number of local hypotheses is $h^i_{k|k} = h^i_{k|k-1}(1+m_k)$. For missed detection hypotheses ($i\in\{1,\dots,n_{k|k-1}\}$, $\assoc^i\in\{1,\dots,h_{k|k-1}^{i}\})$, the parameters are
	\begin{subequations}
		\begin{align}
			\assocset^k(i,\assoc^i) &= \assocset^{k-1}(i,\assoc^i) , \label{eq:MissUpdateMeasSet} \\
			w_{k|k}^{i,\assoc^i} &= w^{i,\assoc^i}_{k|k-1}\left(1-r_{k|k-1}^{i,\assoc^i} \conv{f_{k|k-1}^{i,\assoc^i}}{P_{k}^{\rm D}} \right) , \label{eq:MissUpdateW} \\
			r^{i,\assoc^i}_{k|k} &= \frac{r_{k|k-1}^{i,\assoc^i}\conv{f_{k|k-1}^{i,\assoc^i}}{1-P_{k}^{\rm D}}}{1-r_{k|k-1}^{i,\assoc^i}\conv{f_{k|k-1}^{i,\assoc^i}}{P_{k}^{\rm D}}} , \label{eq:MissUpdatePex}\\
			f^{i,\assoc^i}_{k|k}(\traj) &= \frac{\left(1-P_{k}^{\rm D}(\traj)\right)f_{k|k-1}^{i,\assoc^i}(\traj)}{\conv{f_{k|k-1}^{i,\assoc^i}}{1-P_{k}^{\rm D}}} . \label{eq:MissUpdateKin}
		\end{align}
		\label{eq:BernoulliMissedUpdate}
	\end{subequations}
	For a previous local hypothesis $\tilde{\assoc}^i\in\{1,\dots,h^i_{k|k-1}\}$ of an existing track ($i\in\{1,\dots,n_{k|k-1}\}$), the local hypothesis generated by the association with measurement $z_k^j$ has $\assoc^i=\tilde{\assoc}^i+h^i_{k|k-1} j$ and
	\begin{subequations}
	\begin{align}
	\assocset^k(i,\assoc^i) &= \assocset^{k-1}(i,\tilde{\assoc}^i) \cup \{(k,j)\} , \label{eq:DetUpdateMeasSet} \\
	w_{k|k}^{i,\assoc^i} &= w^{i,\tilde{\assoc}^i}_{k|k-1}r_{k|k-1}^{i,\tilde{\assoc}^i} \conv{f_{k|k-1}^{i,\tilde{\assoc}^i}}{\varphi(z_k^j|\cdot) P_{k}^{\rm D}} , \label{eq:DetUpdateW}\\
	r^{i,\assoc^i}_{k|k} &= 1 , \label{eq:DetUpdatePex}\\
	f^{i,\assoc^i}_{k|k}(\traj) &= \frac{\varphi(z_k^j|\traj)P_{k}^{\rm D}(\traj)f_{k|k-1}^{i,\tilde{\assoc}^i}(\traj)}{\conv{f_{k|k-1}^{i,\tilde{\assoc}^i}}{\varphi(z_k^j|\cdot) P_{k}^{\rm D}}} .\label{eq:DetUpdateKin}
	\end{align}
	\label{eq:BernoulliMeasurementUpdate}
	\end{subequations}
	Finally, the new track initiated by measurement $z_k^j$ has $i=n_{k|k-1}+j$, $h^i_{k|k}= 2$, and
	\begin{subequations}
	\begin{align}
	\assocset^k(i,1) &= \emptyset, \quad
	w^{i,1}_{k|k} = 1, \quad r^{i,1}_{k|k} = 0 , \label{eq:NewTargetNonExistWQ} \\
	\assocset^k(i,2) &= \{(t,j)\} \label{eq:PoisUpdateMeasSet} \\
	w^{i,2}_{k|k} &=  \lfa(z_k^j) + \conv{\lambdau_{k|k-1}}{\varphi(z_k^j|\cdot)P_{k}^{\rm D}} , \label{eq:PoisUpdateW} \\
	r_{k|k}^{i,2} &= \frac{\conv{\lambdau_{k|k-1}}{\varphi(z_k^j|\cdot)P_{k}^{\rm D}}}{\lfa(z_k^j) + \conv{\lambdau_{k|k-1}}{\varphi(z_k^j|\cdot)P_{k}^{\rm D}}} , \label{eq:PoisUpdatePex}\\
	f_{k|k}^{i,2}(\traj) &= \frac{\varphi(z_k^j|\traj)P_{k}^{\rm D}(\traj)\lambdau_{k|k-1}(\traj)}{\conv{\lambdau_{k|k-1}}{\varphi(z_k^j|\cdot)P_{k}^{\rm D}}} . \label{eq:PoisUpdateKin} 
	\end{align}%
	\label{eq:new_target_update}%
	\end{subequations}%
\end{theorem}%
\begin{proof}
	Analogous to the \pmbm update proof \cite[Thm. 2]{Wil12} \cite{GarciaFernandezWGS:2018} considering the measurement model \eqref{eq:single_object_measurement_model}.
\end{proof}

\subsection{Properties of the resulting \tpmbm filters}\label{sec:PMBMtrackerProperties}
Two \tpmbm filters are obtained from the theorems:
\begin{enumerate}
	\item A \tpmbm filter for the set of all trajectories (\probform~\ref{pf:all_trajectories}). Its filtering recursion is provided by the prediction in Theorem~\ref{th:PredictionAllTrajectories} and the update in Theorem~\ref{th:Update}.
	\item A \tpmbm filter for the set of current trajectories (\probform~\ref{pf:current_trajectories}). Its filtering recursion is provided by the prediction in Theorem~\ref{th:PredictionCurrentTrajectories} and the update in Theorem~\ref{th:Update}.
\end{enumerate}

Both \tpmbm filters are track oriented. For each measurement, a potential new track is initiated, see \eqref{eq:new_target_update}. In the update, additional hypotheses are created, as indicated in \eqref{eq:BernoulliMissedUpdate} and \eqref{eq:BernoulliMeasurementUpdate}. In the prediction, the number of tracks and hypotheses remains constant, see \eqref{eq:PredictedNumberOfTracks} and \eqref{eq:PredictedNumberOfHypotheses}.

We proceed to discuss the representation of the \ppp intensity and the Bernoulli densities in the \tpmbm filters. We also present relevant properties of the pmf of the start time step and the pmf of the end time step for each Bernoulli trajectory.

\subsubsection{Density/intensity representation}%
\label{sec:DensityIntensityEstimation}
Consider a trajectory mixture density of the form
\begin{subequations}
\begin{align}
	f_{}(\traj) & =  \sum_{j\in\indexSetJ} \weight_{}^{j} f_{}^{j}(\traj ; \trajdensityparams^j), \label{eq:Mixture_Density_Intensity}\\
	f_{}^{j}(\traj ; \trajdensityparams^j) & = \begin{cases} p^{j}(\stseq_{\timeseq{b^j}{e^j}}) & (\tb,\td)=(b^j,e^j), \\
	0 & \text{otherwise}, \end{cases} \\
	\trajdensityparams^j & = \left(b^j,e^j,p^j(\cdot)\right), 
\end{align}%
where $\indexSetJ$ is an index set, and each mixture component is characterised by a weight $\weight_{}^{j}$ and a parameter $\trajdensityparams^j$. The parameter consists of a birth time $b_{}^{j}$, a most recent time $e_{}^{j}$, where $b^{j} \leq e^{j}$, and a density $p_{}^{j}(\stseq_{\timeseq{b^j}{e^j}})$. 
For the density \eqref{eq:Mixture_Density_Intensity}, the pmf of $(\tb,\td)$ is
\begin{align}
	P(\tb,\td) = \int f_{}(\traj) \diff\stseq_{\timeseq{\tb}{\td}} =   \sum_{j\in\indexSetJ} \weight_{}^{j} \Delta_{b^j}(\tb) \Delta_{e^j}(\td) .
\end{align}%
\label{eq:mixture_trajectory_density}%
\end{subequations}%

For the mixture weights, it holds that $\sum_{j}\weight_{}^{j} = 1$ if $f_{}(\traj)$ is a density, and $\sum_{j}\weight_{}^{j} \geq 0$ if $f_{}(\traj)$ is a \ppp intensity. Note that there is no restriction in \eqref{eq:mixture_trajectory_density} that $b^j$ and $e^j$ must be unique, i.e., we may have $b^j=b^{j'}$ and/or $e^j=e^{j'}$ for $j,j'\in\indexSetJ$, $j\neq j'$. Densities/intensities of the form \eqref{eq:mixture_trajectory_density} facilitate representations for the state sequence $\stseq_{\timeseq{\tb}{\td}}$, conditioned on $\tb$ and $\td$.


The target birth \ppp intensity $\lambda_{k}^{\mathrm{b}}(x_k)$ is often modelled as an unnormalised distribution mixture, often a Gaussian mixture. It then follows that the trajectory birth \ppp intensity $\lambda_{k}^{{\rm B}}(\traj)$, cf. \eqref{eq:PoissonBirthIntensity}, is of the form \eqref{eq:mixture_trajectory_density}, with the special structure that $b^j=e^j=k$. From this, it follows that the \ppp intensity $\lambdau_{k|k'}(\traj)$, and all Bernoulli densities $f_{k|k'}^{i,\assoc^i}(\traj)$ will be of the form \eqref{eq:mixture_trajectory_density}. 

In other words, the parameters of the \tpmbm posterior are
\begin{subequations}
\begin{align}
	\lambdau_{}(\cdot), \left\{ \left( w_{}^{i,\assoc^i} , r_{}^{i,\assoc^i} , f_{}^{i,\assoc^i}(\cdot) \right)  \right\}_{\assoc \in\assocspace_{}, \ i\in\trackTable_{}} ,
\end{align}
with intensity and state densities of the form \eqref{eq:mixture_trajectory_density},
\begin{align}
	\lambdau_{}(\traj) & = \sum_{j\in\indexSetJ_{}^{\rm u}} \weight^{{\rm u},j} f^{{\rm u},j}(\traj ; \trajdensityparams^{{\rm u},j}), \\
	f_{}^{i,\assoc^i}(\traj) & = \sum_{j\in\indexSetJ_{}^{i,\assoc^i}} \weight^{i,\assoc^i,j} f^{i,\assoc^i,j}(\traj ; \trajdensityparams^{i,\assoc^i,j}),
\end{align}%
\label{eq:PosteriorTrajectoryPMBMparameters_time_k}%
\end{subequations}%
where $\indexSetJ_{}^{\rm u}$ and $\indexSetJ_{}^{i,\assoc^i}$ are index sets for the mixture density components, and $\weight^{{\rm u},j}$ and $\weight^{i,\assoc^i,j}$ are weights such that $ \sum_{j\in\indexSetJ_{}^{\rm u}} \weight^{{\rm u},j} \geq 0$ and $ \sum_{j\in\indexSetJ_{}^{i,\assoc^i}} \weight^{i,\assoc^i,j} = 1$.


\subsubsection{Properties of the pmf of the time of birth}

Consider a data association in which a measurement $z_{k}$ at time step $k$ is associated to a potential new target. Conditioned on the association, for the trajectory end time $\td$, we have that $\Pr(\td=k)=1$. We proceed to focus on the time of birth pmf. The new Bernoulli track density is of the form \eqref{eq:PoisUpdateKin},
\begin{align}
	f_{k|k}(\traj) =& \frac{\varphi(z_k|\traj)P_{k}^{\rm D}(\traj)\lambdau_{k|k-1}(\traj)}{\conv{\lambdau_{k|k-1}}{\varphi(z_k|\cdot)P_{k}^{\rm D}}} .
\end{align}
With a \ppp intensity of the mixture form \eqref{eq:mixture_trajectory_density},
\begin{align}
	\lambdau_{k|k-1}(\traj) & = \sum_{j\in\indexSetJ_{k|k-1}^{\rm u}} \weight_{k|k-1}^{{\rm u},j} f_{k|k-1}^{{\rm u},j}(\traj ; \trajdensityparams_{k|k-1}^{{\rm u},j}),
\end{align}%
we get a multi-modal posterior trajectory density $f_{k|k}(\traj)$, which can be pruned if necessary.  The pmf for $\tb$ is
\begin{align}
	P_{k|k}(\tb) = 
		\begin{cases}  
			\frac{  \sum_{j\in\indexSetJ_{k|k-1}^{{\rm u},\td=k}} \weight_{k|k-1}^{{\rm u},j} q_{k}^{{\rm u},j}(z_k) \Delta_{ b_{k|k-1}^{{\rm u},j} }( \tb ) }{  \sum_{j\in\indexSetJ_{k|k-1}^{{\rm u},\td=k}} \weight_{k|k-1}^{{\rm u},j} q_{k}^{{\rm u},j}(z_k)  } & \tb \leq k, \\
			0 & \tb > k ,
		\end{cases}
\end{align}
where $\indexSetJ_{k|k-1}^{{\rm u},\td=k} = \left\{ j \in \indexSetJ_{k|k-1}^{\rm u} : e_{k|k-1}^{{\rm u},j} = k \right\}$, and
\begin{align}
	q_{k}^{{\rm u},j}(z_k) = \int \varphi^{z}(z_{k} | \stseq_{k}) P_{k}^{\rm D}(\stseq_{k}) p_{k|k-1}^{{\rm u},j}\left(\stseq_{k} \right) \diff\stseq_{k} .
\end{align}
This implies that the time of birth of a newly detected trajectory is not deterministic, but has a certain pmf. In addition, as more measurements are associated, the trajectory density changes, which implies that the pmf of the time of birth can also change.


\subsubsection{Properties of the pmf of the trajectory end time}
Consider a posterior Bernoulli density $f_{k|k}(\traj)$ of the form \eqref{eq:mixture_trajectory_density}, to which a measurement was associated at time $k$.  For the sake of brevity, assume that it has a single mixture component with parameter $\trajdensityparams~=~\left(b,k,p_{k|k}(\cdot)\right)$.
In addition, the probability of survival is a constant $P^{\rm S}(x_{k}) = P^{\rm S}$, and $Q^{\rm S}=1-P^{\rm S}$. Then, the predicted density at time $k+1$ is a mixture
\begin{subequations}
\begin{align}
	& f_{k+1|k}(\traj_{k+1}) = Q^{\rm S}f^{0}(\traj_{k+1} ; \trajdensityparams^{0}) + P^{\rm S} f^{1}(\traj_{k+1} ; \trajdensityparams^{1}) , \\
	& \trajdensityparams^{0} = \left(b, k , p_{k|k}(\stseq_{\timeseq{b}{k}}) \right) , \label{eq:TrajectoryEnded} \\
	& \trajdensityparams^{1} = \left(b, k+1 , p_{k|k}(\stseq_{\timeseq{b}{k}}) \pi^{x}(x_{k+1} | x_{k}) \right).  \label{eq:TrajectoryContinued}
\end{align}
\end{subequations}
Note that the state sequence density for $\td=k$, \eqref{eq:TrajectoryEnded}, is given by marginalising $x_{k+1}$ from the state sequence density for $\td=k+1$, \eqref{eq:TrajectoryContinued}. This has important implications for the implementation of the \tpmbm filter for the set of all trajectories: during the prediction step, the hypothesis space for the trajectory density increases, due to the fact that we do not know if the trajectory ended at time $k$, \eqref{eq:TrajectoryEnded}, or continued to time $k+1$, \eqref{eq:TrajectoryContinued}. However, it is not necessary to explicitly represent both state sequence densities, instead it is sufficient to explicitly represent the state sequence density that continued to time $k+1$, as well as a pmf $P_{k+1|k}(\td)$. This is especially important when there are several consecutive misdetections associated: for $N$ consecutive misdetections, with a single pmf $P_{k|k'}(\td)$ and a single density $p_{k|k'}(\stseq_{\timeseq{b}{k}})$ we can compactly represent $N+1$ different hypotheses for $\td$ and $\stseq_{\timeseq{b}{\td}}$. Similar observations were made for \mht in \cite[Sec. IV]{CorCar14}. 

We now also assume that $P^{\rm D}(\stseq_k) = P^{\rm D}$, with $Q^{D}=1-P^{\rm D}$, as this facilitates the closed-form expression for the pmf of the time of the latest state $P_{k|k}(\td)$. Consider a Bernoulli density where $\tau$ is the last time step that a measurement was associated. Given this association, we have that $\Pr(\td\geq\tau)=1$. Then, at time $k > \tau$, if $\qdps = Q^{\rm D}P^{\rm S}<1$,\footnote{$Q^{\rm D}P^{\rm S}=1 \Rightarrow P^{\rm S}=1,\ P^{\rm D}=0$. Note that \mtt without detections ($P^{\rm D}=0$) is a problem of little interest.} the posterior pmf of $\td$ is
\begin{align}
	P_{k|k}(\td) & = 
		\begin{cases} 
			0 & \td < \tau \text{ or } \td > k , \\
			\frac{1}{C}Q^{\rm S}\qdps^{i} & \td=\tau+i,  i\in\timeset{0}{k-\tau-1} , \\
			\frac{1}{C}\qdps^{k-\tau} & \td = k , 
		\end{cases}
\end{align}%
where $C = Q^{\rm S}\left({1 - \qdps^{k-\tau}}\right)\left({1 - \qdps}\right)^{-1} + \qdps^{k-\tau}$.
When only misdetections are associated, in the limit, the pmf of $\td$ is
\begin{align}
	\lim_{k\rightarrow\infty}P_{k|k}(\td) =
		\begin{cases} 
			0 & \td < \tau, \\
			\qdps^i (1-\qdps) & \td=\tau+i, i\geq 0,
		\end{cases}
\end{align}
which is the pmf of a geometric distribution with success probability $1-\qdps$.
From this it follows that, eventually, for the given association (only misdetections after time step $\tau$), the maximum a posteriori (\map) estimate of $\td$ is $\tau$ and the expected value of $\td$ is $\tau+\frac{\qdps}{1-\qdps}$, which can be rounded to the nearest integer. Note that $\Pr(\td>\tau)=\phi$, which may be significant for high $P^{\rm S}$ and/or low $P^{\rm D}$. 
The pmf $P_{k|k}(\td)$ can also be computed recursively. 


\subsection{Estimator}
\label{sec:estimator}
We proceed to explain how to estimate the set of trajectories. 

\subsubsection{Single trajectory estimator}
\label{sec:SingleTrajectoryEstimator}
Given a trajectory density $f(\traj) = p(\stseq_{\timeseq{\tb}{\td}} | \tb,\td) P(\tb,\td)$, in this work, we obtain a trajectory estimate $\hat{\traj} = (\hat{\tb}, \ \hat{\td}, \ \hat{\stseq}_{\hat{\tb}:\hat{\td}} )$ as follows. The \map estimates of the trajectory start and end times are
\begin{align}
	\left(\hat{\tb}, \hat{\td}\right) = \arg \max_{(\tb,\td)} P(\tb,\td).
\end{align}
For trajectory densities of the form \eqref{eq:mixture_trajectory_density}, this corresponds to finding the component with maximum weight.
The trajectory estimate from the estimated birth time to the estimated most recent state time is computed as the conditional expectation
\begin{align}
	\hat{\stseq}_{\hat{\tb}: \hat{\td}} = \mathbb{E}\left[\stseq_{\hat{\tb}: \hat{\td}} \left | \hat{\tb}, \hat{\td} \right.\right]. \label{eq:ExpectedStateSequence}
\end{align}

\subsubsection{Multiple trajectory estimator}
\label{sec_estimator}

This paper uses the following estimator of the set of trajectories given the \tpmbm posterior density. We first select the global hypothesis with the highest weight
\begin{subequations}
\begin{align}
	\assoc_{\star} = \arg \max_{\assoc} w_{k|k}^{\assoc}.
\end{align}
Then, the estimated set of trajectories is
\begin{align}
	\hat{\settraj}_{k|k} = \left\{ \hat{\traj}_{k|k}^{i,\assoc_{\star}^{i}}  \right\}_{i \in \mathbb{T}_{k|k} \ : \ r_{k|k}^{i,\assoc_{\star}^{i}} > r^{\rm e}} ,
\end{align}%
\label{eq:MultiTrajectoryEstimator}%
\end{subequations}%
where the detection threshold $r^{\rm e}$ on the probability of existence is a parameter, and the trajectory estimate $\hat{\traj}_{k|k}^{i,\assoc_{\star}^{i}}$ is computed as in Section~\ref{sec:SingleTrajectoryEstimator}. The estimator \eqref{eq:MultiTrajectoryEstimator} corresponds to \pmbm Estimator 1 \cite[Sec. VI.A]{GarciaFernandezWGS:2018}.



\section{Properties of \pmbm densities on sets of trajectories}
\label{sec:pmbm_relations}
This section provides some properties on \pmbm densities on sets of trajectories. Section~\ref{sec:time_marginalisation} states the time marginalisation theorem. Section \ref{subsec:pmbm_arbitrary_time} provides a result analogous to \eqref{eq:Gaussian_state_sequence} for \pmbm densities. Finally, Section~\ref{sec:pmbm_tracker_filter_relation} establishes relations between the \tpmbm filters and the \pmbm filter.

\subsection{\tpmbm time marginalisation}
\label{sec:time_marginalisation}
Given a density on a set of trajectories defined on a time interval, time marginalisation is the process of marginalising states for particular times from the trajectories, such that only states related to a different time interval remain. 

We use $\int \diff \stseq_{\timeseq{\tb}{\td} \backslash \timeseq{b}{e}}$ to denote the integral of the states in the sequence $\stseq_{\timeseq{\tb}{\td}}$ that are not in the interval $ \timeseq{b}{e}$, with $\tb\leq b\leq e\leq \td$. The following theorem explains the time marginalisation operation for \pmbm densities on sets of trajectories.


\begin{figure*}
	\centering

	\includegraphics[width=0.2\textwidth]{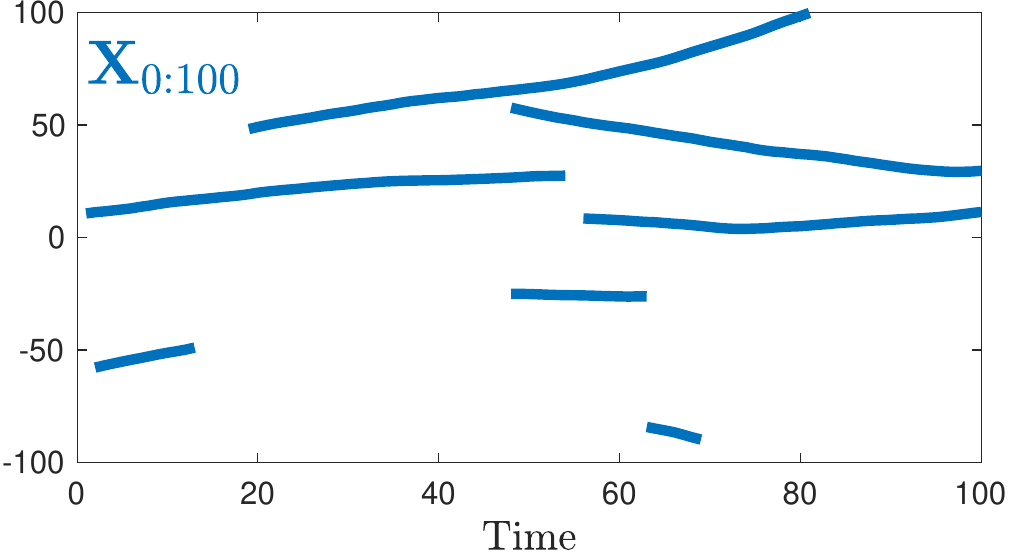}\includegraphics[width=0.2\textwidth]{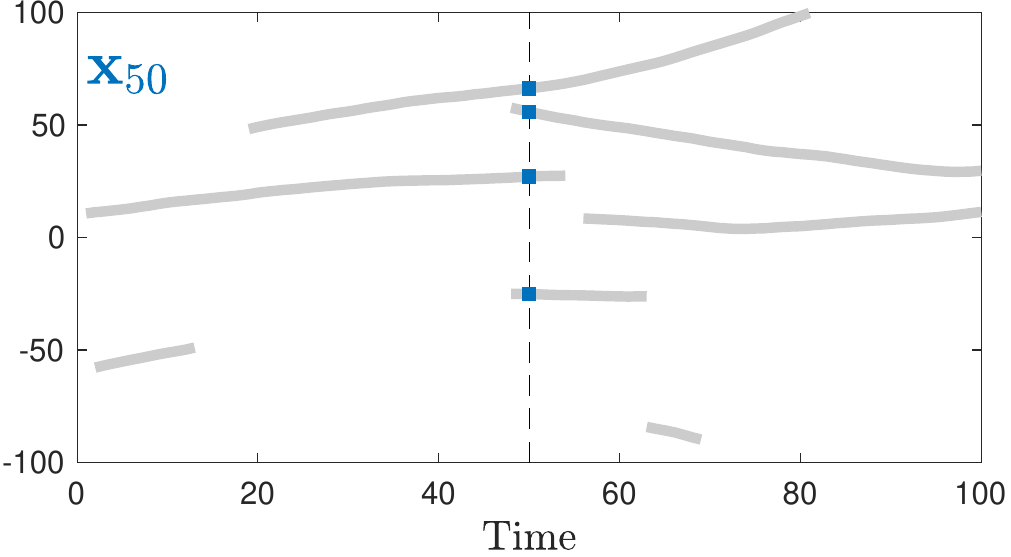}\includegraphics[width=0.2\textwidth]{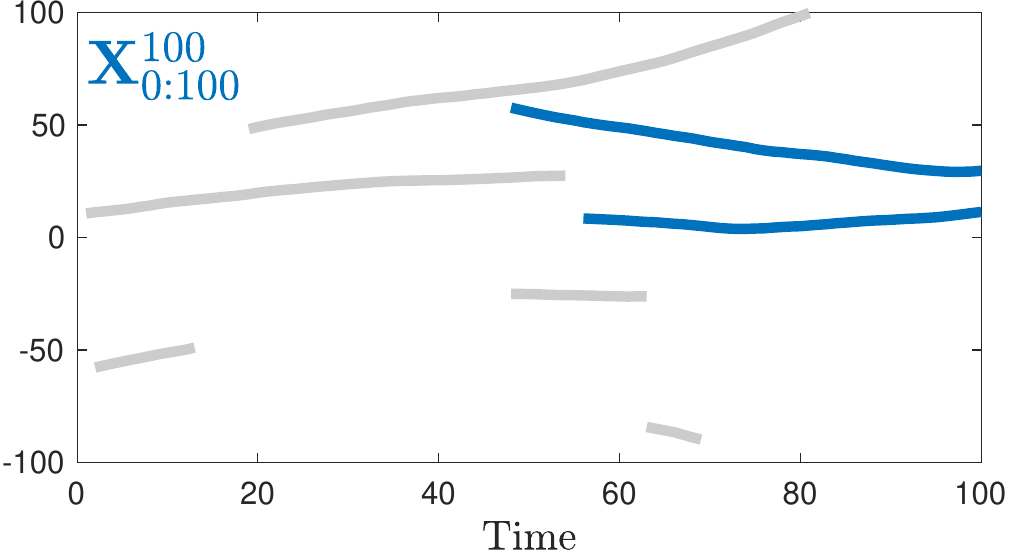}\includegraphics[width=0.2\textwidth]{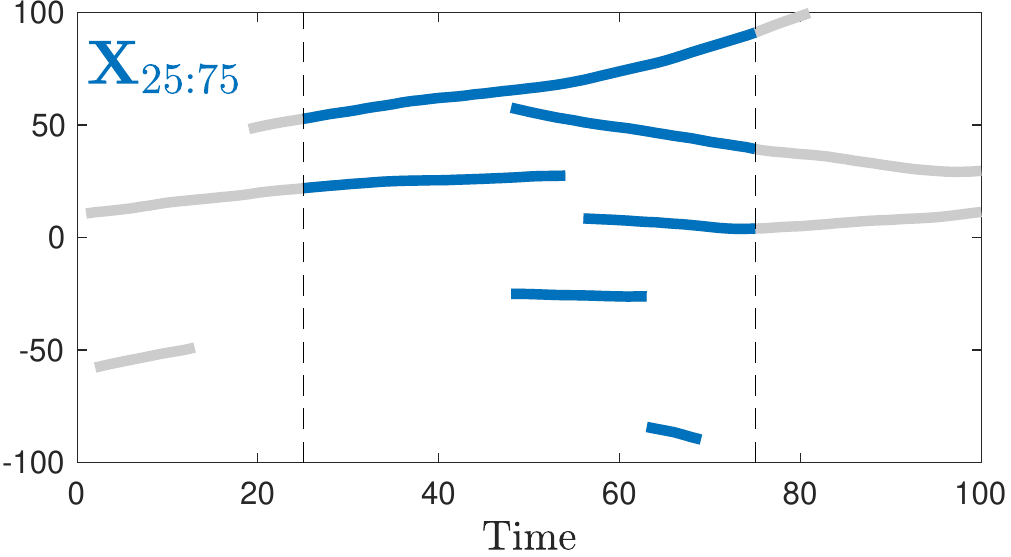}\includegraphics[width=0.2\textwidth]{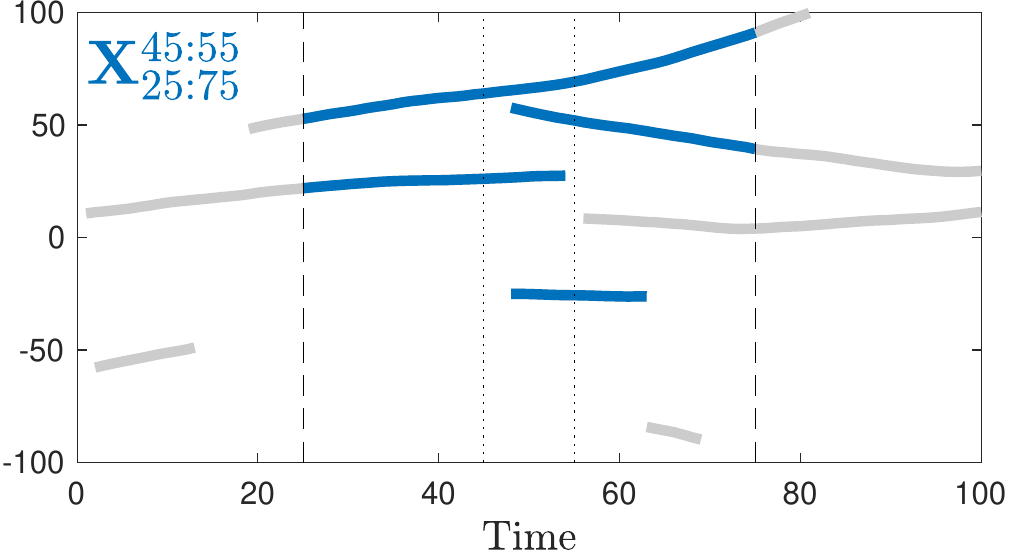}
	
	\caption{Illustration of the trajectories considered in Theorem~\ref{thm:all_to_trajs_in_a2g_alive_in_e2z}. Left to right: set of all (1-D) trajectories up to the current time step $100$, $\settraj_{\timeseq{0}{100}}$; corresponding target set at time step $50$, $\settarget_{50}$; corresponding set of current trajectories $\settraj_{\timeseq{0}{100}}^{100}$; corresponding set of trajectories in the time interval $25$ to $75$, $\settraj_{\timeseq{25}{75}}$; corresponding set of trajectories in the time interval $25$ to $75$ that alive at some time in the interval $45$ to $55$, $\settraj_{\timeseq{25}{75}}^{\timeseq{45}{55}}$. If $\settraj_{\timeseq{0}{100}}$ is \pmbm distributed, so are the rest of the sets of trajectories.}
	\label{fig:example_trajectory_sets}
\end{figure*}

\begin{theorem}\label{thm:all_to_trajs_in_a2g_alive_in_e2z}
	Given a posterior \pmbm density, with parameters \eqref{eq:PosteriorTrajectoryPMBMparameters_time_k}, for the set of all trajectories at time $k$, $\settraj_{\timeseq{0}{k}}$,  the posterior density for the set of trajectories in the interval $\timeseq{\alpha}{\gamma}$ whose trajectories were present in the surveillance area at least in one time step in the interval $\timeseq{\eta}{\zeta}$, $\settraj_{\timeseq{\alpha}{\gamma}}^{\timeseq{\eta}{\zeta}}$ with $\timeset{\eta}{\zeta} \subseteq \timeset{\alpha}{\gamma}\subseteq\timeset{0}{k}$, is a \pmbm characterised by
	\begin{subequations}
		\begin{align}
		\tilde{\lambda}^{\rm u}_{}(\cdot), \left\{ \left( w_{}^{i,\assoc^i} , \tilde{r}_{}^{i,\assoc^i} , \tilde{f}_{}^{i,\assoc^i}(\cdot) \right)  \right\}_{\assoc \in\assocspace_{}, \ i\in\trackTable_{} :  \tilde{r}_{}^{i,\assoc^i} \neq 0}
		\end{align}
		with \ppp intensity
		\begin{align}
		& \tilde{\lambda}^{\rm u}_{}(\traj) = \sum_{j\in \indexSetJ_{}^{{\rm u},\timeseq{\eta}{\zeta}}} \weight^{{\rm u},j}  \tilde{f}^{{\rm u},j}(\traj ; \tilde{\trajdensityparams}^{{\rm u},j}) , \\
		& \indexSetJ_{}^{{\rm u},\timeseq{\eta}{\zeta}} = \left\{ j \in \indexSetJ_{}^{\rm u} : \timeset{b^{{\rm u},j}}{e^{{\rm u},j}} \cap  \timeset{\eta}{\zeta} \neq \emptyset  \right\} , \label{eq:PoissonCompsAlive} \\
		& \tilde{\trajdensityparams}^{{\rm u},j}  = \left( \tilde{b}^{{\rm u},j} , \tilde{e}^{{\rm u},j} , \tilde{p}^{{\rm u},j}(\cdot) \right), \\
		& \tilde{b}^{{\rm u},j} = \max\left( b^{{\rm u},j}, \alpha \right), \quad \tilde{e}^{{\rm u},j} = \min\left( e^{{\rm u},j}, \gamma \right) ,\label{eq:PoissonCompBirthDeath} \\
		& \tilde{p}^{{\rm u},j}(\stseq_{\timeseq{\tilde{b}^{{\rm u},j}}{\tilde{e}^{{\rm u},j}}}) = \int p^{{\rm u},j}(\stseq_{\timeseq{b^{{\rm u},j}}{e^{{\rm u},j}}}) \diff \stseq_{ \timeseq{b^{{\rm u},j}}{e^{{\rm u},j}} \backslash \timeseq{\tilde{b}^{{\rm u},j}}{\tilde{e}^{{\rm u},j}} } ,\label{eq:PoissonCompStateSeqDensity} 
		\end{align}%
		and Bernoulli parameters
		\begin{align}
		& \tilde{r}_{}^{i,\assoc^i} = r_{}^{i,\assoc^i}  \sum_{j\in\indexSetJ_{}^{i,\assoc^i,\timeseq{\eta}{\zeta}}} \weight^{i,\assoc^i,j} , \\
		& \tilde{f}_{}^{i,\assoc^i}(\traj) = \frac{ \sum_{j\in\indexSetJ_{}^{i,\assoc^i,\timeseq{\eta}{\zeta}}} \weight^{i,\assoc^i,j} \tilde{f}^{i,\assoc^i,j}(\traj ; \tilde{\trajdensityparams}^{i,\assoc^i,j}) }{ \sum_{j\in\indexSetJ_{}^{i,\assoc^i,\timeseq{\eta}{\zeta}}} \weight^{i,\assoc^i,j} }, \\
		& \indexSetJ_{}^{i,\assoc^i,\timeseq{\eta}{\zeta}} = \left\{ j\in\indexSetJ_{}^{i,\assoc^i} : \timeset{b^{i,\assoc^i,j}}{e^{i,\assoc^i,j}} \cap \timeset{\eta}{\zeta} \neq \emptyset \right\} , \label{eq:BernoulliCompsAlive} \\
		& \tilde{\trajdensityparams}^{i,\assoc^i,j} = \left( \tilde{b}^{i,\assoc^i,j}, \tilde{e}^{i,\assoc^i,j}, \tilde{p}^{i,\assoc^i,j}(\cdot) \right) , \\
		& \tilde{b}^{i,\assoc^i,j} = \max\left( b^{i,\assoc^i,j}, \alpha \right), \tilde{e}^{i,\assoc^i,j} = \min\left( e^{i,\assoc^i,j}, \gamma \right) ,\label{eq:BernoulliCompBirthDeath}\\
		& \tilde{p}^{i,\assoc^i,j}(\stseq_{\timeseq{\tilde{b}^{i,\assoc^i,j}}{\tilde{e}^{i,\assoc^i,j}}}) \label{eq:BernoulliCompStateSeqDensity}\\
		& = \int p^{i,\assoc^i,j}(\stseq_{\timeseq{b^{i,\assoc^i,j}}{e^{i,\assoc^i,j}}}) \diff \stseq_{ \timeseq{b^{i,\assoc^i,j}}{e^{i,\assoc^i,j}} \backslash \timeseq{\tilde{b}^{i,\assoc^i,j}}{\tilde{e}^{i,\assoc^i,j}} } . \nonumber
		\end{align}%
		
	\end{subequations}%
\end{theorem}%
\begin{proof} See Appendix~\ref{app:pmbm_relations_theorem_proofs}. 
\end{proof}

Theorem~\ref{thm:all_to_trajs_in_a2g_alive_in_e2z} is understood as follows: from \eqref{eq:PoissonCompsAlive}/\eqref{eq:BernoulliCompsAlive}, we obtain the densities of the trajectories that are alive in the time interval $\timeseq{\eta}{\zeta}$. From \eqref{eq:PoissonCompBirthDeath}/\eqref{eq:BernoulliCompBirthDeath}, we get the time of birth and time of most recent state parameters, and from \eqref{eq:PoissonCompStateSeqDensity}/\eqref{eq:BernoulliCompStateSeqDensity}, we get the state sequence densities corresponding to $\timeseq{\alpha}{\gamma}$. Some examples of time marginalisation are illustrated in Fig. \ref{fig:example_trajectory_sets}.

\subsection{\tpmbm density for arbitrary time intervals}\label{subsec:pmbm_arbitrary_time}
It follows from Theorem~\ref{th:PredictionAllTrajectories} and Theorem~\ref{th:Update} that the density
\begin{align}
	f_{k|k'} (\settraj_{\timeseq{0}{k}}) \triangleq f_{\timeseq{0}{k}|\timeseq{0}{k'}} (\settraj_{\timeseq{0}{k}}) \label{eq:General_PMBM_tracker_density}
\end{align}
is \pmbm, regardless if we are doing prediction ($k>k'$) or filtering ($k=k'$). In other words, in terms of prediction and filtering, we have a result that is analogous to \eqref{eq:Gaussian_state} for the Gaussian distribution in linear Gaussian models. The following theorem presents more general results.
\begin{theorem}\label{thm:density_traj_alpha2gamma_alive_eta2zeta_meas_xi2chi}
	Given the measurements in any time interval $\timeseq{\xi}{\chi}$, the density of the set of trajectories in any time interval $\timeseq{\alpha}{\gamma}$ that are alive at any point in the time interval $\timeseq{\eta}{\zeta}$, such that $\timeset{\eta}{\zeta}\subseteq\timeset{\alpha}{\gamma}$, is \pmbm. This density is denoted by $f_{\timeseq{\alpha}{\gamma} | \timeseq{\xi}{\chi} }^{\timeseq{\eta}{\zeta}}\left(\settraj_{\timeseq{\alpha}{\gamma}}^{\timeseq{\eta}{\zeta}}\right)$. 
\end{theorem}
The following corollary to Theorem~\ref{thm:density_traj_alpha2gamma_alive_eta2zeta_meas_xi2chi} is analogous to \eqref{eq:Gaussian_state_sequence}.
\begin{corollary}
	If $\eta=\alpha$ and $\zeta=\gamma$ in Theorem~\ref{thm:density_traj_alpha2gamma_alive_eta2zeta_meas_xi2chi}, then 
	\begin{align}
		f_{\timeseq{\alpha}{\gamma} | \timeseq{\xi}{\chi} }\left(\settraj_{\timeseq{\alpha}{\gamma}}\right)
	\end{align}
	is \pmbm for any $\timeseq{\alpha}{\gamma}$ and $\timeseq{\xi}{\chi}$.
\end{corollary}
We prove Theorem~\ref{thm:density_traj_alpha2gamma_alive_eta2zeta_meas_xi2chi} in Appendix \ref{app:pmbm_relations_theorem_proofs}, making use of the time marginalisation theorem for \tpmbm densities (Theorem \ref{thm:all_to_trajs_in_a2g_alive_in_e2z}). To illustrate Theorem \ref{thm:density_traj_alpha2gamma_alive_eta2zeta_meas_xi2chi}, let us consider for example the sets $\settraj_{\timeseq{0}{100}}$, $\settarget_{50}$, and $\settraj_{\timeseq{25}{75}}^{\timeseq{45}{55}}$ in Figure \ref{fig:example_trajectory_sets}. Theorem~\ref{thm:density_traj_alpha2gamma_alive_eta2zeta_meas_xi2chi} indicates that these sets are \pmbm distributed given the measurements in any time interval.

\subsection{Relation between the \tpmbm filters and the \pmbm filter}
\label{sec:pmbm_tracker_filter_relation}
Theorems \ref{thm:all_to_trajs_in_a2g_alive_in_e2z} and \ref{thm:density_traj_alpha2gamma_alive_eta2zeta_meas_xi2chi} provide a way to compute the density $f_{\timeseq{\alpha}{\gamma} | \timeseq{\xi}{\chi} }^{\timeseq{\eta}{\zeta}}\left(\settraj_{\timeseq{\alpha}{\gamma}}^{\timeseq{\eta}{\zeta}}\right)$: first use the \tpmbm filter for the set of all trajectories, and then use time marginalisation to obtain the desired density. However, we note that the application of the \tpmbm filter followed by time marginalisation is possibly not the only way to compute the density $f_{\timeseq{\alpha}{\gamma} | \timeseq{\xi}{\chi} }^{\timeseq{\eta}{\zeta}}\left(\settraj_{\timeseq{\alpha}{\gamma}}^{\timeseq{\eta}{\zeta}}\right)$. Compare to the Gaussian result in \eqref{eq:Gaussian_state_sequence}: it is possible to compute the Gaussian density $p(\stseq_{k}| z_{\timeseq{0}{K}})$ by computing the joint density $p(\stseq_{\timeseq{0}{K}}| z_{\timeseq{0}{K}})$ and then marginalizing, however, one can also use forward-backward \textsc{rts}-smoothing, see, e.g., \cite{Sarkka:2013}. 


Using Theorem~\ref{thm:all_to_trajs_in_a2g_alive_in_e2z} for $\alpha=\gamma=\eta=\zeta=k$, the relations between the two \tpmbm filters, and between the \tpmbm filters and the \pmbm filter \cite{Wil12}, become clear. This is illustrated in the following diagram indicating how several densities can be obtained by applying different theorems:
\begin{align*}\small
	\begin{array}{c c c}
		f_{k|k}(\settraj_{\timeseq{0}{k}})  \overset{\text{Thm.~\ref{th:PredictionAllTrajectories}}}{\rightarrow} & f_{k+1|k}(\settraj_{\timeseq{0}{k+1}}) \overset{\text{Thm.~\ref{th:Update}}}{\rightarrow} & f_{k+1|k+1}(\settraj_{\timeseq{0}{k+1}})  \\
		\\
		\big\downarrow{\scriptsize\text{Thm.~\ref{thm:all_to_trajs_in_a2g_alive_in_e2z}}} & \big\downarrow{\scriptsize\text{Thm.~\ref{thm:all_to_trajs_in_a2g_alive_in_e2z}}} & \big\downarrow{\scriptsize\text{Thm.~\ref{thm:all_to_trajs_in_a2g_alive_in_e2z}}} \\
		\\
		f_{k|k}(\settarget_{k}) \overset{\text{\cite[Thm. 1]{Wil12}}}{\rightarrow} & f_{k+1|k}(\settarget_{k+1}) \overset{\text{\cite[Thm. 2]{Wil12}}}{\rightarrow} & f_{k+1|k+1}(\settarget_{k+1})
	\end{array}%
\end{align*}%

It can be seen that marginalising out the past states in the \tpmbm filter posterior, we obtain the \pmbm filter posterior \cite{Wil12}. Predicting and updating the \pmbm posterior yields the same result as predicting and updating the \tpmbm filter, and then marginalising out past states. Thus, the hypotheses in the \pmbm filter are implicitly operating on the latest time step of the trajectory hypotheses in the \tpmbm filter. Then, the target estimates of the \pmbm filter corresponding to the same Bernoulli component (created by a specific measurement) are implicitly estimating the current state of the associated trajectory in the \tpmbm filter. This indicates that we can sequentially connect target state estimates at different time steps to estimate trajectories using \pmbm filter metadata. This sequential track building procedure for the \pmbm filter is analogous to linking target state estimates with the same auxiliary variable, which can be used to mark each track (Bernoulli component) \cite{Angel20_e}. It should anyway be noted that this sequential procedure to estimate the set of trajectories is sub-optimal. The optimal estimation of the set of trajectories must be done directly from the \tpmbm posterior density.


\section{Linear Gaussian implementation}
\label{sec:linear_gaussian_implementation}
Let $\targetStateSpace = \mathbb{R}^{n_x}$, $\measStateSpace = \mathbb{R}^{n_z}$, and let the transition density and measurement model be linear and Gaussian, 
\begin{subequations}
\begin{align}
	\pi_{}^{x}(x | x') &= \Npdfbig{x}{F x'}{Q}, \\
	\varphi^{z}(z | x) &= \Npdfbig{z}{H x}{R},
\end{align}%
\label{eq:linear_gaussian_model}%
\end{subequations}%
where $F \in \mathbb{R}^{n_x \times n_x}$ is a state transition matrix, $H \in \mathbb{R}^{n_z \times n_x}$ is a measurement matrix, $Q\in \mathbb{S}_{n_x}^{++}$ and $R\in \mathbb{S}_{n_z}^{++}$ are the covariance matrices of the process noise and measurement noise, respectively, and $\mathbb{S}_{d}^{++}$ denotes the space of positive definite matrices of size $d\times d$.

For linear and Gaussian models \eqref{eq:linear_gaussian_model} we present three representations of the state sequence density. Let $m_{k|k}^{[a]}$ and $P_{k|k}^{[a]}$ be the parts of the mean and the covariance matrix for time step $a$. Let $P_{k|k}^{[\timeseq{a}{b},\timeseq{c}{d}]}$ denote the part of the covariance matrix with rows for time steps $a$ to $b$ and columns for time steps $c$ to $d$. The predictions and updates can be generalised to non-linear models by linearization \cite{Sarkka:2013,reustice-2006b}. 

\subsection{Gaussian moment form}
\label{sec:GaussianStateSequenceDensity}
The state sequence density can be expressed on Gaussian moment form, with mean $m_{k|k'}$ and covariance $P_{k|k'}$,
\begin{align}
	p_{k|k'}(\stseq_{\timeseq{\tb}{k}}) & = \mathcal{N}\left(\stseq_{\timeseq{\tb}{k}} \ ; \ m_{k|k'}, P_{k|k'} \right) , 
	\label{eq:regular_gaussian}%
\end{align}%

\subsubsection{Prediction}
Given the motion model \eqref{eq:linear_gaussian_model} and posterior parameters  $m_{k|k}$ and $P_{k|k}$, the predicted parameters are
\begin{subequations}
\begin{align}
	m_{k+1|k} & = \begin{bmatrix} m_{k|k} \\ Fm_{k|k}^{[k]} \end{bmatrix}, \\
	P_{k+1|k} & = \begin{bmatrix} 
		P_{k|k} & P_{k|k}^{[\timeseq{\tb}{k},k]}F^{\tp}  \\
		F P_{k|k}^{[k,\timeseq{\tb}{k}]} & FP_{k|k}^{[k]}F^{\tp}+Q
	\end{bmatrix} \label{eq:predicted_covariance_matrix}.
\end{align}%
\label{eq:Gaussian_prediction}%
\end{subequations}%

\subsubsection{Update}
Given the measurement model \eqref{eq:linear_gaussian_model}, predicted parameters $m_{k+1|k}$ and $P_{k+1|k}$,  and an associated detection $z$, the posterior parameters are
\begin{subequations}
\begin{align}
	m_{k+1|k+1} & = m_{k+1|k} + K (z - H m_{k+1 | k}^{[k]}), \\
	P_{k+1|k+1} & = P_{k+1|k} - KHP_{k+1|k}^{[k,\timeseq{\tb}{k}]}, \label{eq:updated_covariance_matrix}  \\
	K & = P_{k+1|k}^{[\timeseq{\tb}{k},k]}H^{\tp}(HP_{k+1|k}^{[k]}H^{\tp}+R)^{-1} .
\end{align}
\label{eq:Gaussian_update}%
\end{subequations}

\subsection{Gaussian information form}\label{subsec:Gaussian_information}
The density \eqref{eq:regular_gaussian} can be expressed in information form with information vector $\infvec_{k|k'}$ and information matrix $\infmat_{k|k'}$, such that $\infvec_{k|k'}=P_{k|k'}^{-1}m_{k|k'}$ and $\infmat_{k|k'}= P_{k|k'}^{-1}$ \cite{reustice-2006b}.

\subsubsection{Prediction}
Given the motion model \eqref{eq:linear_gaussian_model} and posterior parameters $\infvec_{k|k}$ and $\infmat_{k|k}$, the predicted parameters are \cite{reustice-2006b}
\begin{subequations}
\begin{align}
	\infvec_{k+1|k} & = \begin{bmatrix} \text{\footnotesize$\infvec_{k|k}$} \\ \text{\footnotesize$\zeromatrix_{n_x \times 1}$} \end{bmatrix} , \\
	\infmat_{k+1|k} & = 
		\begin{bmatrix} 
			\text{\footnotesize$\infmat_{k|k}^{[\timeseq{\tb}{k-1}]}$}			& \text{\scriptsize$\infmat_{k|k}^{[\timeseq{\tb}{k-1},k]}$}			& \text{\footnotesize$\zeromatrix_{(\tlen-1)n_x \times n_x}$} \\
			\text{\footnotesize$\infmat_{k|k}^{[k,\timeseq{\tb}{k-1}]}$}			& \text{\footnotesize$\infmat_{k|k}^{[k]}+F^{T}Q^{-1}F$}	& \text{\footnotesize$-F^{T}Q^{-1}$} \\
			\text{\footnotesize$\zeromatrix_{n_x\times(\tlen-1)n_x}$}	& \text{\footnotesize$-Q^{-1}F$}					& \text{\footnotesize$Q^{-1}$} 
		\end{bmatrix} , \label{eq:predicted_information_matrix}
\end{align}%
\label{eq:information_prediction}%
\end{subequations}%
where $\zeromatrix_{m\times n}$ is an $m$ by $n$ all-zero matrix. 

\subsubsection{Update}
Given the measurement model \eqref{eq:linear_gaussian_model}, predicted parameters $\infvec_{k+1|k}$ and $\infmat_{k+1|k}$, and an associated detection $z$, the posterior parameters are \cite{reustice-2006b,WalterEL:2007,MahonWPJR:2008}
\begin{subequations}
\begin{align}
	\infvec_{k+1|k+1} & = \infvec_{k+1|k} + \begin{bmatrix} \text{\footnotesize$\zeromatrix_{(\tlen-1)n_x \times 1}$} \\ \text{\footnotesize$H^{T} R^{-1} z$} \end{bmatrix}, \\
	\infmat_{k+1|k+1} & = \infmat_{k+1|k} + \begin{bmatrix} \text{\footnotesize$\zeromatrix_{(\tlen-1)n_x \times (\tlen-1)n_x}$}  & \text{\footnotesize$\zeromatrix_{(\tlen-1)n_x \times n_x}$} \\ \text{\footnotesize$\zeromatrix_{n_x \times (\tlen-1)n_x}$} & \text{\footnotesize$H^{T} R^{-1} H$} \end{bmatrix}. \label{eq:updated_information_matrix}
\end{align}%
\label{eq:information_update}%
\end{subequations}%

\subsection{Gaussian $L$-scan approximation}
In the $L$-scan approximation, the covariance matrix of \eqref{eq:regular_gaussian} is approximated as \cite[Sec. VI.C]{GarciaFernandezS:2019}
\begin{align}
	P_{k|k'} =  \diag{ P_{k|k'}^{[\tb]} , \ \ldots, P_{k|k'}^{[\td-L]}, \ P_{k|k'}^{[\td-L+1:\td]} }  .
\end{align}
In other words, states before the last $L$ time steps are assumed independent of the most recent $L$ states.

\subsubsection{Prediction}
Given the motion model \eqref{eq:linear_gaussian_model} and posterior parameters $m_{k|k}$ and $P_{k|k}$, the predicted parameters are
\begin{subequations}
\begin{align}
	m_{k+1|k} & = \begin{bmatrix} m_{k|k} \\ Fm_{k|k}^{[k]} \end{bmatrix} , \\
	P_{k+1|k} & = \diag{ P_{k|k}^{[\tb]} , \ \ldots, \ P_{k|k}^{[k-L]} , \ P_{k|k}^{[k-L+1]} , \ \Psi } , \label{eq:predicted_L_latest_covariance_matrix} \\
	\Psi & = \begin{bmatrix} 
				\text{\footnotesize$P_{k|k}^{[\timeseq{k-L+2}{k}]}$}	&	\text{\footnotesize$P_{k|k}^{[\timeseq{k-L+2}{k},k]}F^{\tp}$}  \\
				\text{\footnotesize$F P_{k|k}^{[k,\timeseq{k-L+2}{k}]}$}	&	\text{\footnotesize$FP_{k|k}^{[k]}F^{\tp}+Q$}
			\end{bmatrix}.
\end{align}%
\label{eq:Gaussian_L_latest_prediction}%
\end{subequations}%

\subsubsection{Update}
Given the measurement model \eqref{eq:linear_gaussian_model}, predicted parameters $m_{k+1|k}$ and $P_{k+1|k}$ and an associated detection $z$, the posterior parameters are
\begin{subequations}
\begin{align}
	m_{k+1|k+1} &= \begin{bmatrix} m_{k+1|k}^{[\timeseq{\tb}{k-L}]} \\ m_{k+1|k}^{[\timeseq{k-L+1}{k}]} + K (z - H m_{k+1 | k}^{[k]}) \end{bmatrix}, \\
	P_{k+1|k+1}^{[\timeseq{\tb}{k-L}]} &= P_{k+1|k}^{[\timeseq{\tb}{k-L}]} , \\
	P_{k+1|k+1}^{[\timeseq{k-L+1}{k}]} &= P_{k+1|k}^{[\timeseq{k-L+1}{k}]} - KHP_{k+1|k}^{[k,\timeseq{k-L+1}{k}]}, \label{eq:updated_L_latest_covariance_matrix}  \\
	K &= P_{k+1|k}^{[\timeseq{k-L+1}{k},k]}H^{\tp}(HP_{k+1|k}^{[k,k]}H^{\tp}+R)^{-1}.
\end{align}%
\label{eq:Gaussian_L_latest_update}%
\end{subequations}%

\subsection{Discussion}

In all three state sequence predictions, \eqref{eq:Gaussian_prediction}, \eqref{eq:information_prediction} and \eqref{eq:Gaussian_L_latest_prediction}, the mean/covariance, or information vector/matrix, are augmented to account for the additional time step that, following the prediction, is represented by the trajectory. The moment form update~\eqref{eq:Gaussian_update} and the $L$-scan update~\eqref{eq:Gaussian_L_latest_update} affect all states in the state sequence, and the $L$ latest states, respectively. In comparison, the information form update~\eqref{eq:information_update} only affects the parts of the information vector and information matrix corresponding to the single most recent state.


For the information form, a key result is that the bottom left and top right corners of the predicted information matrix \eqref{eq:predicted_information_matrix} are exactly zero, and that the update \eqref{eq:updated_information_matrix} only affects the part of the information matrix related to the current state. This means that the information matrix is sparse \cite{Koch11}.


In theory, when using the Gaussian information form, computing the weights \eqref{eq:DetUpdateW} and \eqref{eq:PoisUpdateW} as well as an estimate $\hat{\stseq}_{\hat{\tb}:\hat{\td}}$ involves the inverse of the information matrix. However, it is not necessary to compute the inverse in practice. Instead, multiplications with the inverse information matrix are solved efficiently as sparse, symmetric, positive-definite, linear system of equations. Utilising the sparseness makes the computations significantly faster but the computational cost increases with the trajectory length. An alternative is to compute, in parallel to the information vector and matrix, the mean and covariance for the current state, making them available for computation of the weights \eqref{eq:DetUpdateW} and \eqref{eq:PoisUpdateW}. 

For the $L$-scan approximation, 
the appropriate choice of $L$ depends on the linear and Gaussian models \eqref{eq:linear_gaussian_model}, and how well the cross-covariances $P_{k|k}^{[k,k']}$ can be approximated by all-zero matrices $\zeromatrix_{n_x \times n_x}$ for time steps $k' < k-L$. The larger $L$ is, the more accurate the resulting tracking filter is, at the price of an increased computational cost. One implementation alternative that is computationally cheap and has low memory requirements is to perform a fixed-lag \textsc{rts}-smoothing for selected Bernoulli densities $f_{k|k}^{i,\assoc^i}(\traj)$ (keeping multiple hypotheses between scans) when it is deemed necessary, e.g., to perform trajectory estimation \cite{XiaGSGFW:jaifMultiScanPMBMtrackers}. In addition, for \tpmbm implementation for the set of all trajectories, we do not update single trajectory hypotheses with a small probability $P(\epsilon=k)$ of being present at the current time step $k$.

Lastly, as with \mht, \dglmb and \pmbm filters, the number of global hypotheses in the \tpmbm filters grows fast in time, so we use pruning methods for computational tractability. The \pmbm posterior density \eqref{eq:PMBMdensityDefinition} remains a \pmbm, keeping the symmetry of the density over the elements of the set, if we 1) prune {\mb}s with small global hypothesis weights (and then re-normalise remaining weights), 2) prune Bernoulli components with a small probability of existence, and 3) prune \ppp intensity components with small weights, see details in \cite{GranstromFS:2016_PMBMETT,Angel20}. That is, if pruning keeps the set of global hypotheses $\widetilde{\assocspace}_{k|k'}\subseteq \assocspace_{k|k'}$ the \pmbm is of the form \eqref{eq:PMBMdensityDefinition1} with an \mbm density
\begin{equation}
f_{k|k'}^{\rm d}(\settraj_{k}^{\rm d}) \propto \sum_{\assoc_{k|k'}\in\widetilde{\assocspace}_{k|k'}} w_{k|k'}^{\assoc} \sum_{\uplus_{i\in\trackTable_{k|k'}} \settraj_{k}^{i} = \settraj_{k}^{\rm d}} \prod_{i\in\trackTable_{k|k'}} f_{k|k'}^{i,\assoc^i}(\settraj_{k}^i).
\end{equation}

In \pmbm pruning, keeping the {\mb}s with largest weights minimises an upper bound of the mean absolute error \cite[Sec. V.D]{GranstromFS:2016_PMBMETT}. Standard methods can then be used to select local and global hypotheses with non-negligible weights. For example, we can use gating \cite{Challa_book11}, single-scan data association solvers such as the Hungarian method and Jonker-Volgenant-Castanon algorithm \cite{Crouse16}, $M$-best single-scan data association methods, such as Murty's algorithm \cite{Murty:1968}, single-scan sampling (e.g. Gibbs sampling)  \cite{Morelande09,VoVH:2017, Hue02}, multi-scan data associations via Lagrangian relaxation \cite{Poore97, Deb97, XiaGSGFW:jaifMultiScanPMBMtrackers},  multi-scan generalised Murty's algorithm \cite{Fortunato07} and multi-scan sampling (e.g. Markov chain Monte Carlo or Gibbs sampling) \cite{OhRS:2009, He18,VoV:2019}. 

The above pruning algorithms have different pros and cons related to accuracy and computational complexity, and which one to use can be application dependent. Importantly, all the pruning strategies can be used in the different filtering frameworks to address the data association problem. For example, Murty's algorithm is ensured to obtain the $k$-best global hypotheses arising from a previous global hypothesis with a quartic complexity in the number of measurements $O(m_k^4)$. Gibbs sampling can achieve a linear complexity in the number of measurements  $O(m_k)$, but it is not ensured to obtain the $k$-best global hypotheses \cite{VoVH:2017}.




\section{Simulation study}
\label{sec:simulation_study}


We present results from three scenarios that include interesting tracking challenges: Scenario 1) has many targets and several births and deaths; Scenario 2) has long trajectories; and Scenario 3) includes complicated target coalescence. We compare the following \mtt algorithms:
\begin{enumerate}
	\item \tpmbm filters for the set of all/alive trajectories, implemented using \textit{Murty's algorithm} \cite{Murty:1968} and the Gaussian information form, abbreviated \pmbmif/\tpmbmaif.
	\item \tpmbm filters for the set of all/alive trajectories, implemented using \textit{Murty's algorithm} and $L=1$ scan Gaussian approximation, abbreviated \pmbmLone and \tpmbmalscan.
  	\item Multi-scan \tpmbm filters for the set of all/alive trajectories, implemented using \textit{dual decomposition} \cite{XiaGSGFW:jaifMultiScanPMBMtrackers}, with Gaussian information form, abbreviated \multiscanpmbmif and \mtpmbmaif.
   	\item Multi-scan \tpmbm filters for the set of all/alive trajectories, implemented using \textit{dual decomposition}, with $L=1$ scan Gaussian approximation, abbreviated \multiscanpmbmLone and \mtpmbmalscan.
    \item \dglmb with joint prediction and update steps \cite{VoVH:2017}, implemented using \textit{Murty's algorithm}.
    \item \dglmb with joint prediction and update steps as well as \textit{adaptive birth} \cite{ReuterVVD:2014}, abbreviated \dglmbab, implemented using \textit{Murty's algorithm}.
    \item Trajectory \mbmo \cite{GarciaFernandezSM:2019}, implemented using \textit{multi-scan Gibbs sampling} \cite{VoV:2019}, with $L=1$ scan Gaussian approximation, abbreviated \multiscanmbmolLone. Adding labels to the trajectories in the \multiscanmbmolLone filter does not change the estimated set of trajectories \cite[Sec. IV]{GarciaFernandezSM:2019}. This filter is also referred to as multi-scan \dglmb in \cite{VoV:2019}.
    \item Track-oriented \mht implementation \cite{werthmann1992step} from MathWorks' Sensor Fusion and Tracking Toolbox, abbreviated \tomht. 
\end{enumerate}

Note that \multiscanpmbmif, \mtpmbmaif, \multiscanpmbmLone, \mtpmbmalscan and \multiscanmbmolLone solve a multi-scan data association problem, whereas the rest solve a single-scan data association problem. In addition, \pmbmif, \tpmbmaif, \multiscanpmbmif and \mtpmbmaif perform smoothing-while-filtering and therefore provide smoothed trajectory estimates.

\subsection{Simulation setup}
In all scenarios, a 2D constant velocity motion model \cite[Sec. III]{RongLiJ:2003} 
with Gaussian acceleration standard deviation $\sigma_{v}$ was used. Target measurements were simulated using linear position measurements with Gaussian noise with $R=\diag{[\sigma_{r}^2,\ \sigma_{r}^2 ]}$. Clutter measurements were simulated uniformly distributed in the surveillance area, with average number of false alarms per time step $\mu^{\rm FA} = \int \lambda^{\rm FA}(z) \diff z$. The simulation parameters for all three scenarios are listed in Table~\ref{tab:simulation_parameters}, where $K$ denotes the number of total time steps, and $|\settraj_{\timeseq{0}{K}}|$ is the cardinality of the set of all trajectories at the final time step. The ground truth set of trajectories in each scenario is shown in Figure \ref{fig_groundtruth}. The target cardinalities over time for the three scenarios are shown in Appendix \ref{sec:experiments_appendix}. The birth models are described in the following.

\begin{figure*}[!t]
	\centering
	\subfloat[Scenario 1]{\includegraphics[width = 0.33\textwidth]{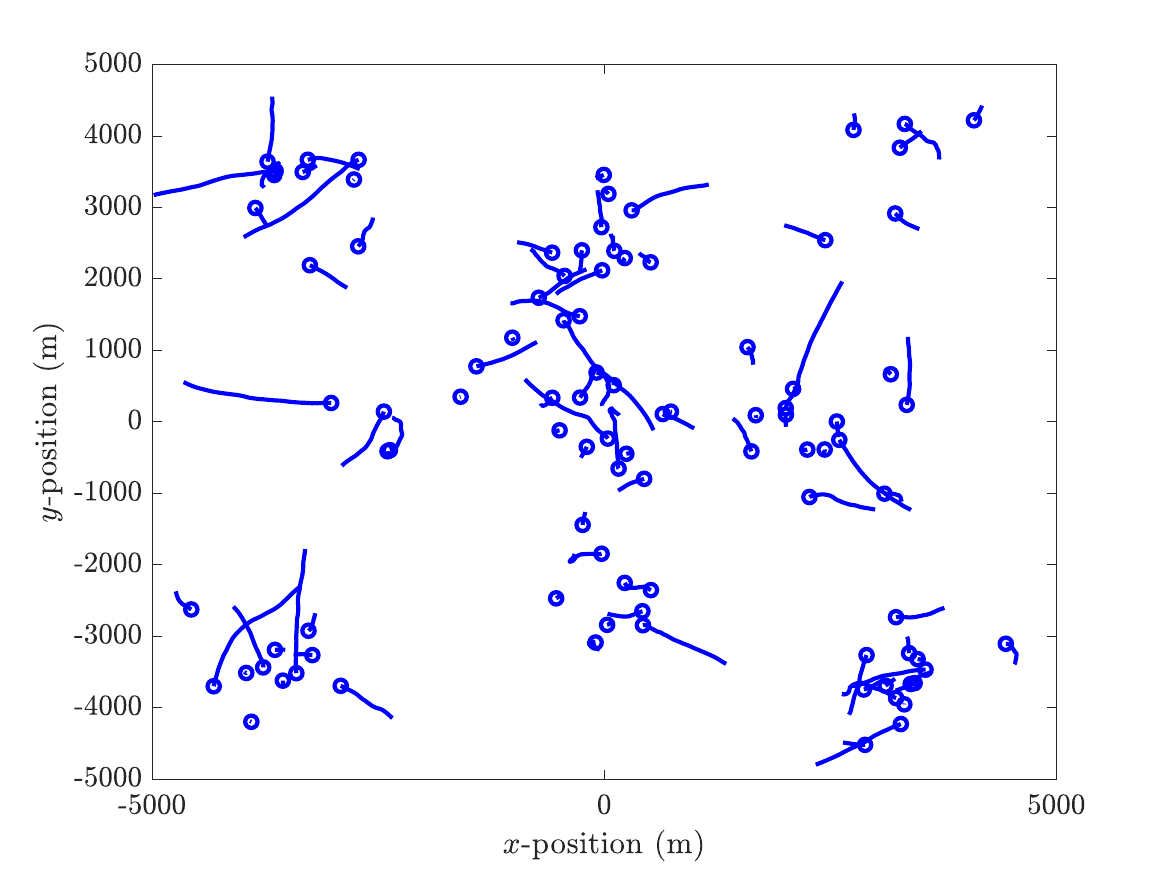}}
	\subfloat[Scenario 2]{\includegraphics[width = 0.33\textwidth]{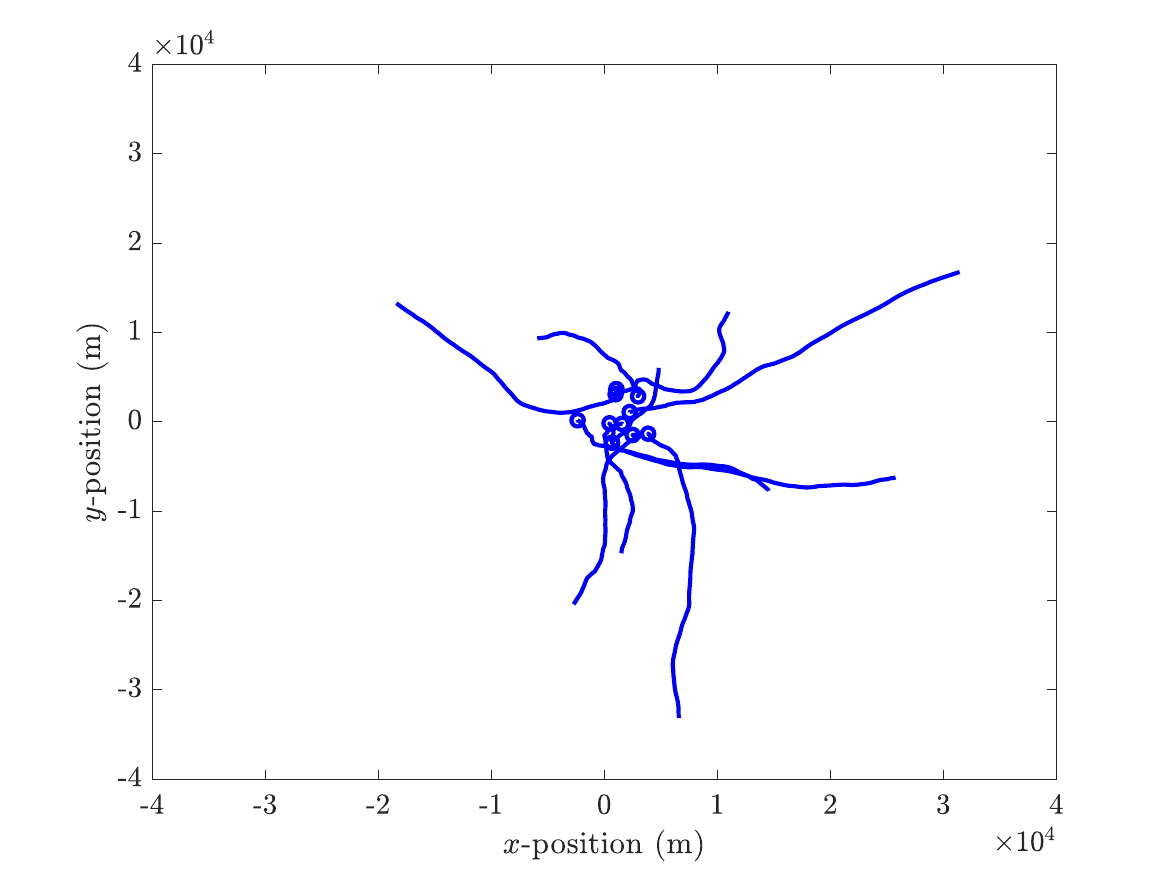}}
	\subfloat[Scenario 3]{\includegraphics[width = 0.33\textwidth]{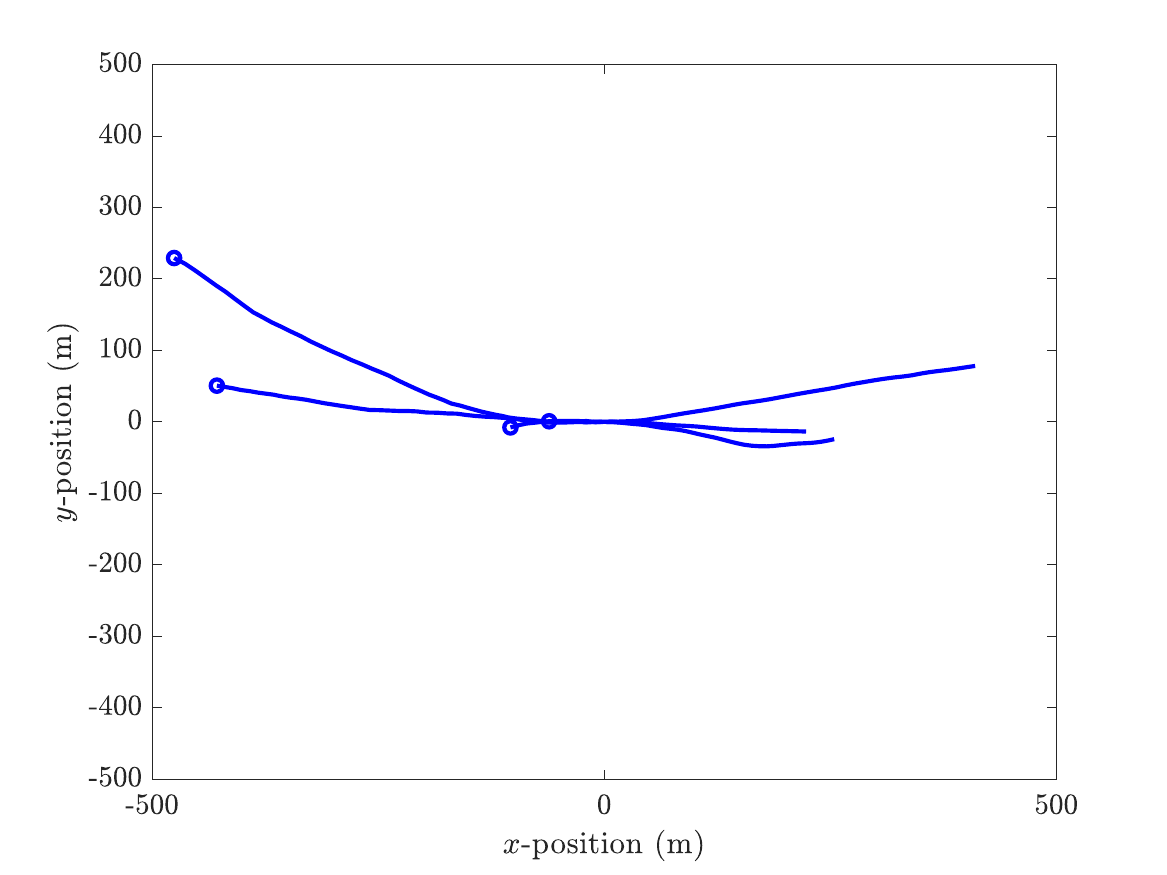}}
	\caption{Ground truth set of trajectories in each of the three scenarios in the simulations. The start time of each trajectory is marked with a circle.}
	\label{fig_groundtruth}
\end{figure*}

\subsubsection{Scenario with a large number of trajectories}
The true trajectories were generated by simulating the models in the surveillance area $[-5\times 10^3,5\times 10^3]\times[-5\times 10^3,5\times 10^3]$. The trajectories vary in length from $2$ to $98$ time steps, with mean/median length $38.9 / 38$.  The birth density had nine components, with mean positions given by the columns in the following matrix,
\begin{align*}
	\begin{bmatrix}
		-\sqrt{2} & \sqrt{2} & -\sqrt{2} & \sqrt{2} & 1 & -1 & 0 & 0 & 0 \\
		-\sqrt{2} & -\sqrt{2} & \sqrt{2} & \sqrt{2} & 0 & 0 & 1 & -1 & 0 
	\end{bmatrix}\cdot \frac{10^4}{4} ,
\end{align*}
zero velocities, and covariances $\diag{[500^2 , 500^2 , 10^2 , 10^2]}$. The expected number of births per time step for each component was set to $1/9$ for all filters. That is, the \mb birth has 9 components, each with a Gaussian single-target density and probability of existence $1/9$. The \ppp birth intensity is a Gaussian mixture, each one with a weight $1/9$. This ensures that the Kullback-Leibler divergence between the \mb birth, used by the \dglmb filters, and the \ppp used by the \pmbm filters \cite{Mahler:2014, Angel20_e}, is minimised. 

\subsubsection{Scenario with long trajectories}
The trajectory birth and death time steps were set deterministically to $10,\ 20, \ldots, 100$ and $990,\ 980, \ldots, 900$, respectively. The trajectories vary in length from $801$ to $981$ time steps. The true trajectories were generated by sampling the initial state from the birth process, and then simulating the motion model in the surveillance area $[-4\times10^4 , 4\times10^4]\times[-4\times10^4 , 4\times10^4]$. The birth density had a single component with mean $[0,0,0,0]^{\tp}$, and covariance $\diag{[4\times 10^6 , 4\times 10^6 , 1 , 1]}$. The expected number of births per time step was set to $0.01$ for the filters.

\subsubsection{Scenario with target coalescence}
As noted in several previous publications \cite{Wil12,GarciaFernandezWGS:2018,xia2023trajectory}, many spatially separated trajectories, or very long trajectories, is not necessarily the most difficult scenario. Therefore, we also simulated a scenario in the surveillance area $[-5\times10^2,5\times10^2]\times[-5\times10^2,5\times10^2]$ with several targets that start separated, move in proximity to each other at the mid-point of the scenario, and then separate again. The birth density had a single component with mean $[0,0,0,0]^{\tp}$, and covariance $\diag{[10^3 , 10^3 , 10^2 , 10^2]}$. The expected number of births per time step was set to $0.04$ for the filters.

\subsubsection{Performance evaluation}
We evaluate the trajectory estimation performance using the generalised optimal subpattern assignment (\gospa) metric for sets of trajectories in \cite{RahmathullahGS16a}, which we refer to as trajectory metric (\trajmetric) in this paper. Specifically, we use its linear programming implementation \cite{RahmathullahGS16a}, which is reviewed in Appendix \ref{app:Trajectory_metric}. It should be noted that the filters based on sets of trajectories directly estimate a set of trajectories, which can be compared with the ground truth set of trajectories with the trajectory metric. In contrast, \dglmb and \tomht require some extra processing. At each time step, \dglmb and \tomht estimate a set of targets with an associated label and track ID, respectively. By linking target state estimates with the same labels or track IDs over time, one can derive an estimated set of trajectories. The labels or track IDs are then discarded after linking.

The parameters of the \trajmetric are: 1-norm base metric, location error cut-off $c=20$, order $p=1$, and switch cost $\gamma=2$. The \trajmetric penalises four types of error: location error (\locationerror), missed target error (\missederror), false target error (\falseerror), and switch error (\switcherror), see \cite{RahmathullahGS16a} for metric definitions and further details. The \trajmetric error is computed at each time step, and the result is normalised by the time step. This enables a comparison of how the trajectory estimation error changes over time.

\begin{table}
	\caption{Simulation parameters}
	\label{tab:simulation_parameters}
	\centering
	\begin{tabular}{c | c c c c c c c}
		Scenario & $K$ & $|\settraj_{\timeseq{0}{K}}|$ & $\sigma_v$ & $\sigma_{r}$ & $P^{\rm S}$ & $P^{\rm D}$ & $\mu^{\rm FA}$ \\
		\hline
		1 & $100$ & $100$ & $1$ & $1$ & 0.99 & 0.9 & $30$ \\
		2 & $1000$ & $10$ & $1$ & $1$ & 0.99 & 0.7 & $60$ \\
		3 & $100$ & $4$ & $0.5$ & $1$ & 0.98 & 0.98 & $10$ \\
	\end{tabular}
\end{table}

\begin{figure*}[!t]
	\centering
	\subfloat[Scenario 1]{\includegraphics[width = 0.33\textwidth]{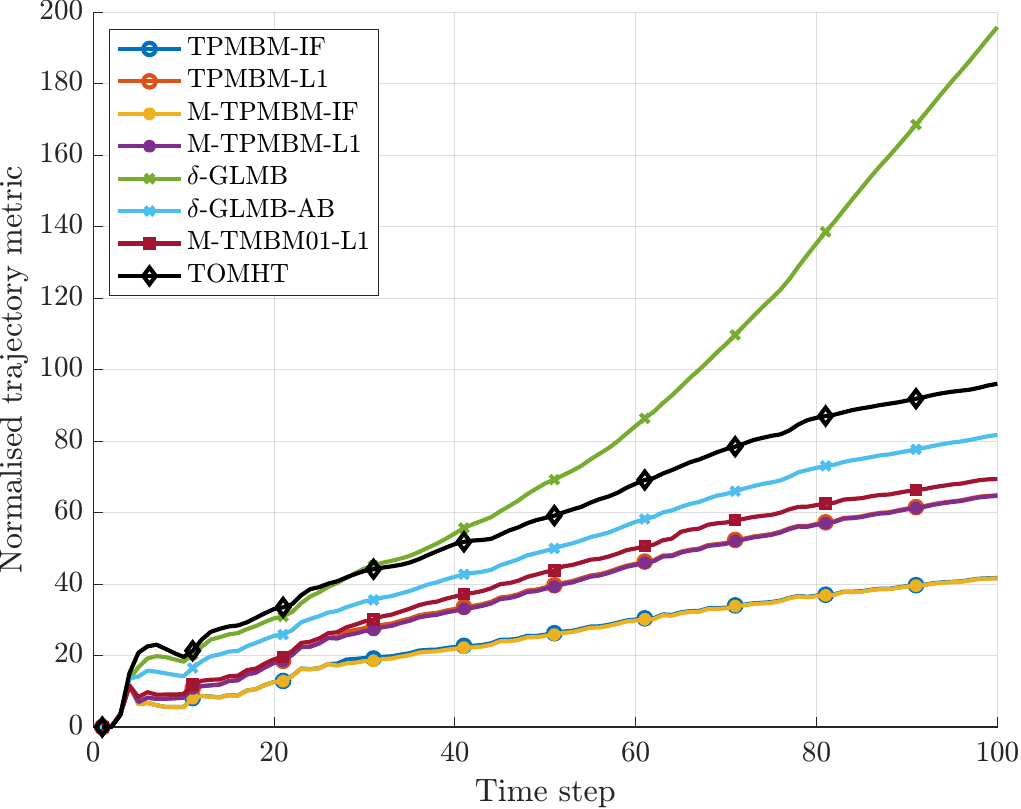}}
	\subfloat[Scenario 2]{\includegraphics[width = 0.33\textwidth]{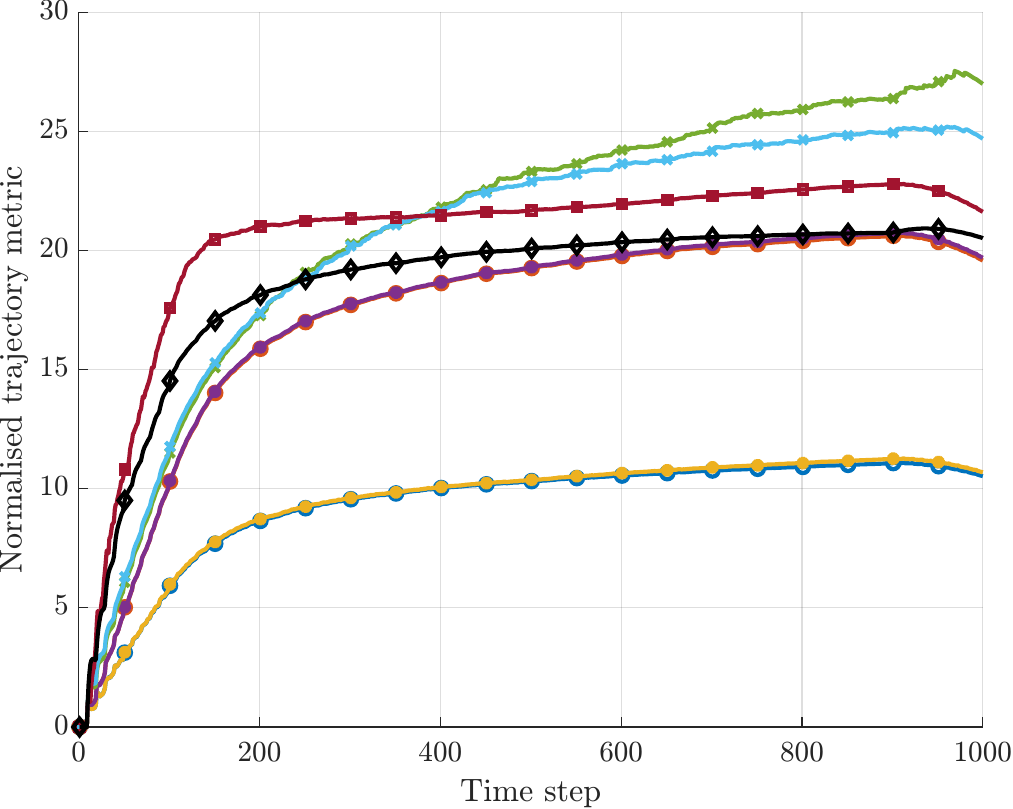}}
	\subfloat[Scenario 3]{\includegraphics[width = 0.33\textwidth]{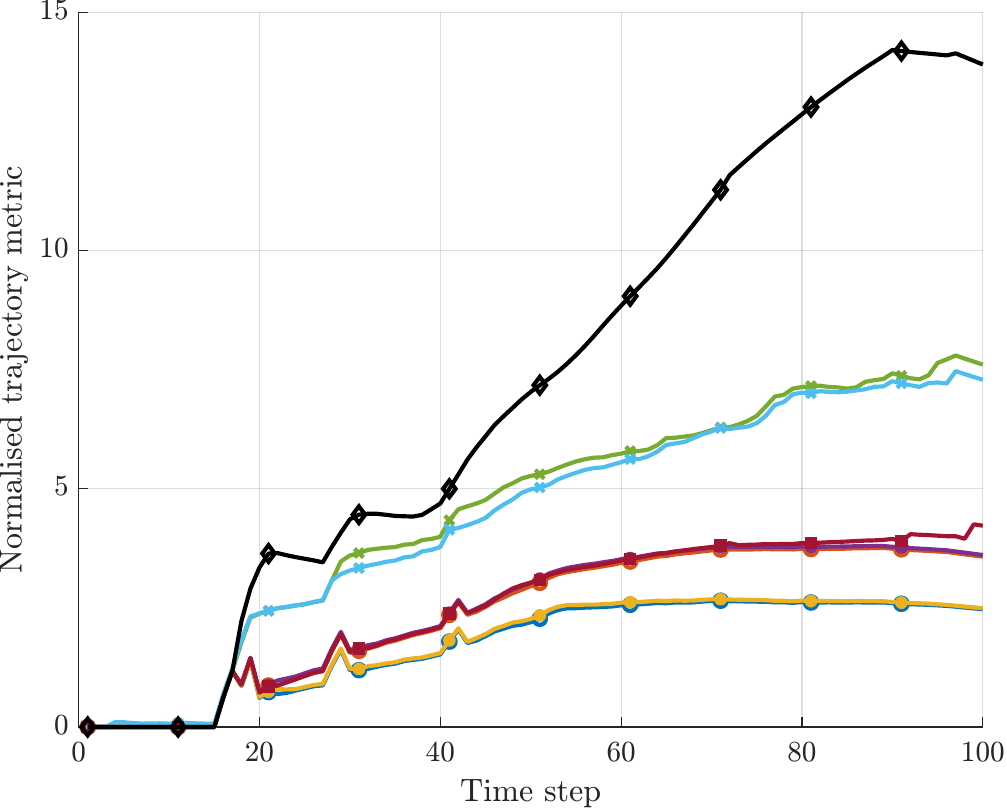}}
	\caption{Trajectory estimation performance for the set of all trajectories in terms of the normalised \trajmetric over time.}
	\label{fig_normalised_traj_metric}
\end{figure*}

\begin{table*}[!t]
	\caption{Trajectory estimation performance for the set of all trajectories in terms of the normalised \trajmetric (summed over all the time steps) and its decomposition}
	\label{tab:normalised_traj_metric}
	\centering
	\resizebox{\textwidth}{!}
	{
		\begin{tabular}{c|ccccc|ccccc|ccccc}
			& \multicolumn{5}{c|}{Scenario 1} & \multicolumn{5}{c|}{Scenario 2} & \multicolumn{5}{c}{Scenario 3} \\ \hline
			Filter & TM   & LE   & ME   & FE   & SE  & TM   & LE   & ME   & FE   & SE  & TM   & LE   & ME   & FE  & SE  \\ \hline
			\pmbmif	 &   2514.6   &   2159.4   &   339.6   &   15.1   &   0.5  &  \underline{9351.0}    &   9118.5   &   231.1   &   0.7   &  0.7   &   \underline{176.2}   &   160.7   &  10.3    &  0   &  5.2   \\
			\pmbmLone	   &  3796.0    &   3362.6   &   378.7   &   54.2   &   0.5  &   17323.7   &   16834.3   &    359.5  &   129.1   &  0.7   &   241.9   &  225.5    &   10.3   &  0   &   6.1  \\
			\multiscanpmbmif	   &   \underline{2490.4}   &   2162.4   &   312.5   &   15.2   &  0.2   &  9443.8    &   9111.6   &   330.9   &   0.6   &  0.7   &  179.8    &   160.7   &  13.5    &  0   &  5.6   \\
			\multiscanpmbmLone	   &   3717.5   &   3373.1   &   287.1   &   56.8   &  0.5   &   17387.0   &  16795.9    &   460.4   &   130.1   &  0.7   &   246.8   &  226.7    &   13.6   &  0.2   &  6.4   \\
			\dglmb	   &   8164.4   &   2841.0   &   3462.8   &   1900.7   &  52.0    &  21046.1    &   15383.4   &  5739.5   &   42.4   & 2.3     &  453.0    &  209.1   &  248.3  & 2.1 &  12.3\\
			\dglmbab	   &   4870.8   &  3158.1    &   1414.7   &   310.0   &   7.8  &   20474.4   &   15650.7   &   4871.5   &   136.0   &   9.3  &   438.9   &   212.5   &   223.9   &  7.4   &  13.0   \\
			\multiscanmbmolLone	   &   4141.8   &   3306.3   &   655.9   &   178.7   &  0.9   &  20545.0    &  16644.3    &   1200.0   &  2696.9    &   3.9  &   250.2   &   225.9   &   15.3   &  2.8   &    6.2 \\
			\tomht & 5861.1 & 2970.7 & 2432.3 & 457.2 & 0.9 & 18608.0 & 16913.6 & 1548.4 & 145.9 & 0.1 & 731.8 & 283.8 & 370.9 & 66.0 & 11.0
		\end{tabular}
	}
\end{table*}

\subsection{Parameters of the filters}
All the implementations used gating with gate probability $0.999$. In all the \tpmbm filters, we prune Bernoulli components with probability of existence smaller than $10^{-5}$ and Gaussian components in the \ppp intensity with weight smaller than $10^{-5}$. In addition, we do not update single trajectory hypotheses whose probability of being alive at current time step is smaller than $10^{-5}$. In \pmbmif, \pmbmLone and the three \dglmb filters, we prune global hypotheses with weight smaller than $10^{-5}$, and their number is capped to $10^3$. In \pmbmif, \pmbmLone, \dglmb and \dglmbab, for global hypothesis $a$ in the prediction step, the $M=\lceil  10^3w^a \rceil$ best global hypotheses are found using Murty's algorithm. 

In \multiscanpmbmif and \multiscanpmbmLone, we consider the data association problem over $3$ scans of measurements ($N=2$ scan pruning \cite{blackman2004multiple}), and the resulting 4D assignment problem is solved using dual decomposition with convergence threshold $0.01$. Moreover, we limit the number of single trajectory hypotheses from the same Bernoulli component to 30. In \multiscanmbmolLone, we also consider the 3-scan data association problem, and the number of iterations in multi-scan Gibbs sampling is set to $\lceil  3\times10^4w^a \rceil$ for global hypothesis $a$ in the prediction step. In \dglmbab, parameter $r_{B,max}$ of adaptive birth \cite{ReuterVVD:2014} is set to the expected number of births per time step in each scenario. For the \dglmbab user-defined single target density, the mean state is the position indicated by the measurement with zero velocity, and the covariance matrix is $\diag{[10^2 , 10^2 , 10^2 , 10^2]}$. For extracting the estimates, \tpmbm filters use the estimator described in Section \ref{sec_estimator} with $r^e = 0.5$, and \dglmb filters use the estimator in \cite{VoVH:2017}. 

As for \tomht, we set the maximum number of hypotheses to 100 and the maximum number of track branches to 50. In addition, we use $N$-scan pruning with maximum 3 scans and the minimum probability required to keep a branch is $10^{-4}$. The performance of \tomht is sensitive w.r.t. the choice of track confirmation and deletion thresholds which are tailored to each scenario to obtain good performance. Specifically, the confirmation and deletion thresholds  are $(10,-8.2)$ for Scenario 1, $(10,-11.7)$ for Scenario 2, and $(20,-7)$ for Scenario 3.

\subsection{Results for estimating the set of all trajectories}

\begin{table}[!t]
	\caption{Average runtime per time step (in seconds)}
	\label{tab:running_time}
	\centering
	\begin{tabular}{c|ccc}
		Filter & Scenario 1 & Scenario 2 & Scenario 3 \\ \hline
		\pmbmif	   &     0.25       &     0.22       &      0.15      \\
		\pmbmLone	   &     0.24       &     0.21       &      0.15      \\
		\multiscanpmbmif	   &     0.51       &     0.23       &      0.03      \\
		\multiscanpmbmLone	   &     0.50       &     0.22       &      0.02      \\
		\tpmbmaif &     0.14       &     0.06       &      0.09      \\
		\tpmbmalscan &     0.14       &     0.05       &      0.08      \\
		\mtpmbmaif &     0.43       &     0.17       &      0.02      \\
		\mtpmbmalscan &     0.42       &     0.15       &      0.02      \\
		\dglmb	   	   &     1.91       &     0.11       &      0.09      \\
		\dglmbab	   &     6.35       &     2.13       &      0.98      \\
		\multiscanmbmolLone  &     11.89       &    2.13        &     1.44      \\
		\tomht & 0.63 & 0.29 & 0.48
	\end{tabular}
\end{table}

We analyse the performance of the algorithms to track the set of all trajectories using Monte Carlo simulation with 100 runs. The trajectory estimation performance is evaluated in terms of the normalised \trajmetric error over time, shown in Fig. \ref{fig_normalised_traj_metric}. In addition, the \trajmetric error and its decomposition into location error (\locationerror), missed target error (\missederror), false target error (\falseerror), and switch error (\switcherror) summed over all time steps and normalised over the time window is given in Table \ref{tab:normalised_traj_metric}. The average runtime per time step with a \matlab implementation on an Apple M2 Pro is provided in Table \ref{tab:running_time}, which also includes the runtimes of the filters that track the set of alive trajectories, whose results will be analysed in the next section. In addition,  the number of targets at each time step and the \trajmetric evaluated at the final time step are given in Appendix \ref{sec:experiments_appendix}.

The results show that for all three scenarios, the \tpmbm filters outperform the two \dglmb filters, \multiscanmbmolLone and \tomht. In particular, the two \tpmbm filters with Gaussian information form outperform their counterparts with Gaussian $L=1$ scan approximation by providing smoothed trajectory estimates, with slightly increased runtime. \tpmbm filters are generally faster than  \dglmb filter implementations, mainly due to a more efficient structure of the global hypotheses \cite{GarciaFernandezWGS:2018}. \tpmbm filters are also faster than the \tomht implementation. 

\tomht uses a track management algorithm that does not fully take into account the target birth and death model information. In Scenarios 1 and 2, \tomht performs better than \dglmb, but worse than the \tpmbm filters. In Scenario 3, which deals with a challenging data association problem, \tomht performs worse than the rest of the filters, especially when targets get in close proximity.

Compared to \tpmbm filters using Murty's algorithm, multi-scan \tpmbm filters using dual decomposition present similar estimation performance in Scenario 1 and 3, worse performance in Scenario 2, and much lower runtime in Scenario 3. By using an adaptive birth model, the performance of \dglmb is improved in all three scenarios, in particular Scenario 1, at the price of significantly increased runtime. \multiscanmbmolLone outperforms \dglmbab by 1) keeping the full trajectory information and 2) solving a more complex data association problem using Gibbs sampling, but it also has the highest runtime. Moreover, we can observe that \multiscanmbmolLone has difficulty initiating new trajectories in Scenario 2, where the expected number of births per time step is low.

In general, Scenario 1 and 2 illustrate that the \tpmbm filters are not limited to handling a high number of relatively short trajectories, but are capable of handling long trajectories. Scenario 3 is challenging due to the complex data association problem when the targets are in proximity for several consecutive time steps, and the \tpmbm filters produce reasonable trajectory estimates. However, the \dglmb filters have a higher  \missederror and \falseerror, and are susceptible to unrealistic track switching given the motion model, see \cite[Fig. 3]{GarciaFernandezSM:2019} and \cite[Fig. 2]{GranstromSXGFW:2018}.

\subsection{Results for estimating the set of alive trajectories}

In this section, we use the same scenarios and parameters as in the previous section, but we analyse the performance to estimate the set of alive trajectories at each time step. In this section, we therefore also include the considered \tpmbm filters that directly approximate the posterior of the set of alive trajectories, namely, \tpmbmaif, \tpmbmalscan, \mtpmbmaif and \mtpmbmalscan. In addition, to estimate the set of alive trajectories we have also tested the previously considered filters, including the \tpmbm filters that approximate the posterior over the set of all trajectories. For these filters, the estimated set of alive trajectories is obtained by extracting the alive trajectories from the estimated set of all trajectories. The average runtime per time step was shown in Table \ref{tab:running_time}, where we can see that the \tpmbm filters for alive trajectories are faster than the \tpmbm filters for all trajectories.

The normalised \trajmetric error over time for the set of alive trajectories is shown in Fig. \ref{fig_normalised_traj_metric_alive}. Also, the \trajmetric error and its decomposition across all time steps is provided in Table \ref{tab:normalised_traj_metric_alive}. We can draw the similar main conclusion as before. \tpmbm filters generally work better than the other algorithms. Specifically, the \multiscanpmbmif filter is the one that provides the most accurate trajectories in the three scenarios.  The \tpmbm filters that directly estimate alive trajectories have a very similar performance to the corresponding \tpmbm filters that track all trajectories but only report the alive trajectories.

\begin{figure*}[!t]
	\centering
	\subfloat[Scenario 1]{\includegraphics[width = 0.33\textwidth]{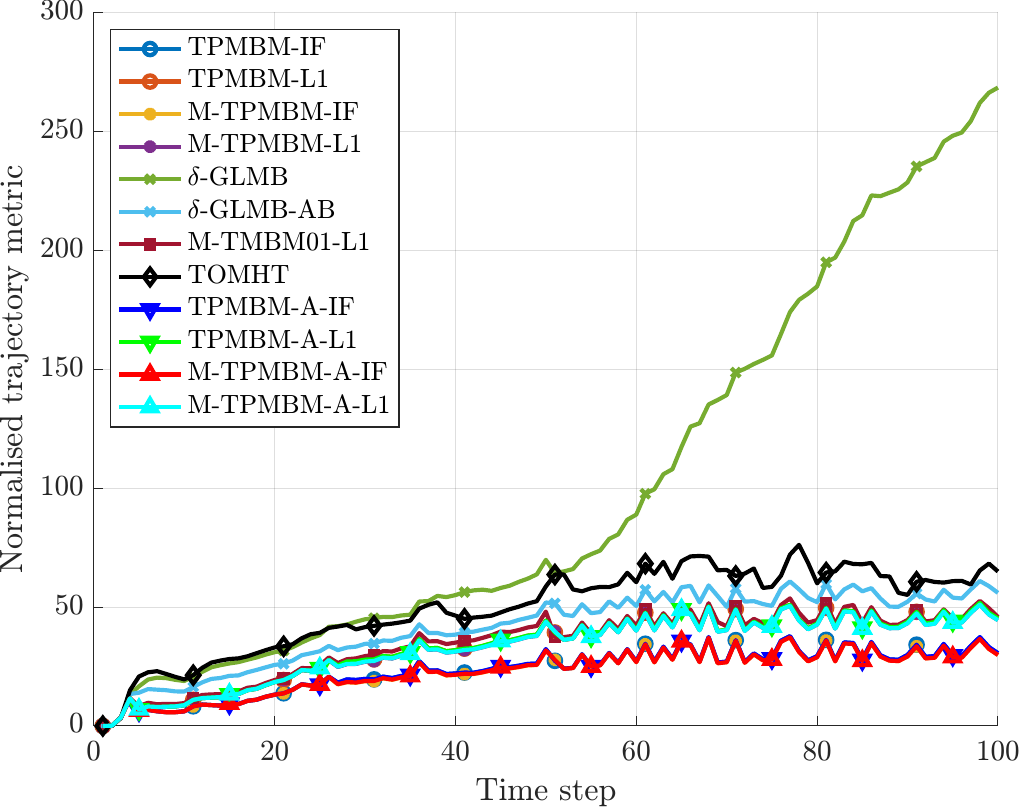}}
	\subfloat[Scenario 2]{\includegraphics[width = 0.33\textwidth]{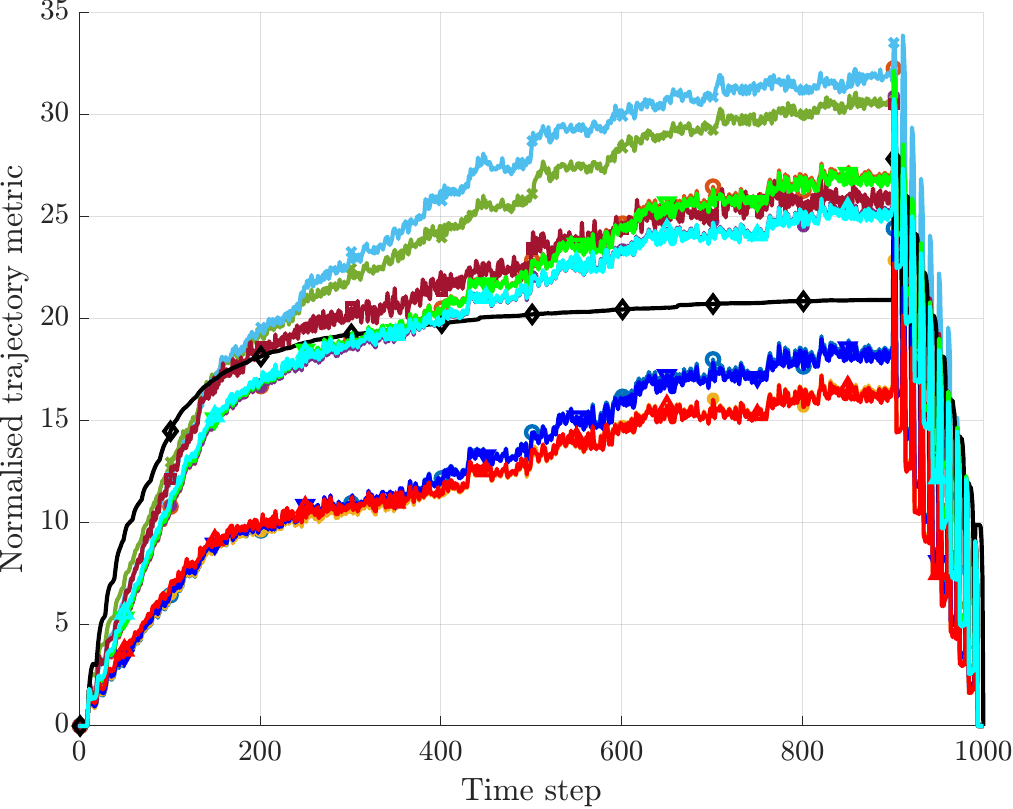}}
	\subfloat[Scenario 3]{\includegraphics[width = 0.33\textwidth]{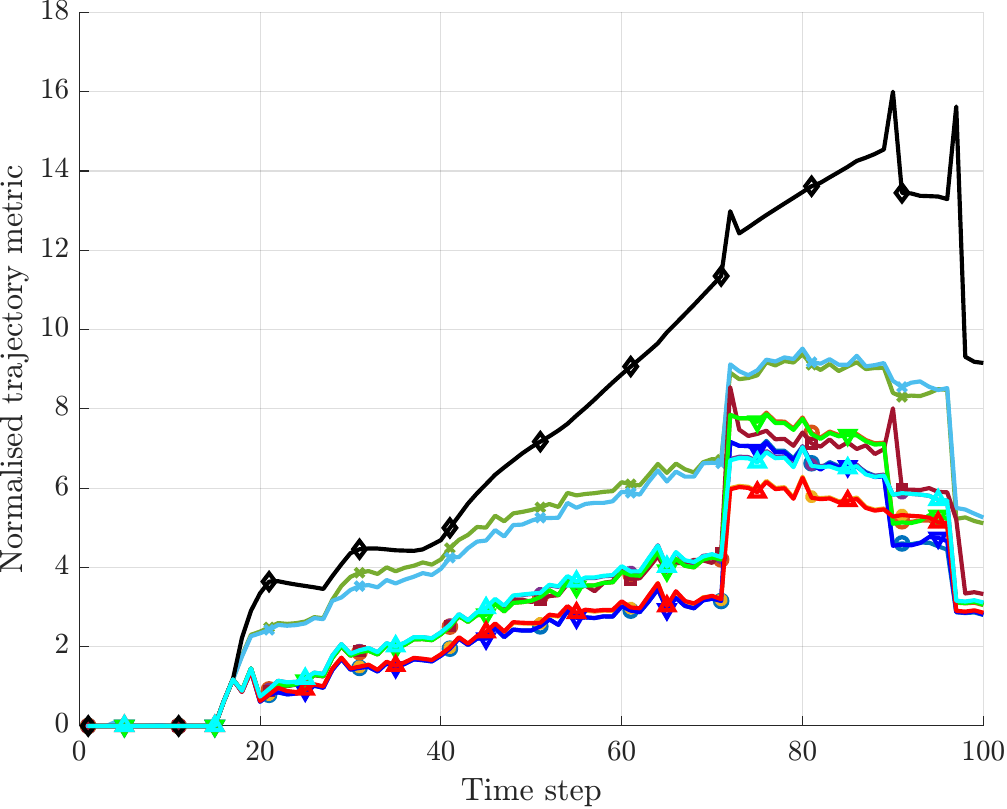}}
	\caption{Trajectory estimation performance for the set of alive trajectories in terms of the normalised \trajmetric over time}
	\label{fig_normalised_traj_metric_alive}
\end{figure*}

\begin{table*}[!t]
	\caption{Trajectory estimation performance for the set of alive trajectories in terms of the normalised \trajmetric \\ (summed over all the time steps) and its decomposition}
	\label{tab:normalised_traj_metric_alive}
	\centering
	\resizebox{\textwidth}{!}
	{
		\begin{tabular}{c|ccccc|ccccc|ccccc}
			& \multicolumn{5}{c|}{Scenario 1} & \multicolumn{5}{c|}{Scenario 2} & \multicolumn{5}{c}{Scenario 3} \\ \hline
			Filter & TM   & LE   & ME   & FE   & SE  & TM   & LE   & ME   & FE   & SE  & TM   & LE   & ME   & FE  & SE  \\ \hline
			\pmbmif	 &   2334.7   &   1695.1   &   453.2   &   186.3   &  0   &    12434.5  &   8098.0   &   4111.9   &   224.5   &  0   &   280.7   &  144.1    &   57.7   &  74.4   &  3.8   \\
			\pmbmLone	   &   3329.8   &   2630.7   &  483.0    &   216.1   &  0   &  19472.5    &  14969.9    &   4195.0   &   307.6   &  0   &   331.2   &  194.5    &   57.7   &  74.5   &  4.5   \\
			\multiscanpmbmif	   &   \underline{2297.0}   &   1700.3   &   409.7   &   186.9  &  0   &   \underline{11569.5}   &   8219.0   &   3117.8   &   232.6   &  0   &   \underline{272.4}   &   145.0   &  52.3    &  70.9   &  4.2   \\
			\multiscanpmbmLone	   &   3294.1   &   2634.6   &   441.2   &   218.4   &  0   &    18705.8  &  15176.1    &   3207.5   &   322.3   &  0   &   325.1   &   197.1   &   52.3   &  70.9   &  4.7   \\
			\dglmb	   &   10122.3   &   2188.1   &   2960.1   &   4972.2   &   13.5  &   22381.4   &   13764.7   &   8473.8   &   222.1   &  0   &   499.2   &  190.4    &   244.7   &  69.6   &  8.7  \\
			\dglmbab	   &   4149.9   &   2442.5   &   1427.3   &   284.5   &   0.8  &   23495.4   &   13473.4   &   9889.5   &   218.5   &  0.1   &   491.8   &   196.1   &   225.4   &  75.9   &  9.0   \\
			\multiscanmbmolLone	   &   3483.4   &   2600.7   &   599.2  &   283.5   &  0   &   19873.2   &   14900.5   &   4544.9   &   427.8   &  0   &   336.9   &   198.3   &   53.3 & 80.7   &  4.6   \\
			\tomht & 4963.7 & 2326.9 & 2055.3 & 580.8 & 0.6 & 18353.9 & 15846.5 & 1610.2 & 897.2 & 0 & 728.1 & 252.6 & 305.7 & 170.0 & 8.8 \\
			\tpmbmaif	   &   2334.7   &   1695.2  &   453.2   &   186.3   &  0   &   12408.3   &   8099.3   &   4084.2   &   224.8   &  0   &   280.1   &   143.9   &   58.0 & 74.4   &  3.8   \\
			\tpmbmalscan	   &   3329.5   &   2630.6   &   482.9   &   216.0   &  0   &   19449.8   &   14974.8   &   4167.2   &   307.8   &  0   &   331.2   &   194.3   &   58.0 & 74.4   &  4.5   \\
			\mtpmbmaif	   &   2299.3   &   1700.0   &   412.5   &   186.8   &  0   &   11651.7   &  8199.0 &  3222.0   &   230.7   &  0  &   272.7   &   144.8   &   53.3 & 70.4   &  4.2   \\
			\mtpmbmalscan	   &   3295.6   &   2634.2   &   443.5   &   217.9   &  0   &   18761.4   &   15135.3   &   3308.7   &   317.4   &  0   &   325.1   &   196.6   &   53.3 & 70.4   &  4.7 
		\end{tabular}
	}
\end{table*}


\section{Conclusions}
\label{sec:conclusion}
In this paper we have shown that the \pmbm density for the set of trajectories is conjugate for \mtt with standard point target models. We have presented two \tpmbm filtering recursions for the set of trajectories. The benefits of the proposed \mtt filters compared to relevant \mtt algorithms is shown via simulations. Further, we have showed that, regardless of what time window we consider, the multi-trajectory posterior is \pmbm, a result analogous to the Gaussian distribution in linear/Gaussian systems.

An interesting line of future work is the development of algorithms that make use of the \tpmbm time marginalisation theorem (Theorem \ref{thm:all_to_trajs_in_a2g_alive_in_e2z}) and the \tpmbm for arbitrary time interval theorem (Theorem \ref{thm:density_traj_alpha2gamma_alive_eta2zeta_meas_xi2chi}) to answer certain types of trajectory-related questions in \mtt.

\ifCLASSOPTIONcaptionsoff
  \newpage
\fi



%

\bibliographystyle{IEEEtran}
\bibliography{references,traj}

%









 
\appendices
\begin{center}\LARGE\textbf{SUPPLEMENTAL MATERIAL}\end{center}

\section{Proofs of Theorem~\ref{thm:all_to_trajs_in_a2g_alive_in_e2z} and Theorem~\ref{thm:density_traj_alpha2gamma_alive_eta2zeta_meas_xi2chi}} 
\label{app:pmbm_relations_theorem_proofs}
In this appendix, we first prove Theorem~\ref{thm:all_to_trajs_in_a2g_alive_in_e2z}, and then we prove Theorem~\ref{thm:density_traj_alpha2gamma_alive_eta2zeta_meas_xi2chi}.

\subsection{Preliminaries}
For trajectories $\traj\in\trajStateSpace{\timeseq{0}{k}}$ we define the function
\begin{subequations}
\begin{align}
	\tau_{\timeseq{\alpha}{\gamma}}^{\timeseq{\eta}{\zeta}}(\traj) = &
		\begin{cases}
			\left\{ \left( b,e,\stseq_{\timeseq{b}{e}} \right) \right\} & \timeset{\tb}{\td} \cap \timeset{\eta}{\zeta} \neq \emptyset, \\
			\emptyset & \text{otherwise,}
		\end{cases} \\
	b = & \max(\tb,\alpha), \quad e = \min(\td,\gamma).
\end{align}
\end{subequations}
where $0\leq\alpha\leq\eta\leq\zeta\leq\gamma\leq k$.
Note that a function corresponding to $\tau_{\timeseq{k}{k}}^{\timeseq{k}{k}}(\traj)$ was defined in \cite[Sec. II.A]{GarciaFernandezSM:2019}. For set inputs, like in \cite[Sec. II.A]{GarciaFernandezSM:2019}, we have
\begin{align}
	\tau_{\timeseq{\alpha}{\gamma}}^{\timeseq{\eta}{\zeta}}(\settraj) & = \begin{cases} \bigcup_{\traj\in\settraj} \tau_{\timeseq{\alpha}{\gamma}}^{\timeseq{\eta}{\zeta}}(\traj) & \settraj \neq \emptyset \\
	\emptyset & \settraj = \emptyset
	\end{cases}
\end{align}%
For $\tau_{\timeseq{\alpha}{\gamma}}^{\timeseq{\eta}{\zeta}}(\settraj)$ it holds that $\tau_{\timeseq{\alpha}{\gamma}}^{\timeseq{\eta}{\zeta}}(\settraj^1 \cup \settraj^2) = \tau_{\timeseq{\alpha}{\gamma}}^{\timeseq{\eta}{\zeta}}(\settraj^1) \cup \tau_{\timeseq{\alpha}{\gamma}}^{\timeseq{\eta}{\zeta}}(\settraj^2) $.

The multi-target Dirac delta is defined as \cite{Mahler:2014}
\begin{align}
	\delta_{\setX'}(\setX) \triangleq 
		\begin{cases}
			0 & |\setX| \neq |\setX'| \\
			1 & \setX = \setX' = \emptyset \\
			\displaystyle\sum_{\sigma\in\Gamma_{n}} \prod_{i=1}^{n} \delta_{\traj_{\sigma_{i}}'}(\traj_i) & 
			\begin{cases}
				\setX = \left\{\traj_{i}\right\}_{i=1}^{n} \\
				\setX' =  \left\{\traj_{i}'\right\}_{i=1}^{n}
			\end{cases}
		\end{cases}
\end{align}
From \cite[Thm. 11]{GarciaFernandezSM:2019} we know that using the multi-target Dirac delta function and $\tau_{\timeseq{\alpha}{\gamma}}^{\timeseq{\eta}{\zeta}}(\settraj)$ we can formulate a transition density from the set $\settraj_{\timeseq{0}{k}}$ to the set $\settraj_{\timeseq{\alpha}{\gamma}}^{\timeseq{\eta}{\zeta}}=\tau_{\timeseq{\alpha}{\gamma}}^{\timeseq{\eta}{\zeta}}(\settraj_{\timeseq{0}{k}})$ as
\begin{align}
	\delta_{\tau_{\timeseq{\alpha}{\gamma}}^{\timeseq{\eta}{\zeta}}(\settraj_{\timeseq{0}{k}})}\left(\settraj_{\timeseq{\alpha}{\gamma}}^{\timeseq{\eta}{\zeta}}\right).
\end{align}
Lemma~\ref{lem:pmbm_function_integration_general} shows how to find the density of $\tau_{\timeseq{\alpha}{\gamma}}^{\timeseq{\eta}{\zeta}}(\settraj)$.

\begin{lemma}\label{lem:pmbm_function_integration_general}
	Suppose $f(\settraj)$ is a \pmbm density parameterised as in Section~\ref{sec:pmbm_trackers}. Then, the density of $\setY=\tau_{\timeseq{\alpha}{\gamma}}^{\timeseq{\eta}{\zeta}}(\settraj)$ is
	\begin{subequations}
	\begin{align}
		g(\setY) & = \int \delta_{\tau_{\timeseq{\alpha}{\gamma}}^{\timeseq{\eta}{\zeta}}(\settraj)}(\setY) f(\settraj) \delta \settraj  \\
		&= \sum_{\substack{\left(\uplus_{i\in\trackTable} \setY^{i} \right) \uplus \setY^{\rm u} = \setY}} \int \delta_{\tau_{\timeseq{\alpha}{\gamma}}^{\timeseq{\eta}{\zeta}}\left( \settraj^{\rm u}\right)}(\setY^{\rm u})  f^{\rm u}(\settraj^{\rm u}) \delta \settraj^{\rm u} \nonumber \\
		& \quad \times \sum_{\assoc\in\assocspace} w_{\assoc} \prod_{i\in\mathbb{T}} \int \delta_{\tau_{\timeseq{\alpha}{\gamma}}^{\timeseq{\eta}{\zeta}}\left( \settraj^i\right)}(\setY^i)  f^{i,a^i}(\settraj^i) \delta \settraj^i 
	\end{align}
	\end{subequations}
\end{lemma}
\begin{proof}
	Follows from \cite[Lem. 2]{Williams:2015} and \cite[Thm. 11]{GarciaFernandezSM:2019}.
\end{proof}
Note that in Lemma~\ref{lem:pmbm_function_integration_general} we have integrals with the transition density and \ppp densities and Bernoulli densities. The following Lemmas provide the solutions to these integrals.

\begin{lemma}\label{lem:all_Bernoulli_to_Bernoulli_in_time_interval_general}
	Let $f(\setX_{0:k})$ be a Bernoulli density, with parameters $r$ and $f(\traj)$. The trajectory \rfs density
	\begin{align}
		f(\settraj_{\timeseq{\alpha}{\gamma}}^{\timeseq{\eta}{\zeta}}) = \int \delta_{ \tau_{\timeseq{\alpha}{\gamma}}^{\timeseq{\eta}{\zeta}}(\settraj_{\timeseq{0}{k}}) }( \settraj_{\timeseq{\alpha}{\gamma}}^{\timeseq{\eta}{\zeta}} )f(\settraj_{\timeseq{0}{k}})\delta\settraj_{\timeseq{0}{k}}
	\end{align}
	is a trajectory Bernoulli density with parameters
	\begin{subequations}
	\begin{align}
		\tilde{r} = & r \Pr\left( \timeset{\tb}{\td}\cap \timeset{\eta}{\zeta} \neq\emptyset \right)  \\
		\tilde{f}(\traj) = & \frac{\displaystyle\sum_{\substack{\tb,\td : \\ \timeset{\tb}{\td}\cap \timeset{\eta}{\zeta} \neq\emptyset}} \int p(\stseq_{\timeseq{\tb}{\td}} | \tb,\td) P(\tb,\td) \diff \stseq_{\timeseq{\tb}{\td}\backslash\timeseq{b}{e}}}{ \Pr\left( \timeset{\tb}{\td}\cap \timeset{\eta}{\zeta} \neq\emptyset \right)}
	\end{align}
	\end{subequations}
	where
	\begin{align}
		\Pr\left(\timeset{\tb}{\td}\cap \timeset{\eta}{\zeta} \neq\emptyset \right) = \sum_{\substack{\tb,\td : \\ \timeset{\tb}{\td}\cap \timeset{\eta}{\zeta} \neq\emptyset}}  P(\tb,\td).
	\end{align}
	Specifically, for
	\begin{align}
		 f(\traj) = \sum_{j\in\indexSetJ} \weight^{j} f^{j} (\traj ; \trajdensityparams^j),
	\end{align}
	we get
	\begin{subequations}
	\begin{align}
		& \tilde{r} = r \sum_{j\in\indexSetJ^{\timeseq{\eta}{\zeta}}} \weight^{j}  \\
		& \tilde{f}(\traj) = \frac{\sum_{j\in\indexSetJ^{\timeseq{\eta}{\zeta}}} \weight^{j}  \tilde{f}^{j} (\traj ; \tilde{\trajdensityparams}^j) }{\sum_{j\in\indexSetJ^{\timeseq{\eta}{\zeta}}} \weight^{j}} \\
		& \tilde{\trajdensityparams}^j = \left(\tilde{b}^j, \tilde{e}^j, \tilde{p}^j(\stseq_{\timeseq{\tilde{b}^j}{\tilde{e}^j}})  \right)  \\
		& \tilde{b}^j = \max(b^j,\alpha), \quad \tilde{e}^j = \min(e^j,\gamma) \\
		& \tilde{p}^j(\stseq_{\timeseq{\tilde{b}^j}{\tilde{e}^j}})  = \int p^j(\stseq_{\timeseq{b^j}{e^j}}) \diff \stseq_{\timeseq{b^j}{e^j}\backslash\timeseq{\tilde{b}^j}{\tilde{e}^j}}
	\end{align}
	\end{subequations}
	where $\indexSetJ^{\timeseq{\eta}{\zeta}} = \left\{ j \ : \ \timeset{b^j}{e^j}\cap \timeset{\eta}{\zeta} \neq\emptyset \right\}$.
\end{lemma}
\begin{proof}
See \eqref{eq:proof_all_Bernoulli_to_Bernoulli_in_time_interval_general}.
\begin{figure*}
	\rule{\textwidth}{1pt}
	\begin{subequations}
	\begin{align}
		f&\left(\settraj_{\timeseq{\alpha}{\gamma}}^{\timeseq{\eta}{\zeta}}\right) = \int \delta_{ \tau_{\timeseq{\alpha}{\gamma}}^{\timeseq{\eta}{\zeta}}(\settraj_{\timeseq{0}{k}}) }\left(\settraj_{\timeseq{\alpha}{\gamma}}^{\timeseq{\eta}{\zeta}}\right)f(\settraj_{\timeseq{0}{k}})\delta\settraj_{\timeseq{0}{k}} = \delta_{ \tau_{\timeseq{\alpha}{\gamma}}^{\timeseq{\eta}{\zeta}}(\emptyset) }\left(\settraj_{\timeseq{\alpha}{\gamma}}^{\timeseq{\eta}{\zeta}}\right) (1-r) + \int \delta_{ \tau_{\timeseq{\alpha}{\gamma}}^{\timeseq{\eta}{\zeta}}(\{\traj\}) }\left(\settraj_{\timeseq{\alpha}{\gamma}}^{\timeseq{\eta}{\zeta}}\right) r f(\traj)\diff\traj \\
		= & \delta_{\emptyset}\left(\settraj_{\timeseq{\alpha}{\gamma}}^{\timeseq{\eta}{\zeta}}\right) (1-r) + \sum_{\substack{\tb,\td:\\ \timeset{\tb}{\td}\cap\timeset{\eta}{\zeta}=\emptyset}} \int \delta_{ \emptyset }\left(\settraj_{\timeseq{\alpha}{\gamma}}^{\timeseq{\eta}{\zeta}}\right) r f((\stseq_{\timeseq{\tb}{\td}},\tb,\td))\diff\stseq_{\timeseq{\tb}{\td}}  + \sum_{\substack{\tb,\td:\\ \timeset{\tb}{\td}\cap\timeset{\eta}{\zeta}\neq\emptyset}} \int \delta_{ \{ \left( b,e,\stseq_{\timeseq{b}{e}} \right) \} }\left(\settraj_{\timeseq{\alpha}{\gamma}}^{\timeseq{\eta}{\zeta}}\right) r f((\stseq_{\timeseq{\tb}{\td}},\tb,\td))\diff\stseq_{\timeseq{\tb}{\td}} \\
		= & \delta_{\emptyset}\left(\settraj_{\timeseq{\alpha}{\gamma}}^{\timeseq{\eta}{\zeta}}\right) \left(1-r\sum_{\substack{\tb,\td:\\ \timeset{\tb}{\td}\cap\timeset{\eta}{\zeta}\neq\emptyset}}P(\tb,\td)\right) + r \left(\sum_{\substack{\tb,\td:\\ \timeset{\tb}{\td}\cap\timeset{\eta}{\zeta}\neq\emptyset}}P(\tb,\td)\right) \frac{\sum_{\substack{\tb,\td:\\ \timeset{\tb}{\td}\cap\timeset{\eta}{\zeta}\neq\emptyset}} \int \delta_{ \{ \left( b,e,\stseq_{\timeseq{b}{e}} \right) \} }\left(\settraj_{\timeseq{\alpha}{\gamma}}^{\timeseq{\eta}{\zeta}}\right) p(\stseq_{\timeseq{\tb}{\td}}|\tb,\td)P(\tb,\td)\diff\stseq_{\timeseq{\tb}{\td}}}{\sum_{\substack{\tb,\td:\\ \timeset{\tb}{\td}\cap\timeset{\eta}{\zeta}\neq\emptyset}}P(\tb,\td)} \\
		= & \begin{cases}
			1-r\sum_{\substack{\tb,\td:\\ \timeset{\tb}{\td}\cap\timeset{\eta}{\zeta}\neq\emptyset}}P(\tb,\td) & \text{if } \settraj_{\timeseq{\alpha}{\gamma}}^{\timeseq{\eta}{\zeta}} = \emptyset \\
			r \left(\sum_{\substack{\tb,\td:\\ \timeset{\tb}{\td}\cap\timeset{\eta}{\zeta}\neq\emptyset}}P(\tb,\td)\right) \frac{\sum_{\substack{\tb,\td:\\ \timeset{\tb}{\td}\cap\timeset{\eta}{\zeta}\neq\emptyset}} \int p(\stseq_{\timeseq{\tb}{\td}}|\tb,\td) \diff \stseq_{\timeseq{\tb}{\td}\backslash\timeseq{b}{e}}P(\tb,\td)}{\sum_{\substack{\tb,\td:\\ \timeset{\tb}{\td}\cap\timeset{\eta}{\zeta}\neq\emptyset}}P(\tb,\td)} & \text{if } \settraj_{\timeseq{\alpha}{\gamma}}^{\timeseq{\eta}{\zeta}} = \{\traj\}
		\end{cases}
	\end{align}
	\label{eq:proof_all_Bernoulli_to_Bernoulli_in_time_interval_general}
	\end{subequations}
	\rule{\textwidth}{1pt}
\end{figure*}
\end{proof}

\begin{lemma}\label{lem:all_Poisson_to_Poisson_in_time_interval_general}
	Let $f(\setX_{0:k})$ be a trajectory \ppp density, with intensity $\lambda(\traj)= \mu p(\stseq_{\timeseq{\tb}{\td}} | \tb,\td) P(\tb,\td)$. The trajectory \rfs density
	\begin{align}
		f(\settraj_{\timeseq{\alpha}{\gamma}}^{\timeseq{\eta}{\zeta}}) = \int \delta_{ \tau_{\timeseq{\alpha}{\gamma}}^{\timeseq{\eta}{\zeta}}(\settraj_{\timeseq{0}{k}}) }( \settraj_{\timeseq{\alpha}{\gamma}}^{\timeseq{\eta}{\zeta}} )f(\settraj_{\timeseq{0}{k}})\delta\settraj_{\timeseq{0}{k}}
	\end{align}
	is a trajectory \ppp density with intensity
	\begin{align}
		\tilde{\lambda}(\traj) = & \sum_{\substack{\tb,\td : \\ \timeset{\tb}{\td}\cap \timeset{\eta}{\zeta} \neq\emptyset}} \int \mu p(\stseq_{\timeseq{\tb}{\td}} | \tb,\td) P(\tb,\td) \diff \stseq_{\timeseq{\tb}{\td} \backslash \timeseq{b}{e}}
	\end{align}
	Specifically, for
	\begin{align}
		 \lambda(\traj) = \sum_{j\in\indexSetJ} \weight^{j} f^{j} (\traj ; \trajdensityparams^j),
	\end{align}
	we get
	\begin{subequations}
	\begin{align}
		&\tilde{\lambda}(\traj) = \sum_{j\in\indexSetJ^{\timeseq{\eta}{\zeta}}} \weight^{j} \tilde{f}^{j} (\traj ; \tilde{\trajdensityparams}^j) \\
		& \tilde{\trajdensityparams}^j = \left(\tilde{b}^j, \tilde{e}^j, \tilde{p}^j(\stseq_{\timeseq{\tilde{b}^j}{\tilde{e}^j}})  \right)  \\
		& \tilde{b}^j = \max(b^j,\alpha), \quad \tilde{e}^j = \min(e^j,\gamma) \\
		& \tilde{p}^j(\stseq_{\timeseq{\tilde{b}^j}{\tilde{e}^j}})  = \int p^j(\stseq_{\timeseq{b^j}{e^j}}) \diff \stseq_{\timeseq{b^j}{e^j}\backslash\timeseq{\tilde{b}^j}{\tilde{e}^j}}
	\end{align}
	\end{subequations}
	where $\indexSetJ^{\timeseq{\eta}{\zeta}} = \left\{ j \ : \ \timeset{b^j}{e^j}\cap \timeset{\eta}{\zeta} \neq\emptyset \right\}$.
\end{lemma}
\begin{proof}
From \cite[p. 99, Rmk. 12]{Mahler:2014} we know that the \ppp $\settraj_{\timeseq{0}{k}}$ with intensity $\lambda(\traj)$ can be divided into two disjoint and independent \ppp subsets
\begin{subequations}
\begin{align}
	\settraj_{\timeseq{0}{k}} = & \settraj_{\timeseq{0}{k}}^{\timeseq{\eta}{\zeta}} \cup \settraj_{\timeseq{0}{k}}^{\timeseq{\not\eta}{\not\zeta}} , \quad \settraj_{\timeseq{0}{k}}^{\timeseq{\eta}{\zeta}} \cap \settraj_{\timeseq{0}{k}}^{\timeseq{\not\eta}{\not\zeta}} =\emptyset \\
	\settraj_{\timeseq{0}{k}}^{\timeseq{\eta}{\zeta}} = & \left\{ \traj\in \settraj_{\timeseq{0}{k}} \ : \ \timeset{\tb}{\td} \cap \timeset{\eta}{\zeta}\neq\emptyset \right\} \\
	\settraj_{\timeseq{0}{k}}^{\timeseq{\not\eta}{\not\zeta}} = & \left\{ \traj\in \settraj_{\timeseq{0}{k}} \ : \ \timeset{\tb}{\td} \cap \timeset{\eta}{\zeta}=\emptyset \right\}
\end{align}
\end{subequations}
with \ppp intensities
\begin{subequations}
\begin{align}
	\lambda^{\timeseq{\eta}{\zeta}}(\traj) &= \mathbf{1}_{\timeset{\tb}{\td} \cap \timeset{\eta}{\zeta}} \lambda(\traj) , \\
	\lambda^{\timeseq{\not\eta}{\not\zeta}}(\traj) &= \left(1- \mathbf{1}_{\timeset{\tb}{\td} \cap \timeset{\eta}{\zeta}}\right)\lambda(\traj),
\end{align}
\end{subequations}
where 
\begin{align}
	\mathbf{1}_{\timeset{\tb}{\td} \cap \timeset{\eta}{\zeta}} = \begin{cases} 1 & \timeset{\tb}{\td} \cap \timeset{\eta}{\zeta} \neq\emptyset \\ 0 & \text{otherwise.} \end{cases}
\end{align}
It is straightforward to verify that $\lambda^{\timeseq{\eta}{\zeta}}(\traj) + \lambda^{\timeseq{\not\eta}{\not\zeta}}(\traj) = \lambda(\traj)$. The rest of the proof follows in \eqref{eq:proof_all_Poisson_to_Poisson_in_time_interval_general}.
\begin{figure*}
\rule{\textwidth}{1pt}
\begin{subequations}
\begin{align}
		f^{\gamma}(\setY) = & \int \delta_{ \tau_{\timeseq{\alpha}{\gamma}}^{\timeseq{\eta}{\zeta}}(\settraj_{\timeseq{0}{k}}) }( \setY )f(\settraj_{\timeseq{0}{k}})\delta\settraj_{\timeseq{0}{k}} \\
		= & \sum_{\substack{\settraj_{\timeseq{0}{k}}^{\timeseq{\eta}{\zeta}},\settraj_{\timeseq{0}{k}}^{\timeseq{\not\eta}{\not\zeta}} \\ \settraj_{\timeseq{0}{k}}^{\timeseq{\eta}{\zeta}} \uplus \settraj_{\timeseq{0}{k}}^{\timeseq{\not\eta}{\not\zeta}} = \setY}} \int \delta_{ \tau_{\timeseq{\alpha}{\gamma}}^{\timeseq{\eta}{\zeta}}(\settraj_{\timeseq{0}{k}}^{\gamma}) }( \settraj_{\timeseq{0}{k}}^{\timeseq{\eta}{\zeta}} )f(\settraj_{\timeseq{0}{k}}^{\gamma})\delta\settraj_{\timeseq{0}{k}}^{\gamma}  \int \delta_{ \tau_{\timeseq{\alpha}{\gamma}}^{\timeseq{\eta}{\zeta}}(\settraj_{\timeseq{0}{k}}^{\not \gamma}) }( \settraj_{\timeseq{0}{k}}^{\timeseq{\not\eta}{\not\zeta}} )f(\settraj_{\timeseq{0}{k}}^{\not \gamma})\delta\settraj_{\timeseq{0}{k}}^{\not \gamma} \\
		= & \int \delta_{ \tau_{\timeseq{\alpha}{\gamma}}^{\timeseq{\eta}{\zeta}}(\settraj_{\timeseq{0}{k}}^{\gamma}) }( \setY )f(\settraj_{\timeseq{0}{k}}^{\gamma})\delta\settraj_{\timeseq{0}{k}}^{\gamma} \\
		= & \delta_{ \emptyset }( \setY )  e^{-\conv{\lambda^{\timeseq{\eta}{\zeta}}}{1}} + \sum_{n=1}^{\infty} \frac{1}{n!}\idotsint \delta_{ \tau_{\timeseq{\alpha}{\gamma}}^{\timeseq{\eta}{\zeta}}(\{\traj^1,\ldots,\traj^n\}) }( \setY ) e^{-\conv{\lambda^{\timeseq{\eta}{\zeta}}}{1}} \prod_{i=1}^{n}\lambda^{\timeseq{\eta}{\zeta}}(\traj^i) \diff\traj^1\cdots\diff\traj^n \\
		= & \delta_{ \emptyset }( \setY )  e^{-\conv{\lambda^{\timeseq{\eta}{\zeta}}}{1}} + \sum_{n=1}^{\infty} \frac{ e^{-\conv{\lambda^{\timeseq{\eta}{\zeta}}}{1}}}{n!} \sum_{\substack{\setY^{1},\ldots,\setY^n : \\ \uplus_{i=1}^{n} \setY^i = \setY}} \prod_{i=1}^{n} \int \delta_{ \tau_{\timeseq{\alpha}{\gamma}}^{\timeseq{\eta}{\zeta}}(\{\traj^i\}) }( \setY^i ) \lambda^{\timeseq{\eta}{\zeta}}(\traj^i)\diff\traj^i \\
		= & \delta_{ \emptyset }( \setY )  e^{-\conv{\lambda^{\timeseq{\eta}{\zeta}}}{1}} +  \sum_{n=1}^{\infty} \frac{ e^{-\conv{\lambda^{\timeseq{\eta}{\zeta}}}{1}}}{n!} \sum_{\substack{\setY^{1},\ldots,\setY^n : \\ \uplus_{i=1}^{n} \setY^i = \setY}} \prod_{i=1}^{n} \sum_{\substack{\tb,\td : \\ \timeset{\tb}{\td} \cap \timeset{\eta}{\zeta} }} \int \delta_{ \{X_{}^i\} }( \setY^i ) \lambda(\traj^i)\diff\stseq_{\timeseq{\tb}{\td}\backslash\timeseq{b}{e}}^i \\
		= & e^{-\conv{\lambda^{\timeseq{\eta}{\zeta}}}{1}} \delta_{ \emptyset }( \setY ) + e^{-\conv{\lambda^{\timeseq{\eta}{\zeta}}}{1}} \sum_{n=1}^{\infty} \delta_{\{Y^1,\ldots,Y^n\}}(\setY ) \prod_{i=1}^{n} \sum_{\substack{\tb,\td : \\ \timeset{\tb}{\td} \cap \timeset{\eta}{\zeta} }} \int \delta_{ \{X_{}^i\} }( \{Y^i\} ) \lambda(\traj^i)\diff\stseq_{\timeseq{\tb}{\td}\backslash\timeseq{b}{e}}^i  \\
		= & e^{- \mu \sum_{\substack{\tb,\td : \timeset{\tb}{\td} \cap \timeset{\eta}{\zeta} }} P(\tb,\td) } \prod_{Y\in\setY} \mu p(Y) \sum_{\substack{\tb,\td :  \timeset{\tb}{\td} \cap \timeset{\eta}{\zeta} }} P(\tb,\td)
\end{align} 
\label{eq:proof_all_Poisson_to_Poisson_in_time_interval_general}
\end{subequations}
\rule{\textwidth}{1pt}
\end{figure*}
\end{proof}

\subsection{Proof of Theorem~\ref{thm:all_to_trajs_in_a2g_alive_in_e2z}}

Given a \pmbm density $f_{}(\settraj_{\timeseq{0}{k}})$ for the set of all trajectories from $0$ to $k$, the density for $\settraj_{\timeseq{\alpha}{\gamma}}^{\timeseq{\eta}{\zeta}}$, where $0\leq\alpha\leq\eta\leq\zeta\leq\gamma\leq k$, is \cite[Thm. 11]{GarciaFernandezSM:2019}
	\begin{align}
		f(\settraj_{\timeseq{\alpha}{\gamma}}^{\timeseq{\eta}{\zeta}}) = \int \delta_{ \tau_{\timeseq{\alpha}{\gamma}}^{\timeseq{\eta}{\zeta}}(\settraj_{\timeseq{0}{k}}) }( \settraj_{\timeseq{\alpha}{\gamma}}^{\timeseq{\eta}{\zeta}} )f(\settraj_{\timeseq{0}{k}})\delta\settraj_{\timeseq{0}{k}}
	\end{align}
Theorem~\ref{thm:all_to_trajs_in_a2g_alive_in_e2z} then follows from Lemmas~\ref{lem:pmbm_function_integration_general}, \ref{lem:all_Bernoulli_to_Bernoulli_in_time_interval_general} and \ref{lem:all_Poisson_to_Poisson_in_time_interval_general}.

\subsection{Proof of Theorem~\ref{thm:density_traj_alpha2gamma_alive_eta2zeta_meas_xi2chi}}
\label{sec:proof_of_theorem}
	It follows from Theorem~\ref{th:PredictionAllTrajectories} and Theorem~\ref{th:Update} that
	\begin{align}
	f_{\timeseq{\min(\alpha,\xi)}{\max(\gamma,\chi)} | \timeseq{\xi}{\chi} }\left(\settraj_{\timeseq{\min(\alpha,\xi)}{\max(\gamma,\chi)}}\right), \label{eq:proof_of_theorem}
	\end{align}
	is a \pmbm density, where no update (only prediction) is performed for time steps outside the interval $\timeseq{\xi}{\chi}$. Applying Theorem~\ref{thm:all_to_trajs_in_a2g_alive_in_e2z} to \eqref{eq:proof_of_theorem} concludes the proof.

\section{Additional details on experimental results}\label{sec:experiments_appendix}
In this appendix, we provide further details on the experimental results. First, we show the number of targets at each time step in each of the scenarios in Figure \ref{fig_groundtruth_card}. Second, we show the \trajmetric for the estimated trajectories at the final time step in Figure \ref{fig_traj_metric} and Table \ref{tab:traj_metric}. In this case, the \multiscanpmbmif is also the best performing algorithm in Scenario 1, and \pmbmif in Scenarios 2 and 3.

\begin{figure*}[!t]
	\centering
	\subfloat[Scenario 1]{\includegraphics[width = 0.33\textwidth]{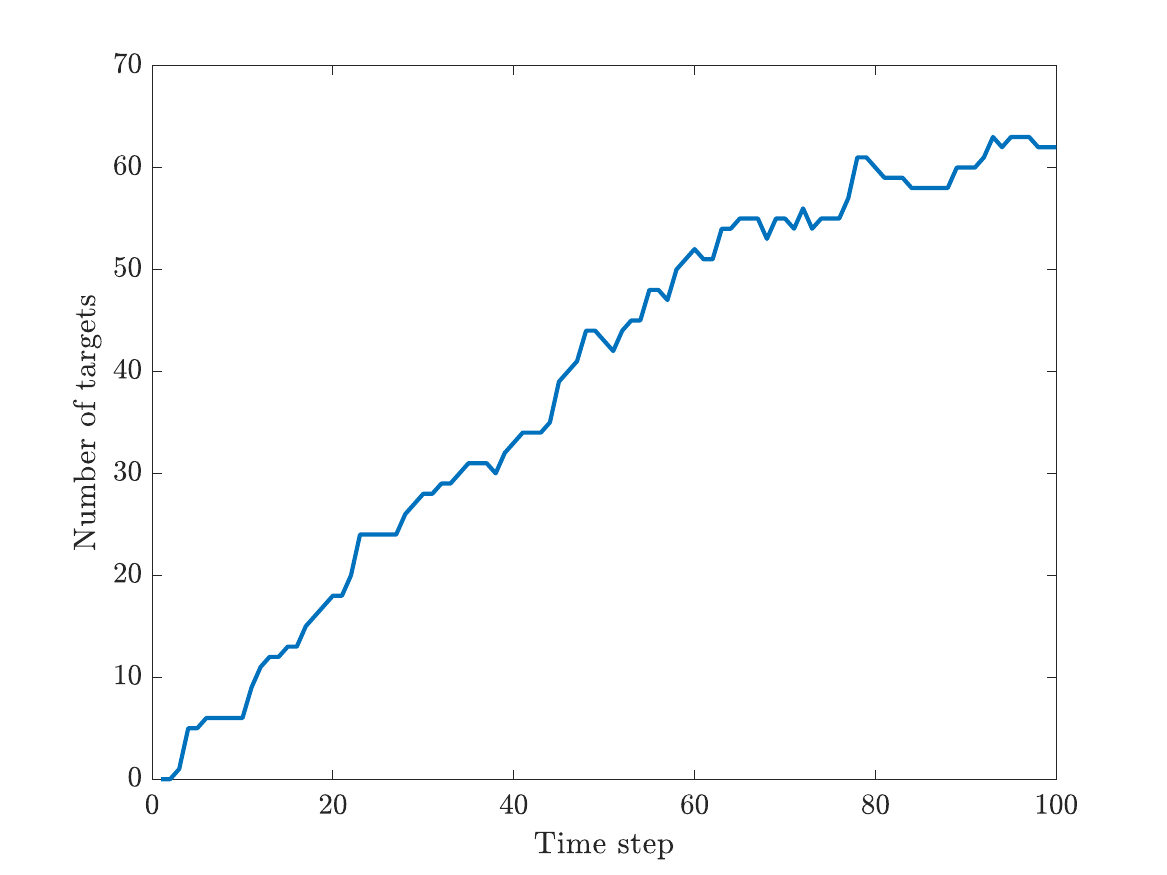}}
	\subfloat[Scenario 2]{\includegraphics[width = 0.33\textwidth]{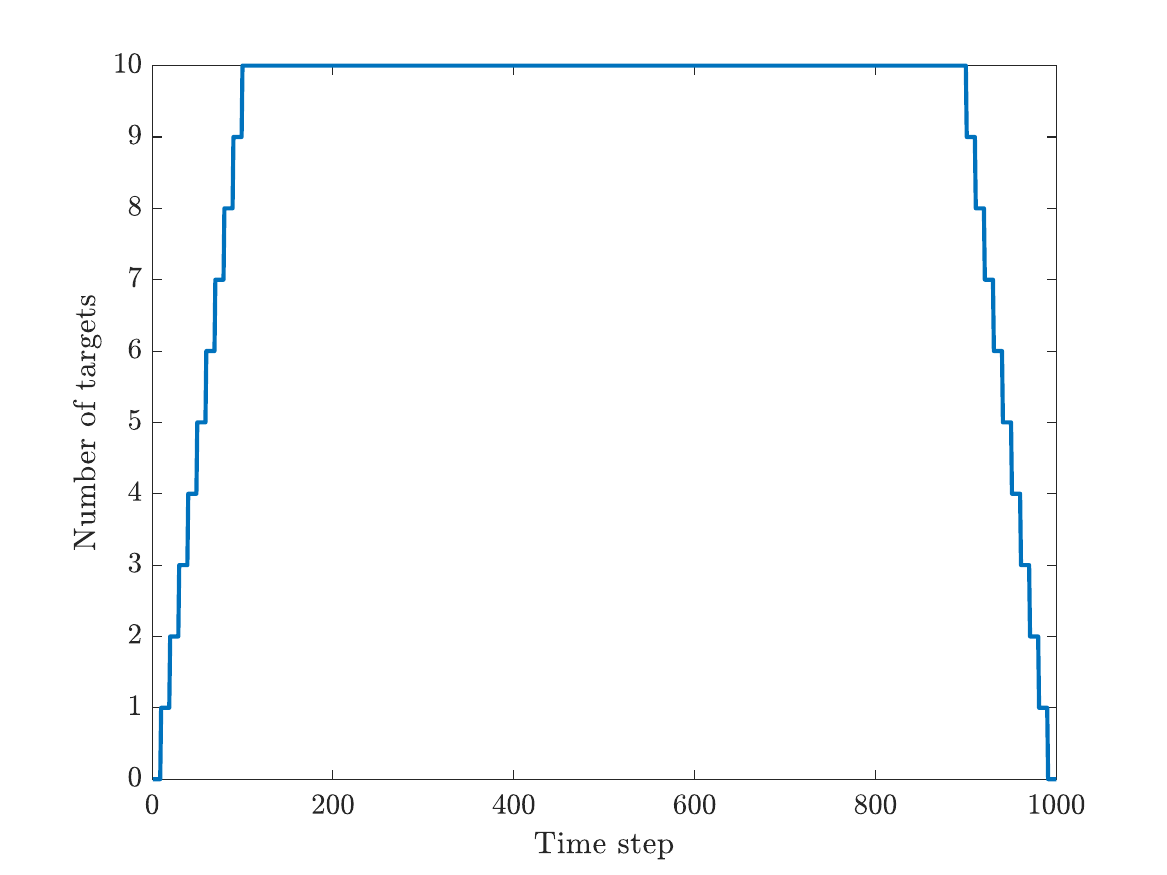}}
	\subfloat[Scenario 3]{\includegraphics[width = 0.33\textwidth]{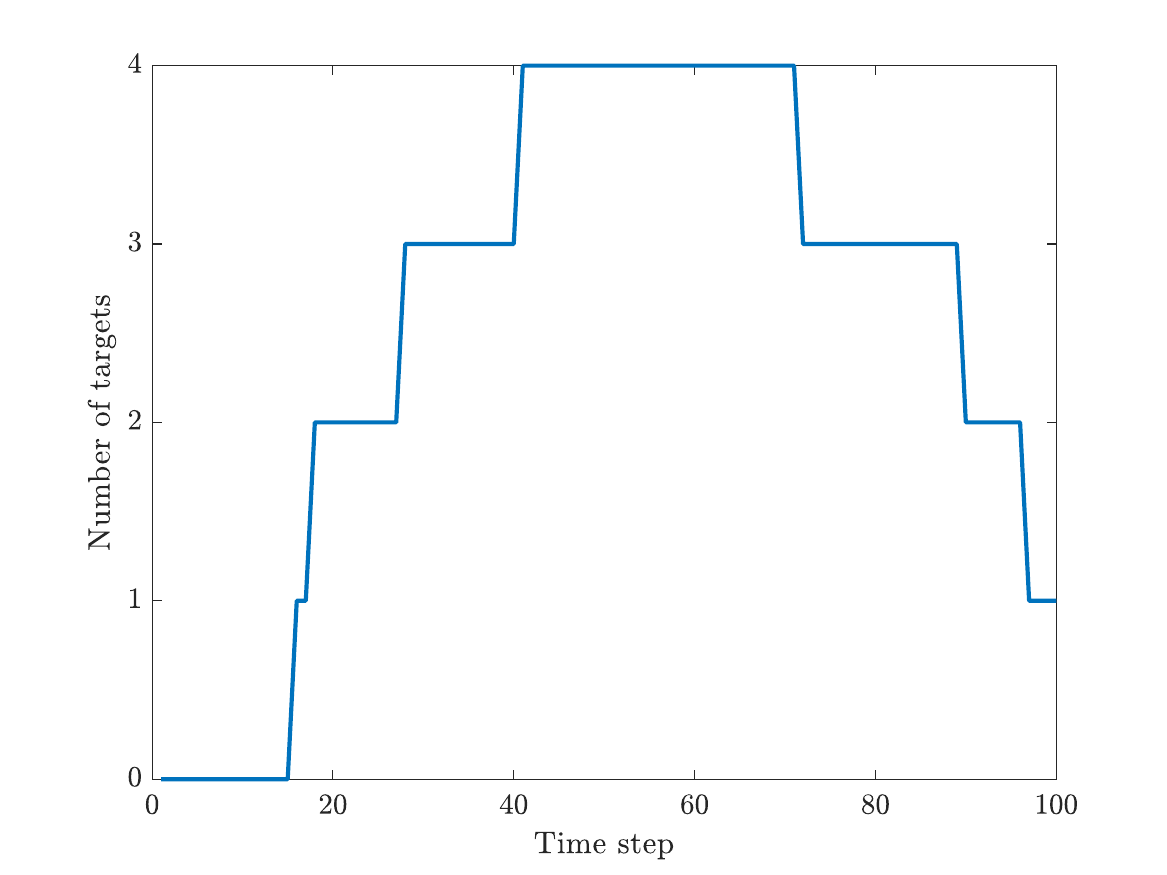}}
	\caption{Target cardinalities against time in each of the three scenarios in the simulations.}
	\label{fig_groundtruth_card}
\end{figure*}

\begin{figure*}[!t]
	\centering
	\subfloat[Scenario 1]{\includegraphics[width = 0.33\textwidth]{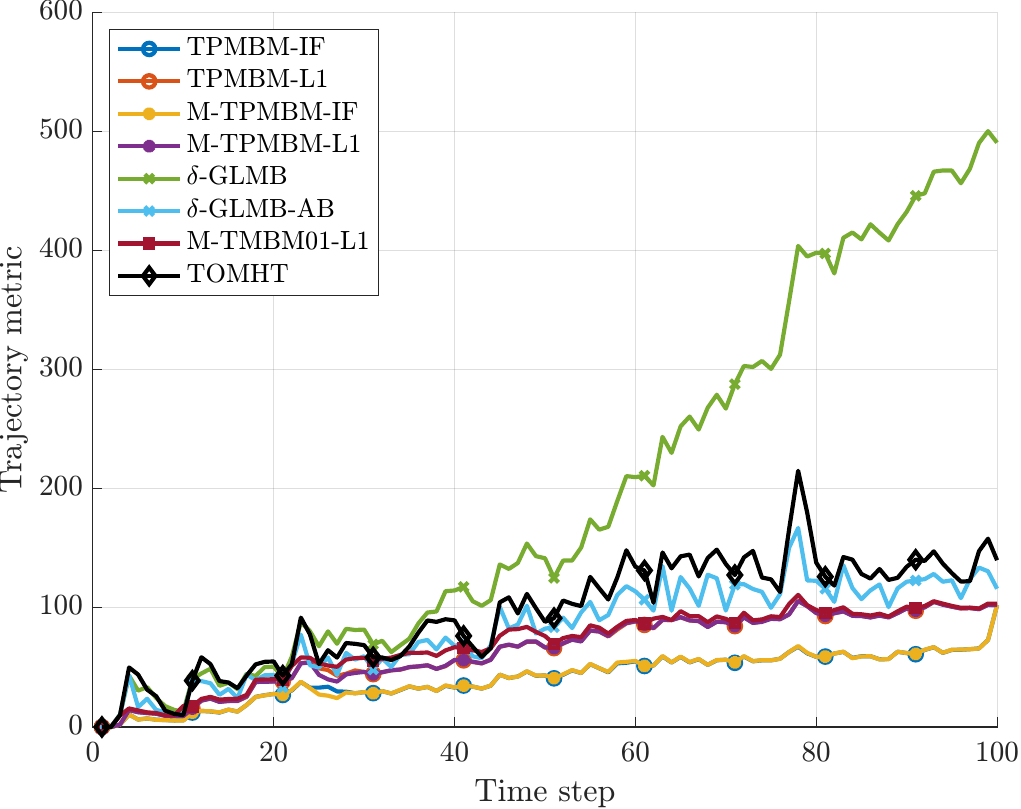}}
	\subfloat[Scenario 2]{\includegraphics[width = 0.33\textwidth]{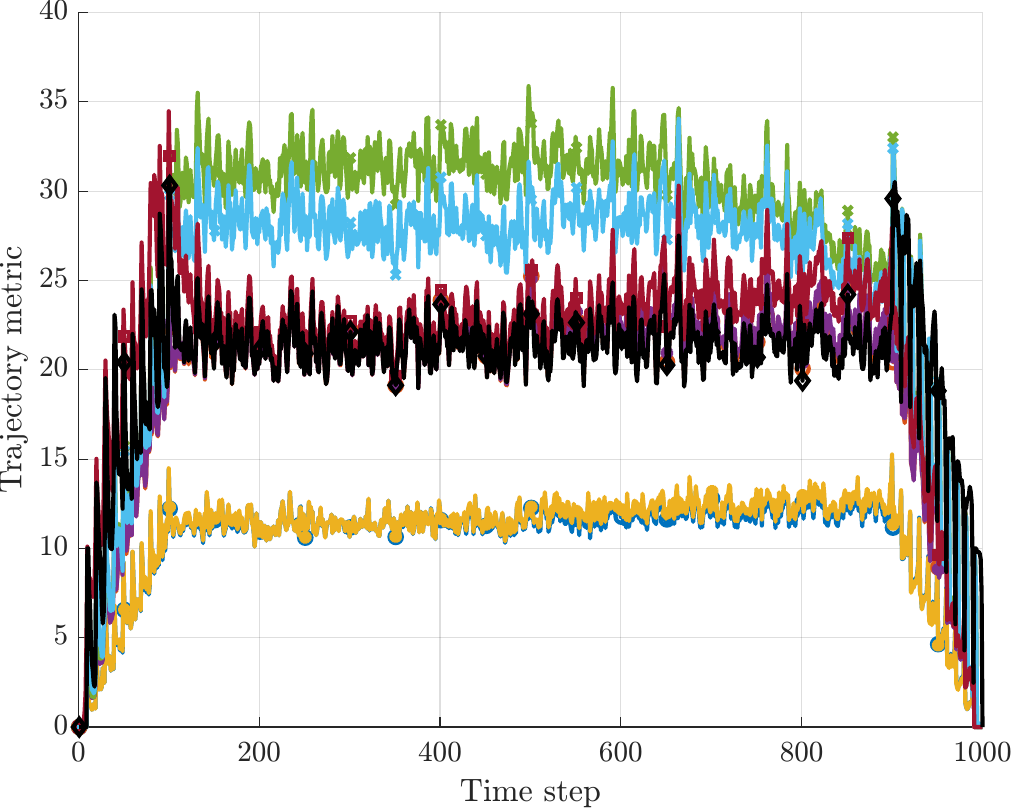}}
	\subfloat[Scenario 3]{\includegraphics[width = 0.33\textwidth]{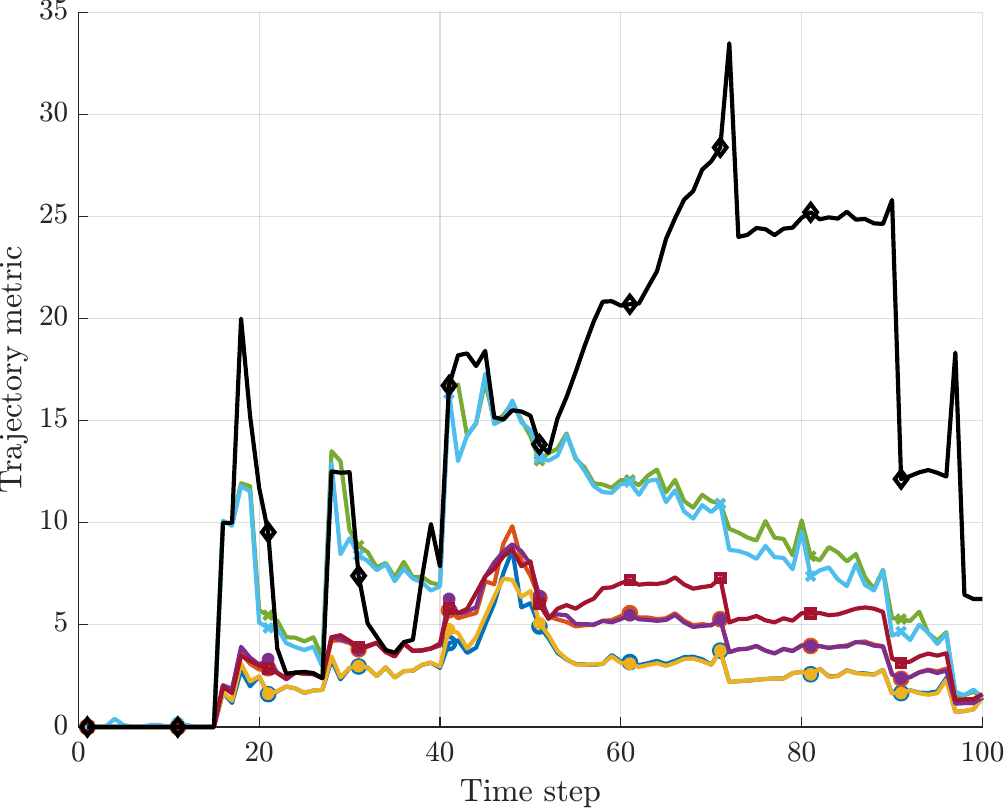}}
	\caption{Trajectory estimation performance (evaluated at the final time step) in terms of the normalised \trajmetric over time}
	\label{fig_traj_metric}
\end{figure*}

\begin{table*}[!t]
	\caption{Trajectory estimation performance in terms of the normalised \trajmetric (evaluated at the final time step) and its decomposition}
	\label{tab:traj_metric}
	\centering
	\resizebox{\textwidth}{!}
	{
		\begin{tabular}{c|ccccc|ccccc|ccccc}
			& \multicolumn{5}{c|}{Scenario 1} & \multicolumn{5}{c|}{Scenario 2} & \multicolumn{5}{c}{Scenario 3} \\ \hline
			Filter & TM   & LE   & ME   & FE   & SE  & TM   & LE   & ME   & FE   & SE  & TM   & LE   & ME   & FE  & SE  \\ \hline
			\pmbmif	 &   4170.1   &   3884.8   &   273.4   &   10.5   &  1.5   &    \underline{10532.8}  &   10182.7   &   349.7   &   0   &  0.4   &   \underline{247.3}   &  235.5    &   4.3   &  0   &  7.5   \\
			\pmbmLone	   &   6487.2   &   6116.4   &  316.1    &   53.2   &  1.5   &  19584.2    &  18918.5    &   507.5   &   157.8   &  0.4   &   357.6   &  344.5    &   4.3   &  0   &  8.8   \\
			\multiscanpmbmif	   &   \underline{4158.4}   &   3885.7   &   262.7   &   9.3   &  0.7   &   10692.7   &   10165.3   &   527.0   &   0   &  0.4   &   250.1   &   235.5   &  6.5    &  0   &  8.0   \\
			\multiscanpmbmLone	   &   6444.5   &   6112.8   &   272.0   &   58.0   &  1.7   &    19698.8  &  18859.0    &   683.2   &   156.2   &  0.4   &   361.6   &   345.9   &   6.6   &  0.1   &  9.1   \\
			\dglmb	   &   19590.7   &   5157.7   &   6433.3   &   8134.9   &   78.3  &   26998.5   &   16496.8   &   10356.0   &   282.5   &  1.1   &   760.8   &  325.1    &   438.4   &  19.5   &  20.2   \\
			\dglmbab	   &   8175.5   &   5855.7   &   1641.7   &   718.7   &   13.3  &   24696.2   &   17269.2   &   7337.5   &   364.0   &  8.7   &   729.0   &   325.2   &   393.8   &  24.8   &  21.7   \\
			\multiscanmbmolLone	   &   6939.8   &   6033.3   &   706.7   &   197.7   &  2.1   &   21628.8   &   18608.0   &   1934.2   &   1085.0   &  1.6   &   422.6   &   344.8   &   38.9 & 30.2   &  8.7   \\
			\tomht & 9606.2 & 5650.2 & 2730.8 & 1223.1 & 2.1 & 20518.0 & 18967.3 & 587.4 & 963.1 & 0.1 & 1390.1 & 507.8 & 542.2 & 325.4 & 15.4
		\end{tabular}
	}
\end{table*}

\section{GOSPA metric for sets of trajectories}\label{sec:Trajectory_metric}
\label{app:Trajectory_metric}
This appendix reviews the GOSPA metric for sets of trajectories based on linear programming \cite[Prop. 2]{RahmathullahGS16a}, which we abbreviate as \trajmetric in this paper. We consider two sets of trajectories $\mathbf{X}$ and $\mathbf{Y}$ from time step 0 to $T$. The number of trajectories is these sets is $n_{\mathbf{X}}$ and $n_{\mathbf{Y}}$, respectively. For $1\leq p < \infty$, $c > 0$, $\gamma >0$ and a base metric $d_{b}(\cdot,\cdot)$ for single targets, the \trajmetric is 
\begin{align}
d\left(\mathbf{X},\mathbf{Y}\right) & =\min_{\substack{W^{k}\in\mathcal{\overline{W}}_{\mathbf{X},\mathbf{Y}}\\
		k=1,\ldots,T
	}
}\Bigg(\sum_{k=1}^{T}\mathrm{tr}\big[\big(D_{\mathbf{X},\mathbf{Y}}^{k}\big)^{\dagger}W^{k}\big]\nonumber \\
& \quad+\frac{\gamma^{p}}{2}\sum_{k=1}^{T-1}\sum_{i=1}^{n_{\mathbf{X}}}\sum_{j=1}^{n_{\mathbf{Y}}}|W^{k}(i,j)-W^{k+1}(i,j)|\Bigg)^{\frac{1}{p}},\label{eq:LP_metric}
\end{align}
where $D_{\mathbf{X},\mathbf{Y}}^{k}$ is a $(n_{\mathbf{X}}+1)\times(n_{\mathbf{Y}}+1)$
matrix whose $(i,j)$ element is
\begin{align}
D_{\mathbf{X},\mathbf{Y}}^{k}(i,j) & =d_{G}\left(\mathbf{x}_{i}^{k},\mathbf{y}_{j}^{k}\right)^{p}.
\end{align}
with $\mathbf{x}_{i}^{k}$ being the target set of the $i$-th trajectory in $\mathbf{X}$ at time step $k$ and  $\mathbf{y}_{j}^{k}$ being the target set of the $j$-th trajectory in $\mathbf{Y}$ at time step $k$, and
\begin{align}
d_{G}\left(\mathbf{x},\mathbf{y}\right) & \triangleq\begin{cases}
\min\left(c,d_{b}\left(x,y\right)\right) & \mathbf{x}=\left\{ x\right\} ,\mathbf{y}=\left\{ y\right\} \\
0 & \mathbf{x}=\mathbf{y}=\emptyset\\
\frac{c}{2^{1/p}} & \mathrm{otherwise.}
\end{cases}\label{eq:baseMetric}
\end{align}
A matrix $W^{k}\in\mathcal{\overline{W}}_{\mathbf{X},\mathbf{Y}}$ whose $(i,j)$ element is $W^{k}(i,j)$ meets
\begin{align}
\sum_{i=1}^{n_{\mathbf{X}}+1}W^{k}(i,j) & =1,\ j=1,\ldots,n_{\mathbf{Y}}\\
\sum_{j=1}^{n_{\mathbf{Y}}+1}W^{k}(i,j) & =1,\ i=1,\ldots,n_{\mathbf{X}}\\
W^{k}(n_{\mathbf{X}}+1,n_{\mathbf{Y}}+1) & =0,\\
W^{k}(i,j) & \geq0,\forall\ i,j.
\end{align}
Detailed explanations of these terms are provided in \cite{RahmathullahGS16a}. The \trajmetric normalised by the time window is then $d\left(\mathbf{X},\mathbf{Y}\right)/T$.

\end{document}